\documentclass[twocolumn,prb,superscriptaddress,showpacs]{revtex4}
\usepackage{amsmath,amssymb}
\usepackage{graphicx}
\usepackage{verbatim}
\usepackage{amsfonts}
\usepackage{subfigure}
\usepackage{dcolumn}
\usepackage{float}
\usepackage{bm}
\usepackage{color}
\usepackage[colorlinks,citecolor=blue]{hyperref}
\usepackage{framed} 
\usepackage{mathrsfs}
\usepackage{ntheorem}
\usepackage{tikz-cd}
\usetikzlibrary{cd}
\newcommand{\beq}{\begin{eqnarray} }
\newcommand{\eeq}{\end{eqnarray} }
\newcommand{\Beq}{\begin{eqnarray*} }
\newcommand{\Eeq}{\end{eqnarray*} }
\newcommand{\Bmat}{\left(\begin{matrix}}
\newcommand{\Emat}{\end{matrix}\right)}
\newcommand{\up}{\uparrow}
\newcommand{\dn}{\downarrow}
\newcommand{\bit}{\begin{itemize} }
\newcommand{\eit}{\end{itemize} }
\newcommand{\ben}{\begin{enumerate} }
\newcommand{\een}{\end{enumerate} }
\newcommand{\veps}{\varepsilon}

\usepackage{verbatim}
\usepackage{hyperref}
\usepackage{amsfonts}
\usepackage{dcolumn}
\usepackage{bm}
\usepackage{natbib}
\usepackage{mathrsfs}
\usepackage{empheq}
\usepackage{multirow}
\usepackage[normalem]{ulem}
\useunder{\uline}{\ul}{}

\newcommand{\mi}{\mathrm{i}}
\newtheorem{theorem}{Theorem}

\newtheorem{lemma}{Lemma}

\newtheorem*{proof}{Proof}

\makeatletter
\newsavebox{\@brx}
\newcommand{\llangle}[1][]{\savebox{\@brx}{\(\m@th{#1\langle}\)}%
  \mathopen{\copy\@brx\kern-0.5\wd\@brx\usebox{\@brx}}}
\newcommand{\rrangle}[1][]{\savebox{\@brx}{\(\m@th{#1\rangle}\)}%
  \mathclose{\copy\@brx\kern-0.5\wd\@brx\usebox{\@brx}}}
\makeatother

\usepackage{longtable}
\usepackage{multirow}

\newcommand{\RNum}[1]{\uppercase\expandafter{\romannumeral #1\relax}}

\usepackage{ulem}

\newcommand{\B}{\color{blue}}

\newcommand{\C}{\mathcal{C}}
\newcommand{\cP}{\mathcal P}
\newcommand{\T}{\mathcal T}
\newcommand{\m}{D'}
\newcommand{\tp}{\tilde{\mathcal P}}
\newcommand{\p}{\tilde{\mathcal P}}

\graphicspath{{figs/}}

\begin{document}

\title{\large \bf Representation Theory for Massless Quasiparticles in Bogoliubov-de Gennes Systems}

\author{Arist Zhenyuan Yang}
\affiliation{Department of Physics and Beijing Key Laboratory of Opto-electronic Functional Materials and Micro-nano Devices, Renmin University of China, Beijing, 100872, China}
\affiliation{Key Laboratory of Quantum State Construction and Manipulation (Ministry of Education), Renmin University of China, Beijing, 100872, China}
\author{Zheng-Xin Liu}
\email{liuzxphys@ruc.edu.cn}
\affiliation{Department of Physics and Beijing Key Laboratory of Opto-electronic Functional Materials and Micro-nano Devices, Renmin University of China, Beijing, 100872, China}
\affiliation{Key Laboratory of Quantum State Construction and Manipulation (Ministry of Education), Renmin University of China, Beijing, 100872, China}

\date{\today}

\begin{abstract}

Linearly dispersive gapless quasiparticles can appear at general momentum points or on high symmetry momentum lines of superconductors due to topological reasons such as K theory or symmetry indicators theory. However, the zero modes associated with these quasiparticles are generally `accidental’ from symmetry point of view. In this work, we apply projective representation (rep) theory to analyze the bulk gapless quasiparticles in BdG systems. Different from the description of semimetals, the particle-hole `symmetry’ requires special treatment since it anti-commutes with the BdG Hamiltonian. Accordingly the notion of `simple irreducible reps (irreps)’ and `composite irreps’ are introduced to label the energy modes.  
Without charge conjugation symmetry (unitary symmetry that commutes with the Hamiltonian), no robust zero modes exist at any fixed momentum point in the bulk. However, robust zero modes at certain high-symmetry momentum points can be protected by (effective) charge conjugation symmetries, resulting in gapless quasiparticles with linear, quadratic, or higher-order dispersions. The 
low-energy properties of the quasiparticles,  including the dispersions and responses to external probe fields, are dictated by the reps carried by these 
zero modes
from the effective \( k \cdot p \) theory. Our theory provides a framework to classify nodal superconductors/superfluids/quantum spin liquids with specified (projective) symmetry group, and sheds light on the realization of Majorana-type massless quasiparticles in condensed matter physics.

\end{abstract}

\maketitle

\section{Introduction}\label{sec:intro}

In condensed matter physics, quasiparticle excitations emergent in the long-wavelength limit behave like elementary particles in high-energy physics. For instance, massless Dirac fermions or Weyl fermions can appear in semimetals at specific points of the first Brillouin zone (BZ) if the crystal has certain space group symmetries \cite{WanTurVis11Topologica, FanGilDai12Multi-Weyl, WanWenWu13Three-dime, YanNag14Classifica, WenFanFan15Weyl-semim, RuaJiaYao16Symmetry-p, WuYuZho20Higher-ord}. These gapless quasiparticles lead to observable physical phenomena, such as negative magnetic resistance and Fermi arc spectra in surface excitations \cite{NohHuaLey17Experiment,GuoYouYan19Observatio,YanXiaSun19Observatio}. Furthermore, new types of quasiparticles without counterparts in high-energy physics can exist in lattice systems \cite{ZyuBur12Topologica,HosQi13Recent-dev,XuZhaZha15Structured}. These gapless quasiparticles are characterized by multiple degeneracies in the energy spectrum, carrying projective representations of their symmetry groups \cite{YanLiu17Irreducibl}. At high-symmetry points in the BZ, these degeneracies are ensured by nontrivial irreducible (projective) representations of the little co-group, forming massless quasiparticles \cite{BraDav68Magnetic-g, BraCra10The-mathem, Bir12Theory-of-, Ham12Group-theo}. For example, at the $K$ and $K'$ points in graphene \cite{graphene}, two-fold degeneracy is protected by a two-dimensional (2-D) representation of the group $\mathscr{C}_{3v}$, resulting in Dirac-like quasiparticles with linear dispersion. Gapless quasiparticles can also arise from symmetry-protected level crossings of two bands along certain symmetric lines in the BZ, where they carry reducible (projective) representations of the little co-group \cite{YouKan15Dirac-semi, WenFanFan16Topologica, BouBla17Global-ban, ChaXuWie17Unconventi, YuWuZha19Circumvent, XiaYeQiu20Experiment}. A comprehensive description and classification of quasiparticles require both the projective representations of symmetry groups and the symmetry invariants that label the classes of these representations \cite{CheGuLiu13Symmetry-p,yang2020unlocking, YanFanLiu21Symmetry-p, CheZhaYan23Classifica}. With the use of symmetry invariants, many new types of quasiparticles have been discovered in magnetically ordered systems with weak spin-orbit coupling \cite{WuTanWan20Exhaustive,YanLiuFan21Symmetry-i,GuoPenLiu22Symmetry-e,TanWan22Complete-c}, whose symmetry groups are known as spin space groups \cite{SSG1,SSG2,SSG3}. Besides isolated momentum points, nodal line and surface structures can form if multiple degeneracies are ensured along a line or surface in the BZ \cite{Vol93Supercondu,BurHooBal11Topologica, PhiAji14Tunable-li, FanCheKee15Topologica, KimWieKan15Dirac-line, MulUchGla15Line-of-Di, BzdWuRue16Nodal-chai, Eza16Loop-nodal, ZhaYuWen16Topologica, ChaYee17Weyl-link-, YanBiShe17Nodal-link}. The dispersion of quasiparticles around these nodal points or lines is described by the effective Hamiltonian known as the $k\cdot p$ model \cite{Kan66The-k-p-Me, ChuCha96k-p-method,Gre18Identifyin, JiaFanFan21A-k-p-effe, YanYanFan21A-Hamilton}. This Hamiltonian, along with the response theory of quasiparticles to external probe fields such as electric fields, magnetic fields, and strain, can be derived from the projective representations associated with the quasiparticles.

The quasiparticles in semimetals discussed above have $U(1)$ charge (or particle number) conservation symmetry. Another class of quasiparticles are of Bogoliubov type, where fermions pair to form Cooper pairs, breaking the $U(1)$ symmetry down to $Z_2$\cite{RMPsc1997}. In the BdG Hamiltonian of SCs/SFs or QSLs, fermion pairing introduces a particle-hole symmetry $\mathcal{P}$, ensuring that the energy spectrum is symmetric around zero energy. Gapless Bogoliubov quasiparticles near the Fermi level (zero energy) are particularly interesting as they determine the system's low-energy physical responses. Due to $\mathcal{P}$, the degenerate points of gapless quasiparticles are exactly at zero energy, termed zero modes. For instance, $d$-wave SCs on a square lattice contain zero modes at four points along the diagonal lines of the BZ, around which linearly dispersive quasiparticles arise. Notably, if charge conjugation symmetry $\mathcal{C}$ is preserved, Bogoliubov quasiparticles become Majorana types, which are their own antiparticles. The `antiparticle' and particle are related by charge conjugation symmetry $\mathcal{C}$ (or `particle-antiparticle exchange symmetry') \cite{ColmanBook}. Unlike in high-energy physics, where Majorana fermions are massive, in condensed matter physics, Majorana excitations at high-symmetry points of the BZ can be massless (gapless) in $\mathcal{C}$-symmetric SCs or SFs. Although Majorana zero modes (gapless chiral Majorana edge states) can appear on the boundary of topological SCs/SFs\cite{Topofirstprinciple2023}, in the present work, we focus on the bulk spectrum.

Previous studies of bulk zero modes have mainly focused on topological aspects, such as $\mathbb{Z}_2$ or $\mathbb{Z}$ invariants or symmetry indicators \cite{Slager2013,Slager2017,PRL2017Timm-pf,Po2017Sym-BasedIndicators,Bradlyn2017TopoQuanChemistry,Ono2018SIs,sumita2019classification,PRB2020_indicator,PRX2022-AZclass}. A comprehensive theory of gapless quasiparticles in BdG systems, including their physical response to external probe fields, is still lacking. In this work, starting from symmetry groups and their representation theory, we systematically study the mechanisms for the appearance of zero modes in the bulk spectrum of SCs/SFs or QSLs, and provide methods to determine the physical properties of gapless quasiparticles. We prove that without (effective) charge conjugation symmetry, zero modes at a given momentum point can be adiabatically removed. With (effective) charge conjugation symmetry, zero modes can stably exist at high-symmetry points in a band representation, forming gapless quasiparticles with linear, quadratic, or higher-order dispersions. We define two types of zero modes—irreducible and reducible—based on the projective representations of the little co-group. This allows us to determine the degeneracy of zero modes and systematically construct the $k\cdot p$ model and response matrix of gapless quasiparticles to probe fields. Finally, we propose a classification scheme for point-nodal SCs/SFs/QSLs for a given symmetry group. Our conclusions apply to all Altland–Zirnbauer symmetry classes of BdG systems.\cite{Zir96Riemannian, AltZir97Nonstandar}. Furthermore, the application of representations differs from conventional symmetry groups due to $\mathcal{P}$ anti-commuting with the Hamiltonian.

The rest of the paper is organized as follows: In Section \ref{sec:zeromodes}, after a brief introduction to symmetry groups for fermionic BdG Hamiltonians, we derive the conditions for the appearance of zero modes in the bulk spectrum of SCs/SFs or QSLs. In Section \ref{sec:kdotpmodel}, we provide an efficient method for obtaining the effective $k\cdot p$ model and the physical response to external probe fields for gapless quasiparticles. Concrete lattice models are provided in Section \ref{Sec:ModelSection} to illustrate the results from previous sections. Section \ref{sec:conc} discusses the classification scheme for point-node SCs/SFs/QSLs and summarizes the main findings.

\section{Zero modes in the Bulk spectrum of general $\text{BdG}$ systems} \label{sec:zeromodes}

In metals and insulators, the symmetry group is commuting with the Hamiltonian. 
However, BdG systems like SCs/SFs are characterized by their particle-hole symmetry ${\mathcal P}$ which is anti-commuting with the Hamiltonian $\mathcal P\mathscr H=-\mathscr H\mathcal P$ and maps the eigenspace of $\veps$ to that of $-\veps$, 
$
\mathscr H\mathcal P|\veps\rangle = -\mathcal P\mathscr H|\veps\rangle = -\veps\mathcal P|\veps\rangle,
$
namely $\mathcal P|\veps\rangle \propto |-\veps\rangle$. Therefore, generally we need to treat the direct sum of the eigenspaces of $\{\veps, -\veps\}$.

\subsection{Symmetry operations}

At the mean field level, the Hamiltonian of a SC/SF/QSL reads
\beq \label{HSC}
H = \sum_{i\neq j} \Big( t_{ij} c_i^\dag c_j + \Delta_{ij} c_i^\dag c_j^\dag + {\rm h.c}\Big) + \sum_i \lambda_i c_i^\dag c_i,
\eeq
which can be written in matrix form in the Nambu bases $\Psi^\dag = [C^\dag, C^T]$ with $C^\dag=[c_1^\dag, c_2^\dag, ..., c_N^\dag]$ and $C^T = [c_1, c_2, ..., c_N]$ ($N$ is the total number of the degrees of freedom), namely 
$$H = \Psi^\dag \mathscr H \Psi + h_0,$$ 
with $h_0 = {1\over2} \sum_i \lambda_i$, $\mathscr H^\dag =\mathscr H$ and ${\rm Tr}\ \mathscr H=0$. In the following, we will briefly introduce the symmetry operations of the Hamiltonian (\ref{HSC}) and categorize them into three distinct types.

(1) {\it Space-time symmetry} described by space groups or magnetic space groups. \\
If $g$ is a spatial symmetry operation, then $$\hat g C^\dag \hat g^{-1} = C^\dag u(g),$$ namely, $\hat g c_i^\dag \hat g^{-1} = \sum_{j=1}^N u_{ji}(g) c_j^\dag$, where operations with $\hat{} $ acts on the Hilbert space. The Hamiltonian is invariant under the action of $g$,
$
\hat g H \hat g^{-1} =H, 
$
or equivalently $U(g)\mathscr H U^{-1}(g) = \mathscr H$ with $$U(g)=\Bmat u(g)&\\& u^*(g)\Emat.$$
For a SC without magnetic order, the symmetry group is a direct product of space group and the time reversal group $Z_2^{\T}=\{E, \T\}$, namely the type-II magnetic space group. Here $\T$ acts as 
$\hat \T C^\dag \hat \T^{-1} = C^\dag u(\T) K$
and $U(\T) K\mathscr H KU^{-1}(\T)= U(\T) \mathscr H^* U^{-1}(\T) = \mathscr H$ with $$U(\T)K =\Bmat u(\T)&\\& u^*(\T)\Emat K.$$ If the system contains magnetic order, then the symmetry group is generally a magnetic space group of type-I, type-III or type-IV. 

(2) Internal symmetries in the spin and charge sectors. \\
For spin-1/2 fermions, the set of spin operations form a group $SU(2)_s$, and the set of charge operations form another group $SU(2)_c$. The action of the two $SU(2)$ groups can be clearly seen by performing an unitary transformation to the complete Nambu bases $\Psi^\dag =[C_\up^\dag, C_\dn^\dag, C_\up^T, C_\dn^T]$ into $\bar\Psi^\dag = [C_\up^\dag, C_\dn^\dag, C_\dn^T, -C_\up^T]$ where $C_{\sigma}^\dag =[c_{1\sigma}^\dag, ..., c_{L\sigma}^\dag]$ with $\sigma=\up,\dn$ and $L$ the number of sites. 

In the new bases $\bar\Psi^\dag$, a spin rotation in the $SU(2)_{\rm s}$ group acts as
\beq\label{SU2s}
\bar\Psi^\dag \to \bar\Psi^\dag I_2\otimes \exp\{-\mi{\pmb \sigma\over 2}\cdot \pmb n\theta\}\otimes I_{L},
\eeq
where the three Pauli matrices $\sigma_{x,y,z}$ are the generators of the spin transformations and $\pmb n$ is the axis of the spin rotation and $\theta\in[0,2\pi]$ is the rotation angle.  If the system is free of spin-orbit coupling(SOC), then the symmetry group of the system is then a direct product group of a space group (or a magnetic space group) and $SU(2)_{\rm s}$. For systems with non-negligible SOC, the spatial point group operations also act on the spin degrees of freedom.  The resulting space groups are known as double groups. 

On the other hand, a charge operation in the $SU(2)_{\rm c}$ group transforms as
\beq\label{SU2c}
\bar\Psi^\dag \to \bar\Psi^\dag \exp\{-\mi{\pmb \tau\over 2}\cdot \pmb m\varphi\} \otimes I_{2L},
\eeq
where the three Pauli matrices $\tau_{x,y,z}$ are the generators of charge transformations, and $\pmb m$ is the axis of the charge rotation and $\varphi\in[0,2\pi]$ is the rotation angle. In electronic systems the $SU(2)_{\rm c}$ charge group generally breaks down to its subgroup. 

Notice that the $SU(2)_s$ and $SU(2)_c$ groups are not completely independent, they share the same center $Z_2^f=\{E,P_f\}$ which is generated by the fermion parity $P_f=e^{-\mi2\pi {{\pmb \tau}\over 2}\cdot \pmb m} \otimes I_{2}=I_2\otimes e^{-\mi2\pi {{\pmb \sigma}\over 2}\cdot \pmb n}$. As a consequence, the internal symmetry operations form a $SO(4)$ group.

In the present work we will adopt the original Nambu bases $\Psi^\dag =[C_\up^\dag, C_\dn^\dag, C_\up^T, C_\dn^T]$, where a spin $SU(2)_{\rm s}$ rotation (\ref{SU2s}) is given by
\beq\label{SU2s0}
\Psi^\dag \to \Psi^\dag \exp\left(-\mi{ \tilde{\bm\sigma}\over 2}\cdot \bm{n}\theta\right)\otimes I_{L},
\eeq
with the three generators
$$
\tilde\sigma_x =\tau_z\otimes\sigma_x,\ \tilde\sigma_y =I_2\otimes\sigma_y,\ \tilde\sigma_z =\tau_z\otimes\sigma_z.
$$
A charge $SU(2)_{\rm c}$ transformation (\ref{SU2c}) then takes the following form
\beq\label{SU2c0}
\Psi^\dag \to \Psi^\dag \exp\{-\mi{\tilde{\pmb \tau}\over 2}\cdot \pmb m\varphi\} \otimes I_{L},
\eeq
with
$$
\tilde\tau_x=-\tau_y\otimes\sigma_y, \ \tilde\tau_y=\tau_x\otimes\sigma_y,\ \tilde\tau_z=\tau_z\otimes I_2.
$$

An important charge operation is the {\it charge conjugation} $\C$ that exchanges the creation operators with the annihilation operators, such as $\bar\C=e^{-\mi{ {\tau}_x\over2}\pi}$ or $\bar\C=e^{-\mi{ {\tau}_y\over2}\pi}$ in the new bases $\bar\Psi^\dag$. Generally a charge conjugation is of order 4 with $\hat \C^2=\hat P_f$. In the original Nambu bases $\Psi^\dag$ the charge conjugation takes the form of $\C=e^{-\mi{ \tilde{\tau}_x\over2}\pi}$ or $\C=e^{-\mi{ \tilde{\tau}_y\over2}\pi}$. A Hamiltonian is said to have charge conjugation symmetry if it is invariant under the charge conjugation transformation, namely
$
\hat {\mathcal C}H \hat{\mathcal C}^{-1} = \Psi^\dag  \C \mathscr H  \C^\dag \Psi =H,
$
which requires
\beq\label{CCS}
 \C \mathscr H  \C ^\dag =\mathscr H. 
\eeq
Since the charge conjugation exchanges creation operators with annihilation operators, to preserve the $\C$ symmetry we must have $\lambda_i=0$ for all sites. Therefore, charge conjugation symmetry is a very stringent constraint in condensed matter systems since it requires the chemical potential to be zero everywhere. If a charge conjugation is combined with operations in the spin and/or lattice sectors to form a symmetry operation, then we call it an effective charge conjugation symmetry.

For fermions carrying integer spin (such as the spinless fermions discussed in section \ref{sec:spinless}) the charge group is simply $O(2)\cong U(1)\rtimes Z_2^{\C}$, where $U(1)=\{e^{-\mi\tau_z\theta}; \theta\in[0,2\pi)\}$ and $Z_2^{\C}=\{E,\tau_x\}$ is generated by the charge conjugation $\C=\tau_x$ \cite{LiuZhouNg10}.

(3) The {\it particle-hole `symmetry'}. \\
Notice that in the complete Nambu bases one has $\Psi^\dag=\Psi^T \Gamma_x$ with $\Gamma_x = \tau_x \otimes I_N$. Furthermore, the transpose of a fermion bilinear Hamiltonian gives rise to a minus sign (if $h_0=0$), so any BdG system has a particle-hole `symmetry'
\beq\label{PHS}
\Gamma_x K \mathscr H K \Gamma_x =- \mathscr H.
\eeq
Owing to (\ref{PHS}), the energy spectrum has a reflection symmetry centered at $E=0$ and the eigenvalues $\pm \varepsilon$ appear in pairs. The particle-hole `symmetry' operation $\hat{\mathcal P} = \Gamma_x K$ in (\ref{PHS}) is {\it anti-unitary} in the first quantization formalism. Since $\hat\cP^2=1$, the corresponding symmetry class is called the D-class. 

For spin-1/2 systems with spin rotational symmetry,  we adopt the reduced Nambu bases $\psi^\dag =[C_\up^\dag, C_\dn^T]$ in which the particle-hole symmetry is redefined as $\hat {\mathcal P} =\mi\Gamma_yK$ with $\Gamma_y = \mi\tau_y \otimes I_L$ and $\hat {\mathcal P}^2=-1$. This symmetry class is called the C-class. In the reduced Nambu bases, the $SU(2)_s$ group acts trivially and the $SU(2)_c$ operations act as 
\beq\label{SU2cr}
\psi^\dag \to \psi^\dag \exp\{-\mi{\pmb \tau\over 2}\cdot \pmb m\varphi\} \otimes I_{L}.
\eeq

Since $\hat{\mathcal P}$ is anti-unitary, it transforms momentum $\pmb k$ to $-\pmb k$. Hence in momentum space, $\hat{\mathcal P}$ is not necessarily a symmetry operation. However, the combination of $\hat{\mathcal P}$ and spatial inversion (or time reversal, et al) operation form an effective particle-hole symmetry at momentum point  $\pmb k$ which anti-commute with the Hamiltonian $\mathscr H_k$. In late discussion, we denote the effective particle-hole symmetry as $\p$ with $\p=\mathcal I\hat{\mathcal P}$ (or $\p=\mathcal T\hat{\mathcal P}$, et al).  Since $\{\tp, \mathscr H_k\}=0$, generally $\tp$ and $\mathscr H_k$ cannot be diagonalized simultaneously.\\

A symmetry operation that commutes with the BdG Hamiltonian can be either the first kind or the second kind, or a combination of them. We write such an operation in a general form as $\alpha = \{s_\alpha, c_\alpha || p_\alpha |t_\alpha\}$, where $s_\alpha, c_\alpha, p_\alpha$ respectively denote the spin, charge and lattice point group operations and $t_\alpha$ is a fractional translation.  If the system contains spin-orbit coupling, then the spin rotations $s_\alpha$ are locked with the corresponding lattice point operations $p_\alpha$. 
In later discussion, we will denote the group formed by the combined operations $\alpha$ as $G_f=\{\alpha\}$, dubbed the `fermionic' group\cite{GuPhysRevX2018,yang2023pairingsymmetry}.    Double magnetic space groups (MSGs)\cite{BraCra10The-mathem} and double spin space groups  (SSGs)\cite{brinkman1966theory, SSG1, SSG2, SSG3} are examples of $G_f$ which  only contain one nontrivial charge operation $P_f$. 
 
When including the particle-hole `symmetry', the full group will be denoted as 
$$F_f=G_f+\mathcal P G_f $$ 
with $\mathcal P^2=E$ or $\mathcal P^2=P_f$. Notice that the fermionic group $G_f$ is generally an extension of the space group of the lattice which is denoted as the `bosonic' group $G_b$. For instance, in section \ref{sec:p4gmz2T}, we will study the wallpaper group $G_b=p4gm\times Z_2^{\mathcal{T}}$ with generators $C_4, \T, \{\mathcal M_x|{1\over2},{1\over2}\}$. When ignoring the spin rotation symmetry and introducing the reduced Nambu bases, the corresponding fermionic group $G_f$ is generated by $C_4'=(E||C_4), \T'=(I K ||\T)$ and $\mathcal M_x'=\{e^{-\mi{\tau_z\over2}\pi}||\mathcal M_x|{1\over2},{1\over2}\}$ satisfying ${\T'}^{2}=1, {\mathcal M_x'}^{2}=\{P_f||1,0\}$ and $G_f/Z_2^f = G_b$ (here the $'$ stands for lattice operations dressed by spin or/and charge operations).

In the present work, we will not dwell on the structures of all possible fermionic groups $G_f$ for a given $G_b$, and we just distinguish two types of them: in the first type, $G_f$ is `diagonal' in the particle-hole sector (such as the double MSGs/SSGs); 
in the second type, $G_f$ contains symmetry operations that are non-diagonal in the particle-hole sector (such as projective symmetry groups of $Z_2$ QSLs).

In the following discussion, we will focus on the spectrum at a momentum $\bm{k}$ that is invariant under an 
effective particle-hole symmetry $\tilde{\mathcal{P}}$,  since zero modes at momentum $\bm{k}$ without this symmetry are unstable when varying the chemical potential. 
Denoting the group of symmetry operations that commute with $\mathscr{H}_k$ as $G_{f}(k)$, then the full little co-group $F_f(k)$ is given by $F_f(k)=G_{f}(k) + \tp G_{f}(k)$. Since $[\mathscr{H}_k, G_{f}(k)]=0$, generally the eigenstates of $\mathscr H_k$ with eigvalue $\veps_k$ carry an irreducible projective representation of $G_f(k)$ (the factor systems $\omega_2(g_1, g_2)$ are given in App.\ref{sec:BandReps}). Noticing that $\tp$ maps $|\veps_k\rangle$ to $|-\veps_k\rangle$, in order to obtain more information of the energy spectrum especially the existence or nonexistence of zero energy modes, we need to consider the full group $F_f(k)$ and its irreps. 

$F_f(k)$ has two kinds of projective irreps (see App.\ref{App.:AA}): \\
(1){\it simple irreps}, which are irreducible when restricted to the subgroup $G_f(k)$; \\
(2){\it composite irreps}, the restricted rep of $G_f(k)$ for a composite irreps is a direct sum of two irreps of $G_f(k)$. 





\subsection{Robust and accidental Zero modes}\label{SubSec:IrrZeromodes}  

Here we systematically study the zero modes in BdG Hamiltonians. Suppose that $\mathscr H_k$ contain zero modes and assume that the zero modes span a Hilbert space $\mathcal L_0$.  
%

(i) If $\mathcal L_0$ carries a simple irrep of $F_f(k)$, namely if $\mathcal L_0$ also carries an irrep of $G_f(k)$, then the zero modes are stable in the sense that they cannot be lifted when keeping all of the symmetries. We call this set of zero modes as `irreducible zero modes'.



Although a set of irreducible zero modes carry a simple irrep of $F_f(k)$, the inverse is not necessarily true. Namely, a simple irrep of $F_f(k)$ may not contribute any zero modes to the spectrum because two simple irreps may couple with each other to form nonzero energy modes. So we have the following theorem.\\

\begin{theorem}\label{thm:Allhavepartner}
Every simple projective irrep $D(F_f(k))$ 
has a partner $D'(F_f(k))$ (carrying the same factor system of $D(F_f(k))$) with whom it can couple to form nonzero energy modes if the bases of $D(F_f(k))$ and $D'(F_f(k))$ are both included in the system.
\end{theorem}

The proof is simple. If $\tp$ is unitary, noticing $\tp^2\in G_f(k)$, one can set $D'(G_f(k))=D(G_f(k))$ and $D'(\tp)=-D(\tp)$. Then the Hamiltonian coupling the two reps $D(F_f(k))$ and $D'(F_f(k))$ reads $h(k) = \veps_k\mu_x \otimes I,$ where $\mu_m, m=x,y,z$ are the Pauli matrices acting on the $D(F_f(k))\oplus D'(F_f(k))$ direct sum sector. If $\tp$ is anti-unitary, one can choose $D'(F_f(k))=D(F_f(k))$, then a possible Hamiltonian that couples the two reps $D(F_f(k))$ and $D'(F_f(k))$ is $h(k) = \veps_k\mu_y\otimes I.$\\

According to theorem \ref{thm:Allhavepartner}, the irrep $D(G_f(k))$ gives rise to `irreducible zero modes' only when its partner $D'(F_f(k))$ is absent in the bases of the system. Later we will show that this is a stringent condition and it happens only when the system contains special symmetries. 


(ii) If $\mathcal L_0$ carries a composite irrep of $F_f(k)$ (namely if $\mathcal L_0$ carries a reducible rep of $G_f(k)$),  then the zero modes in $\mathcal L_0$ are accidental and can be lifted by small perturbations without breaking any symmetry.

(iii) If $\mathcal L_0$ carries a reducible rep of $F_f(k)$, then the zero modes contributed by composite irreps or couples of simple irreps (the simple irreps in each couple are partners of each other) are accidental, while the ones contributed by uncoupled simple irreps are robust. 

%

\subsection{Minimal nonzero modes: $\{\veps_k, -\veps_k\}$ with $\veps_k\neq0$}\label{subsec:MNonZM}

Now we discuss the modes having nonzero energy. Since $\tilde{\mathcal P}$ maps the $n$-dimensional ($n$-D) eigenspace of $\mathscr H_k$ with eigenvalue $\veps_k$ to the one with eigenvalue $-\veps_k$, we need to investigate the $2n$-D subspace $\mathcal L_{\pm\veps_k}$ formed by the direct sum of the eigenspaces $\{\veps_k, -\veps_k\}$ for $\veps_k\neq0$. This subspace $\mathcal L_{\pm\veps_k}$ is called minimal if it cannot be divided into smaller sets of nonzero modes without breaking any symmetry.  There are two possible types of minimal nonzero modes: (1) $\mathcal L_{\pm\veps_k}$ carries a composite irreducible rep of $F_f(k)$ and is called a set of `irreducible nonzero modes'; (2) $\mathcal L_{\pm\veps_k}$ carries a direct sum of two simple irreps of $F_f(k)$, and $\mathcal L_{\pm\veps_k}$ will be called a set of `reducible minimal nonzero modes' (RMNZM).

\begin{table*}[t]
\caption{Reducible minimal nonzero modes for the space $D(F_f(k))\oplus D'(F_f(k))$ of the group $F_f(k)=G_f(k)+\tp G_f(k)$. `U' stands for unitary and `A' denotes anti-unitary. The situations for unitary $\tp$ with $\eta_{\tp}=1$ correspond to two sets of uncoupled irreducible zero modes.}
\small
\centering
\begin{tabular}{c|c|c|c|c}
\hline\hline
$\tp$    &   { $G_f(k)$}   &    {$D(G_f(k))=D'(G_f(k))$} &  { $D(\tp)[D'(\tp)]^{-1}$}   & condition for RMNZM \\
\hline
\multirow{6}{*}{U}    &     U               &        irreducible         &     $\eta_{\tp} I$     &\!\!\!\! \!\!\!\!  \!\!\!\! \!\!\!\! \!\!\!\!\!\!\!\! \!\!\!\!  \!\!\!\! \!\!\!\! \!\!\!\!   (1) $\eta_{\tp}=-1$,\ \ \ $D(F_f(k))\ncong D'(F_f(k))$\\
\cline{2-5}
&     \multirow{5}{*}{A}          &       $\mathbb R$-class         &     $\eta_{\tp} I$     &\!\!\!\! \!\!\!\!  \!\!\!\! \!\!\!\! \!\!\!\!\!\!\!\! \!\!\!\!  \!\!\!\! \!\!\!\! \!\!\!\!    (2) $\eta_{\tp}=-1$,\ \ \  $D(F_f(k))\ncong D'(F_f(k))$\\
\cline{3-5}  
&& \multirow{2}{*}{$\mathbb C$-class}         &     $\eta_{\tp} I$,\ \ if $\{D(\tp), \Omega_z\}\neq0$     &\!\!\!\! \!\!\!\!  \!\!\!\! \!\!\!\! \!\!\!\!\!\!\!\! \!\!\!\!  \!\!\!\! \!\!\!\! \!\!\!\!    (3) $\eta_{\tp}=-1$,\ \ \  $D(F_f(k))\ncong D'(F_f(k))$\\
\cline{4-5}%
&&&  $e^{\mi\theta\Omega_z}$, if $\{D(\tp), \Omega_z\}=0$    &\!\!\!\! \!\!\!\!  \!\!\!\! \!\!\!\! \!\!\!\!\!\!\!\! \!\!\!\!  \!\!\!\! \!\!\!\! \!\!\!\!    (4) $\theta\in[0,2\pi)$,\ \ \!\!$D(F_f(k))\cong D'(F_f(k))$\\
\cline{3-5}
&&\multirow{2}{*}{$\mathbb H$-class}        & \!\!\!\!   $\eta_{\tp} I$,\ \  { if $\{D(\tp), \Omega_n\}\neq0$ for any $\pmb n$}   &\!\!\!\! \!\!\!\!  \!\!\!\! \!\!\!\! \!\!\!\!\!\!\!\! \!\!\!\!  \!\!\!\! \!\!\!\! \!\!\!\!    (5) $\eta_{\tp}=-1$,\ \ \  $D(F_f(k))\ncong D'(F_f(k))$\\
\cline{4-5}%
&&  & $e^{\mi\theta\Omega_{\pmb n}}$, { if $\{D(\tp), \Omega_n\}=0$  for some $\pmb n$}    &(6) {$\theta\in[0,2\pi)$, ${\pmb n}\in S^2$,} \!$D(F_f(k))\cong D'(F_f(k))$\ \ \\
\hline  
A &  U              &   irreducible & $e^{\mi\theta} I$  & \!\!\!\! \!\!\!\!  \!\!\!\! \!\!\!\! \!\!\!\!\!\!\!\! \!\!\!\!  \!\!\!\! \!\!\!\! \!\!\!\!     (7) $\theta\in[0,2\pi)$,\ \ \!\!$D(F_f(k))\cong D'(F_f(k))$ \\
\hline\hline                      
\end{tabular}\label{tab:RMNZM}
\end{table*}

\subsubsection{Irreducible nonzero modes}\label{subsubsec:irrnonzero}

By definition, the supporting space $\mathcal L_{\pm\veps_k}$ of irreducible nonzero modes forms an irrep of $D(F_f(k))$. This irrep must be a composite irrep, namely, the restricted rep on $G_f(k)$ must be reducible, otherwise it would gives rise to a set of irreducible zero modes. We adopt the eigenbases of the Hamiltonian for the irreducible nonzero modes which takes the form $h(k)=\veps_k\mu_z\otimes I_n$, where $\mu_z$ is the third Pauli matrix whose eigenvalues $\pm1$ respectively label the eigen spaces with positive and negative energy. 
Defining $K_{-1}=K, K_{1}=I$, and $s(g)=-1$ for anti-unitary elements $g\in G_f(k)$ and $s(g)=1$ for unitary ones, then
according to App.\ref{App.:AA} the element $g\in G_f(k)$ is represented as 
$$
D(g)K_{s(g)}=\begin{pmatrix} d(g) & 0\\ 0 & d'(g)\end{pmatrix}K_{s(g)},
$$ 
where $d(g)K_{s(g)}$ and $d'(g)K_{s(g)}={\omega_2(g,\tp)\over\omega_2(\tp,\tp^{-1}g\tp)} K_{s(\tp)} d(\tp^{-1}g\tp) K_{s(\tp)}$ are two $n$-D irreps of $G_f(k)$, and $\tp$ is represented as
$$
D(\tp)K_{s(\tp)}=\begin{pmatrix} 0&  \omega_2(\tp,\tp)d(\tp^2)\\ I_n & 0\end{pmatrix}K_{s(\tilde{\mathcal{ P})}}
$$ 
with $\tp^2\in G_f(k)$.

\subsubsection{Reducible minimal nonzero modes}\label{subsubSec:coupledZMs}

We now consider the case in which $\mathcal L_{\pm\veps_k}$ contains two simple irreps of $F_f(k)$ 
noted as $D(F_f(k))$ and $D'(F_f(k))$ respectively (they are partner of each other). By definition, the `minimal nonzero modes' requires that {\it the restricted Reps $D(G_f(k))$ and $D'(G_f(k))$ are irreducible}. 
From Schur's lemma, the restricted reps $D(G_f(k))$ and $D'(G_f(k))$ must be equivalent such that they can couple to each other. That is, $D(G_f(k)) \cong D'(G_f(k))$. We thus assume that the bases have been adjusted such that 
\beq\label{DQpg}
D(G_f(k)) = D'(G_f(k)).
\eeq

Furthermore, as summarized in table~\ref{tab:RMNZM}, the ability of coupling between the two simple irreps
further restricts the relation between $D(\tp)K_{s(\tp)}$ and $D'(\tp)K_{s(\tp)}$. In later discussion, we define the unitary quantity 
$$X=D(\tp)[D'(\tp)]^{-1}$$
to denote the difference between $D(\tp)$ and $D'(\tp)$. As shown in App.\ref{app:RMNZMs}, $X$ is commuting with the restricted rep $D(G_f(k))$. 

\noindent{\bf (I) $\tp$ is unitary} {\it (chiral-like, {\it e.g.} $\tp=\mathcal{TP}$)}

(a) If $G_f(k)$ is unitary, then from Schur's lemma one has $X= \eta_{\tp}I_n$ with $\eta_{\tp}$ a constant. Since $\tp^2\in G_f(k)$, hence $[D(\tp)]^2=[D'(\tp)]^2$, namely $\eta_{\tp}^2=1$, $\eta_{\tp}=\pm1$, see App.\ref{Appsubsubsection:Gfuni}. When $\eta_{\tp}=-1$, then $\mathcal L_{\pm\veps_k}$ is indeed a set of RMNZM. When $\eta_{\tp}=1$, $\mathcal L_{\pm\veps_k}$ contains two simple irreps
uncoupled to each other, hence $\mathcal L_{\pm\veps_k}$ corresponds to zero modes. 

(b) If $G_f(k)$ is anti-unitary, denote the maximum unitary subgroup of $G_f(k)$ as $H_f(k)$ with $G_f(k)=H_f(k)+T_0H_f(k),\ T_0^2\in H_f(k)$.   The projective irreps of $G_f(k)$ are classified into $\mathbb{R}$, $\mathbb{C}$, and $\mathbb{H}$  classes \footnote{ The class is determined by the Wigner indicator of an irrep $D(G_f(k))$ of $G_f(k)=H_f(k)+T_0H_f(k)$, namely $W_D={1\over |H_f(k)|}\sum_{h\in H_f(k)} \omega_2(T_0h, T_0h) {\rm Tr} \left[D(T_0h)^2\right]$. The values of the Wigner indicator in the $\mathbb{R}$, $\mathbb{C}$ and $\mathbb{H}$ classes are $W_D^{\mathbb R}=1, W_D^{\mathbb C}=0, W_D^{\mathbb H}=-2$ respectively.
} according to their centralizers\cite{CMP}.


b1), When the restricted rep $D(G_f(k))$ is of the real class $\mathbb R$, then $X = \eta_{\tp}I_n$ still holds with $\eta_{\tp}=\pm1$. The case with $\eta_{\tp}=-1$ stands for RMNZM. 

b2), When the $D(G_f(k))$ is of the complex class $\mathbb C$, then the restricted rep $D(H_f(k))$ is a direct sum of two non-equivalent irreps with dimension ${n\over 2}$. Suppose $\Omega_z$ is diagonal in this direct sum space and is represented as $\Omega_z=\omega_z\otimes  I_{n\over2}$ with $\omega_z$ the third pauli matrix. 
Then the unitary elements in the centralizer of the rep $D(G_f(k))$ can be written as $e^{\mi\theta\Omega_z}=\cos\theta + \mi\sin\theta\Omega_z$, see App.\ref{Appsubsubsec:Gfanti}, where $\theta\in[0,2\pi)$. So we have
$ X = e^{\mi\theta_0\Omega_z} $
for some $\theta_0$ with a constraint, 
\beq\label{thetaGmz}
e^{-\mi\theta_0\Omega_z}D(\tp)=D(\tp)e^{\mi\theta_0\Omega_z}.
\eeq
If $\{D(\tp), \Omega_z\}\neq0$, then $\sin\theta_0=0$, so $\theta_0=0,\pi$. The case $\theta_0=\pi$ (yielding $X=-I_n$) stands for the RMNZM.  If $\{D(\tp), \Omega_z\}=0$, then $\theta_0\in[0,2\pi)$ and  $D(F_f(k))$ and $D'(F_f(k))$ are equivalent since they are related by a $U(1)$ transformation $\Psi^\dag \to \Psi'^\dag=\Psi^\dag e^{-\mi\Omega_z {\theta_0\over 2}}$, $D(\tp)\to D'(\tp) = e^{-\mi\Omega_z {\theta_0\over 2}}D(\tp)e^{\mi\Omega_z {\theta_0\over 2}} =e^{-\mi\Omega_z {\theta_0}}D(\tp)$. The two irreps can be coupled to form nonzero modes. For instance, when $\theta_0=0$ the Hamiltonian can be chosen as $h_k =\veps_k\mu_y \otimes \Omega_z.$

b3), When $D(G_f(k))$ is of the quaternionic class $\mathbb H$, then the unitary elements in the centralizer of the $D(G_f(k))$ can be written as  $e^{\mi\theta \Omega_n}$ (with $\Omega_n=\pmb \omega \cdot\pmb n\otimes I_{n\over2}$, $\pmb \omega$ the three pauli matrices) for any $\pmb n\in S^2, \theta\in[0,2\pi)$, see App.\ref{Appsubsubsec:Gfanti}. So we have
\beq\label{Gammaz}
X= e^{\mi\theta_0\Omega_{\pmb{n}_0}}.
\eeq
for some $\theta_0,\pmb{n}_0$ with the following constraint
\begin{align}
e^{-\mi\theta_0\Omega_{\pmb{n}_0}}D(\tp)=D(\tp)e^{\mi\theta_0\Omega_{\pmb{n}_0}}.
\end{align}
If $\{D(\tp), \Omega_{\pmb{n}_0}\}\neq0$ for any ${\pmb{n}_0}\in S^2$ (for instance, when ${\rm Tr\ D(\tp)}\neq0$), then only $\theta_0=0,\pi$ satisfies the relation (\ref{thetaGmz}). Then (\ref{Gammaz}) reduces to $X = \eta_{\tp}I_n$ with $\eta_{\tp}=\pm1$ and the case $\eta_{\tp}=-1$ stands for RMNZM. 
If there exist some ${\pmb{n}_0}$ such that $\{D(\tp), \Omega_{\pmb{n}_0}\}=0$, then 
 $D(\tp)$ and $D'(\tp)$ are related by an SU(2) transformation $\Psi^\dag \to\Psi'^\dag = \Psi^\dag e^{-\mi\Omega_{\pmb{n}_0}{\theta_0\over2}}$, $D(\tp)\to D'(\tp) = e^{-\mi\Omega_{\pmb{n}_0}{\theta_0\over2}}D(\tp)e^{\mi\Omega_{\pmb{n}_0}{\theta_0\over2}}=e^{-\mi\Omega_{\pmb{n}_0}{\theta_0}}D(\tp)$. 
Hence, the two irreps $D(F_f(k)$ and $D'(F_f(k))$ are equivalent, they can be coupled to form nonzero modes. For instance, when $\theta_0=0$ the Hamiltonian can be chosen as 
$h_k =\veps_k\mu_y \otimes \Omega_n.$ 

\noindent{\bf (II) $\tp$ is antiunitary}{\it (particle-hole like, {\it e.g.} $\tp=\mathcal I\mathcal P$)}\\
\indent If $G_f(k)$ is anti-unitary,  we can multiply $\tp$ with an anti-unitary element $T_0\in G_f(k)$ to obtain a chiral-like symmetry operator $\tp'=\tp T_0$, which has been discussed in case ({\bf I}). Therefore, we only need to consider the case in which $G_f(k)$ is unitary. Eq. (\ref{DQpg}) indicates that the two irreps $D(F_f(k))$ and $D'(F_f(k))$ are equivalent, and consequently $D'(\tp)K= e^{-\mi\theta}D(\tp)K$ with $\theta\in[0,2\pi)$, namely $X=e^{\mi\theta}I_n$, see App.\ref{Appsubsec:tpanti}. Now $\mathcal L_{\pm\veps}$ is a set of RMNZM since there exist a Hamiltonian matrix $h_k=\veps_k(\sin{\theta\over2}\mu_x+\cos{\theta\over2}\mu_y)\otimes I_n$ which commutes with $\mathscr D(G_f(k))=D(G_f(k))\oplus D(G_f(k))$ and anti-commutes with $\mathscr D(\tp)K = D(\tp)\oplus e^{-\mi\theta}D(\tp)K$.

\subsection{Conditions for the existence of Zero modes}\label{Sec:ConditionZeroCriterion}

Now we consider BdG systems on the lattice having a symmetry group $G_f$, and assume that a fixed number of complex fermion bases (namely local Wannier orbitals) are placed in certain Wyckoff positions in the unit cell. According to App.\ref{sec:BandReps}, for any given momentum $\bm k$, the rep of $F_f(k)$  can be read out from the generalized band representation, and then we can know whether there exist zero modes at the $\bm{k}$. If zero modes do exist, the degeneracy must be even because the dimension of the Hamiltonian $\mathscr H_k$ is even and the nonzero energy levels appear in pairs $\{\veps_k,-\veps_k\}$. In the following theorem~\ref{NoCNoZM} (see App.\ref{App.:thrm2} for proof), however, we show that if the $G_f(k)$ is `diagonal' in the particle-hole sector,  then no stable zero modes can be found  at a given momentum $\bm k$. \\

\begin{theorem}\label{NoCNoZM}
For any BdG system whose symmetry group $G_f(k)$ is `diagonal' in the particle-hole sector, all zero modes at the momentum point $\bm k$ are accidental and can be adiabatically lifted without breaking any symmetry. \\
\end{theorem}

However,  situations will be changed if the system contains `non-diagonal' charge symmetries like the (effective) charge conjugation symmetry. In App.\ref{App.:thrm3}, it is proved that at a given $\tp$ symmetric momentum point $\bm{k}$, one can always find (effective) charge-conjugation symmetries to turn a set of minimal nonzero modes into zero modes -- the minimal `$\mathcal{C}$-zero modes'. If the $\mathcal{C}$-zero modes originate from irreducible nonzero modes of $F_f(k)$, then they form {\it irreducible zero modes} of the group $F^c_f(k)=F_f(k)+\C F_f(k)$; otherwise if the $\mathcal{C}$-zero modes come from RMNZM of $F_f(k)$, then they are called {\it reducible minimal $\mathcal{C}$-zero modes}. Then we have the following theorem~\ref{Thm:Czeromode}. 
 \\
\begin{theorem}\label{Thm:Czeromode}
At a given $\p$-symmetric momentum $\pmb k$ ({\it i.e.} $\p\in F_f(k)$) of an arbitrary BdG system, any minimal nonzero modes (except for one special case) located at the momentum $\bm k$  can be turned into zero modes by adding a single charge-conjugation symmetry to the original BdG system. The exception is when $\tp$ is anti-unitary while $G_f(k)$ is unitary, and the irrep $D(F_f(k))$ carried by a set of irreducible nonzero modes is of $\mathbb{H}$ class, 
then one needs two charge-conjugation operations anti-commuting with each other to gain zero modes. \\
\end{theorem}

Charge conjugations are the simplest charge operations that are non-diagonal in the particle-hole sector. 
The above theorem can be generalized to cases where the charge conjugation operation $\C$ is replaced by general non-diagonal charge operations $\mathcal D$ in $SU(2)_c$, {\it e.g.}  
with $\mathcal D^3=E$ or $\mathcal D^6=E$. 

On the other hand,  zero modes can also appear at unfixed momenta. For instance, a level crossing of irreducible nonzero modes carrying two nonequivalent reps of $G_f(k)$ can give rise to zero modes on certain high symmetry line. So we have the following theorem:\\

\begin{theorem}\label{Thrm:crossing}
On a high symmetry line of the BZ having a unitary symmetry $\mathcal{M}$, if the eigenspaces $\{\veps_k, -\veps_k\}$ of a minimal nonzero modes carry different quantum numbers ({\it i.e.} characters of irreps) of $\mathcal{M}$, then the level crossing of $\veps_k, -\veps_k$ gives rise to a set of stable zero modes whose momentum is not fixed.\\
\end{theorem}

According to the above theorems, gapless quasiparticles can be found under one of the following three situations:

(1) presence of one or two (effective) charge-conjugation 
symmetries 
giving rise to one or more sets of irreducible zero modes 
at a high symmetry point $\pmb k$. 

(2) presence of level crossing of bands carrying different quantum numbers of certain unitary symmetry operation of a high symmetry line.

(3) occurrence of $\pi$-quantized Berry phase around generic $\bm k$ point due to the ($\mathcal I'\mathcal T'$) symmetry with $(\mathcal I'\T')^2=1$. Here $\mathcal I'$ and $\mathcal T'$ respectively stand for the inversion and the time-reversal symmetries dressed by spin or/and charge operations.

The condition (1) can be verified using the band representation theory for SCs/SFs/QSLs  provided in App.\ref{sec:BandReps}. The last two conditions are illustrated via concrete lattice models in Sec.\ref{Sec:ModelSection}.  Especially, the $\mathcal I'\T'$ symmetry enforced $\pi$-quantized Berry phase can give rise to nodal-point like quasiparticles in 2-dimensions and nodal-line structures in 3-dimensions. This is a topological origin of the zero modes (usually linear dispersive), and is closely related to the zero modes at generic $\bm{k}$ points enforced by symmetry indicators \cite{Slager2013,Slager2017,Po2017Sym-BasedIndicators, Bradlyn2017TopoQuanChemistry, Ono2018SIs, PRB2020_indicator, Ono2021Z2SIs, PRX2022-AZclass}. The (projective) representation theory of symmetry groups discussed in the present work provides a different and complementary mechanism to obtain zero modes and gapless quasiparticles. \\

\section{$k\cdot p$ theory and physical response}\label{sec:kdotpmodel}

The existence of degenerate modes (zero modes) in the bulk energy spectrum can give rise to nodal point or nodal line structure. The dispersion around the nodal points or nodal lines are reflected in the effective Hamiltonian called $k\cdot p$ theory.  In this section, we will adopt the Hamiltonian approach \cite{YanYanFan21A-Hamilton} to obtain the $k\cdot p$ theory for zero modes of general BdG Hamiltonians. 

Above all, we define the `particle-hole rep' of $F_f(k)$
\beq\label{ZgRep}
D^{(\rm ph)}(g)=+1, \ \ D^{(\rm ph)}(g\tilde{\mathcal P}) =-1,
\eeq
with $g\in G_f(k)$. Here the complex conjugation $K_{s(g)}$ is hidden since the particle-hole rep is essentially a real rep.

\subsection{$k\cdot p$ Hamiltonian}\label{sec:k.p}
When the full symmetry group $F_f$ and the electron bases in the unit cell are given, the band representation $D(F_f)$ of general BdG systems can be obtained (see App.\ref{sec:BandReps}), from which the stable zero modes are known.  Suppose the stable zero modes at momentum $\bm k$ span a Hilbert space $\mathcal L_0$. When projected onto $\mathcal L_0$, the effective low-energy BdG Hamiltonian reads
$$
H_{\rm eff}=\sum_{\delta k} \Phi_{\pmb k+\delta \pmb k}^\dag\Lambda ({\delta\bm k})\Phi_{\pmb k+\delta \pmb k}=\sum_{\delta k}H_{\pmb k+\delta \pmb k},
$$ 
where the $\Phi_{\pmb k}$ is the basis of the zero modes, and $\Lambda(\delta \pmb k)$ is a Hermitian matrix $\Lambda^\dag(\delta \pmb k)=\Lambda(\delta \pmb k)$. When summing over all the momentum variations, the total Hamiltonian should preserve the $F_f(k)$ symmetry, namely
$\hat g\left(\sum_{\delta \pmb k}H_{\pmb k+\delta\pmb  k}\right)\hat g^{-1}=\left(\sum_{\delta\pmb  k}H_{\pmb k+\delta \pmb k}\right),$ for all $ g\in F_f(k)$.  Assuming $\hat g\Phi^\dag (\pmb k+\delta \pmb k)\hat g^{-1} =  \Phi^\dag (\pmb k+g\delta \pmb k) D(g)K_{s(g)} $, then generally for any $g\in F_f(k)$ one has
\beq\label{gsrh0}
D(g) K_{s(g)} \Lambda \left( \delta \pmb k\right) K_{s(g)} D^{\dagger}(g)=D^{(\rm ph)}(g)\Lambda(g \delta \pmb k),
\eeq
where $D^{(\rm ph)}(g)$ is the 1-D particle-hole Rep of $F_f(k)$ defined in (\ref{ZgRep}).

Starting from (\ref{gsrh0}), we derive the formula to judge the dispersion relation of the BdG system in momentum space around the zero modes.  For instance, we consider linear dispersion around $\pmb k$, namely
\beq
\Lambda(\delta \pmb k)=\sum_{m=1}^d \delta k_m\Lambda^m+O(\delta k^2).
\eeq
Here $\delta\pmb  k$ is a dual vector under the action of group $F_f(k)$, namely
$
\hat g\delta k_m=\sum_n D_{mn}^{(\bar {\rm v})}(h)\delta k_n,
$
where $(\bar {\rm v})$ is the dual Rep of the vector Rep $(\rm v)$ of the group $F_f(k)$. Notice that $\delta \pmb k$ reverses its sign under the following actions: the time reversal $\mathcal T'$ (anti-unitary), the spatial inversion $\mathcal I'$ (unitary), the particle-hole transformation $\mathcal P$ (anti-unitary), hence,
\Beq
D^{(\bar {\rm v})}(\T')K=-1K, D^{(\bar {\rm v})}(\mathcal I')=-1, D^{(\bar {\rm v})}(\mathcal P)K=-1K.
\Eeq
(recall that the $'$ stands for lattice operations dressed by spin or/and charge operations.)

{\it (I) $F_f(k)$ is unitary.} If $F_f(k)$ is a unitary group, then for any $g\in F_f(k)$ one has
\begin{align}\label{varvZUni}
D(g) \Lambda^n  D(g)^\dagger  = \sum_{m} {D}_{mn}^{(\bar {\rm v}\times {\rm ph})}(g)\Lambda^m
\end{align}
with ${D}^{(\bar {\rm v}\times {\rm ph})}(g) \equiv D^{(\rm ph)}(g)D^{(\bar {\rm v})}(g)$. 
Thus, the existence of linear dispersion is determined by whether the product Rep $D(F_f(k))\otimes D^*(F_f(k))$ contains the linear Rep $D^{(\bar {\rm v}\times {\rm ph})}(F_f(k))$ or not. It can be judged from the quantity
\begin{align}\label{Sec:kp20}
a_{(\rm \bar v\times{\rm ph})}={1\over |F_f(k)|}\sum_{g\in F_f(k)} |\chi(g)|^2\cdot\chi^{(\rm v\times{\rm {\overline {ph}}})}(g) 
\end{align}
where $\chi^{({\rm v}\times {\rm {\overline {ph}}})}(g)\equiv\chi^{({\rm v}\times {\rm ph})}(g) =\operatorname{Tr}[ {D}^{( \bar{\rm v}\times {\rm ph})}(g)]^*$. If $a_{({\rm \bar v}\times {\rm ph})} =0$, then the dispersion must be of order higher than 1. If $a_{({\rm \bar v}\times {\rm ph})} \neq 0$, then the dispersion is linear, and one can always find  Hermitian matrices $\Lambda^m$ satisfying the equation (\ref{varvZUni}) noticing that the Rep $({\rm \bar v}\times {\rm ph})$ is real.

{\it (II) $F_f(k)$ is anti-unitary.} The matrices $\Lambda^m$ also carry dual vector rep of the group $F_f(k)$. 

For any $g\in F_f(k)$, one has
\beq\label{varvZ}
D(g) K_{s(g)}\Lambda^n K_{s(g)} D(g)^\dagger  = \sum_{m} \Lambda^m {D}_{mn}^{(\bar {\rm v}\times {\rm ph})}(g),
\eeq
where ${D}^{(\bar {\rm v}\times {\rm ph})}(g) K_{s(g)}\equiv D^{(\rm ph)}(g)D^{(\bar {\rm v})}(g) K_{s(g)}$ is a real Rep of $F_f(k)$ hence the operator $K_{s(g)}$ can be hidden. Hence the existence of linear dispersion also depends on whether the product Rep $D(F_f(k))\otimes D^*(F_f(k))$ contains the Rep $D^{(\bar {\rm v}\times {\rm ph})}(F_f(k))$ or not, under the condition that the CG coefficients can be reshaped into hermitian matrices.

Select an anti-unitary operation $T_0\in F_f(k)$ and introduce the following matrices
$$
\tilde\Lambda^m=\Lambda^m\cdot D^{\rm T}(T_0).
$$ 
According to the Hamiltonian approach\cite{YanYanFan21A-Hamilton}, the symmetry condition (\ref{varvZ}) with $g=T_0$ and the hermitian condition $(\Lambda^m)^\dag =\Lambda^m$ are combined into a single symmetry condition of $\tilde\Lambda^m$ (called $\eta_0$-symmetry condition): 
\begin{align}\label{eta0-condition}
\big(\tilde\Lambda^m\big)^{\rm T}=\eta_0 D(T_0^2)\sum_n D^{(\bar {\rm v}\times {\rm ph})}_{nm}(T_0)\tilde\Lambda^n.
\end{align}
with $\eta_0=\omega_2(T_0,T_0)$.

Therefore, when restricted to the $\eta_0$-symmetric subspace, one only needs to consider the unitary symmetries of $\tilde\Lambda^m$ as discussed in case (I). This is a relatively simpler since the unitary elements are represented in the $\mathbb{C}$ field. On the other hand, $\Lambda^m$ carry real representations of $G_f(k)$ so one need to treat them in the $\mathbb{R}$ field which is more complicated. This subtlety makes the Hamiltonian approach a method with a high efficiency.

Similar to Eq.(\ref{Sec:kp20}), the existence of linear dispersion can be judged by the following quantity
\begin{align}\label{DispersionCriterion}
a_{({\rm \bar v}\times {\rm ph})} =&\frac{1}{|F_f(k)|} \sum_{u }\Big[|\chi(u)|^{2} \chi^{({\rm v}\times {\rm {\overline{ph}}})}(u) + \notag\\
	   &\omega_2\left(u T_0, uT_0\right) \chi\big(\left(u T_0\right)^{2}\big) \chi^{({\rm v}\times {\rm{\overline{ph}}})}\left(u T_0\right) \Big],
\end{align} 
where the sum $\Sigma_u$ runs over all unitary elements $u\in F_f(k)$, namely if $a_{({\rm \bar v}\times {\rm ph})} \neq 0$, the dispersion is linear, otherwise the dispersion is of order higher than 1. Notice that the character $\chi^{({\rm v}\times {\rm{\overline{ph}}})}\left(u T_0\right)$ in (\ref{DispersionCriterion}) is well-defined although $uT_0$ is anti-unitary, because the rep $({\rm v}\times {\rm{\overline{ph}}})$ is a real rep such that the allowed bases transformations can only be real orthogonal matrices which keep $\chi^{({\rm v}\times {\rm{\overline{ph}}})}\left(u T_0\right)$ invariant. 

Now we discuss how to obtain the matrices $\Lambda^m$ or $\tilde \Lambda^m$ in the case $a_{({\rm \bar v}\times {\rm ph})} \neq 0$.  Besides the $\eta_0$-symmetry condition (\ref{eta0-condition}), the symmetry constraints for unitary elements $u\in F_f(k)$ reads
\begin{align}
D(u)\tilde\Lambda^mR^{\rm T}(u)&=\sum_n D^{(\bar {\rm v}\times {\rm ph})}_{nm}(u)\tilde\Lambda^n,\label{F-condition}
\end{align}
with $R(u) = D(T_0) D^*(u) D^\dag(T_0).$
If we consider the set of matrices $\tilde\Lambda^m$ as a single column vector $\tilde{\pmb\Lambda}$ with
$(\tilde{\pmb\Lambda})_{(n-1) \times d^2+(\alpha-1)\times d+\beta}=(\tilde\Lambda^n)_{\alpha\beta},$
then the above constrains (\ref{F-condition}) indicates that $\tilde{\pmb\Lambda}$ carries identity rep of any unitary element $u\in F_f(k)$, namely 
\beq\label{WGm}
W(u)\tilde{\pmb\Lambda}=\tilde{\pmb\Lambda},
\eeq
with 
$
W(u)=D^{({\rm v}\times {\rm{\overline{ph}}})}(u)\otimes D(u)\otimes R(u).
$
Furthermore, $W(T_0)K \tilde{\pmb\Lambda}=\tilde{\pmb\Lambda}$ also holds where $T_0$ is represented as $W(T_0)K = D^{({\rm v}\times {\rm{\overline{ph}}})}(T_0)\otimes D(T_0)\otimes D(T_0)K$.

When projected to the $\eta_0$-symmetric space, the rep $W(u)$ becomes
$W_{\eta_0}(u)\equiv {\rm{Proj}_{\eta_0}}W(u) $ with the matrix entries given by
\Beq
\!\!\! &&[W_{\eta_0}(u)]_{m\gamma\rho,n\alpha\beta}={1\over2}\Big[ D_{mn}^{({\rm v}\times {\rm{\overline{ph}}})}(u) D_{\gamma\alpha}(u) R_{\rho\beta}(u) + 
\\&&\ \ \ \ \eta_0 \sum_\lambda D_{mn}^{({\rm v}\times {\rm{\overline{ph}}})}(uT_0) D_{\gamma\beta}(u) R_{\rho \lambda}(u)D_{\lambda\alpha}\left((T_0)^2\right) \Big]. 
\Eeq
The above matrix $W_{\eta_0}(u)$ is not a rep for the group formed by all unitary elements $u\in F_f(k)$, but it does contain the identity rep for that group, which shows that the dispersion contains linear term. Then the $\Lambda^m$ can be obtained via the following procedure: \\
(1) Obtain the common eigenvectors of $W_{\eta_0}(u)$ with eigenvalue 1 for all unitary elements $u\in F_f(k)$. These eigenvectors span a Hilbert subspace $\mathcal L^{(I)}$. \\
(2) Project $W(T_0)K$ into $\mathcal L^{(I)}$, and perform bases transformation such that $T_0$ is represented as $IK$ in $\mathcal L^{(I)}$. Then the new bases are all of the allowed $\tilde{\pmb \Lambda}$. \\
(3) Reshape each of the new bases into three matrices $\tilde \Lambda^x, \tilde \Lambda^y, \tilde \Lambda^z$. \\
(4) Finally one has $\Lambda^m=\tilde \Lambda^m [D(T_0)^{\rm T}]^{-1}$.\\

The above procedure can be applied to the 0-th order $k\cdot p$ terms, {\it i.e.}, to judge if an irrep 
is simple or composite. 
This can be done by replacing the rep $({\rm v}\times {\rm{\overline{ph}}})$ in eq.(\ref{DispersionCriterion}) with the identity rep or the particle-hole rep. If $a_I=1, a_{\rm ph}=0$, then the irrep is simple and cannot form nonzero modes by itself; 
if $a_I=1, a_{\rm ph}\neq 0$, then the irrep is composite which corresponds to a set of irreducible nonzero modes. 

Similarly, one can obtain higher order $k\cdot p$ terms. For instance, the linear rep carried by the quadratic dispersion $k_\alpha k_\beta$ is shown in Tab.\ref{tab:prob}, where the $l=0$ rep is the identity rep and the $l=2$ rep can be reduced to direct sum of lower dimensional irreps $(\mu_1)+(\mu_2)+...$. From the irreps $(\mu_1),(\mu_2),...$ one can construct the corresponding matrices $\Lambda^{(\mu_1)m_1}, ...$ that constitute the effective Hamiltonian 
\Beq
H = \sum_{\mu, m} f^{(\mu)_m}(\delta k)\Lambda^{(\mu)m} + \delta {\bm k}\cdot\delta {\bm k}\ \Lambda^{(0)},
\Eeq
where $f^{(\mu)_m}(\delta k)$ is the $m$-th basis (a quadratic polynomial of $\delta k$) of the irrep $\mu$ and $\Lambda^{(0)}$ is the $l=0$ term.

\begin{table}[t]
\caption{Linear Reps of the lattice rotation ($l$ stands for angular momentum), time reversal $\T'$, spatial inversion $\mathcal I'$, charge conjugation $\C$ and particle-hole transformation $\mathcal P$ carried by probe fields: chemical potential $\lambda$, momentum $\pmb k$, 
  gradience of temperature $\pmb\nabla T$, magnetic field $\pmb B$, quadratic dispersion $k_{\alpha}k_\beta$, strain field $\epsilon_{\alpha\beta}$, and so on. The last column is the `particle-hole rep' of operations in $F_f(k)$. }
\small
\centering
\begin{tabular}{c|c|c|c|c|c|c|c}
\hline\hline
prob fields    & $\lambda$ &  $\pmb\nabla T$ &  $\pmb B$  & $\pmb k$ & $k_\alpha k_\beta$ & $\epsilon_{\alpha\beta}$ & $D^{\rm{(ph)}}$\\ 
\hline 
rotation           &    $l=0$      &  $l=1$    &  $l=1$      &   $l=1$       &  $l=0,2$   &  $l=0,2$ & $+1$\\
$\T'$               &    $+1$       &    $+1$   &    $-1$     &  $-1$        &    $+1$   &  $+1$ & $+1$\\
$\mathcal I'$    &    $+1$       &    $-1$   &    $+1$     &    $-1$      &    $+1$   &  $+1$ & $+1$\\
$\C$               &    $-1$       &    $+1$   &    $+1$     &   $+1$       &    $+1$   &  $+1$ & $+1$\\
$\mathcal P$   &    $+1$       &    $+1$    &   $+1$      &   $-1$       &    $+1$   &  $+1$ & $-1$\\
\hline\hline                      
\end{tabular}\label{tab:prob}
\end{table}

\subsection{Physical properties}
The above method can also be applied to study the physical properties of the quasiparticles, namely, the response to external fields $\pmb f$, such as the gradience of temperature $\pmb\nabla T$, 
magnetic field $\pmb B$, strain field $\epsilon_{\alpha\beta}$, and so on. As summarized in Tab.\ref{tab:prob}, the prob fields carry linear reps of $g\in F_f(k)$,  namely
\Beq
g f_\alpha = \sum_\beta D_{\beta\alpha}^{(\mu)}(g) f_\beta,
\Eeq
where $(\mu)$ is a linear rep of $F_f(k)$. Then in the low energy effective Hamiltonian, fermion bilinear terms should be added to the $k\cdot p$ model,
\Beq
H' = \sum_\alpha f_\alpha \Lambda_\alpha.
\Eeq
Similar to (\ref{varvZ}), one has $D(g)K_{s(g)}\Lambda_\alpha K_{s(g)} D(g)^{-1} = \sum_\beta D^{{\rm (ph)}\times (\mu)}_{\beta\alpha}(g) \Lambda_\beta$. The matrices $\Lambda_\alpha$ can be obtained using the method introduced in the previous subsection.
Notice that electric fields $\pmb E$ are not an applicable probe fields, since electric fields are screened in SCs and do not couple to the charge neutral SFs/QSLs. Instead,  temperature gradience $\pmb\nabla T$ is an ideal probe field for SCs/SFs/QSLs, and $\pmb\nabla T$ varies under symmetry operations in the same way as $\pmb E$. Magnetic fields can be applied as Zeeman fields coupling to SCs/SFs/QSLs. Despite the Meisner effect of SCs, magnetic fields  can be applied along in-plane directions for thin films of SCs. 

Due to the `particle-hole symmetry' $\mathcal P$ and the effective charge conjugation symmetry $\C$, the response of the Bogoliubov quasiparticles to external field is different from quasiparticles in semimetals. For instance, the chemical potential $\lambda$ is a scaler field under space-time operations, but it does not carry an identity rep of the group $F_f^c(k)=F_f(k)+\C F_f(k)$, hence the response matrix $\Lambda_0$ is not simply an identity matrix. 
Consequently changing the chemical potential may either shift the position of the nodal point or fully gap out the quasiparticle.

Generally, the external fields may give mass to the gapless quasiparticles. The resulting gapped state may have nontrivial topological properties, such as nontrivial Chern numbers. These information are maintained in the $k\cdot p$ model and the response matrix mentioned above. In the following section, we will illustrate the above results by concrete lattice models.


\section{Lattice Models: gapless quasiparticles in BdG systems} \label{Sec:ModelSection}

In this section, we present four models in two-dimensional translationally symmetric lattice systems to illustrate the mechanisms and methods proposed in the previous sections. The first three models discuss zero modes protected by an effective charge conjugation $\C$, while the fourth model is a mean field Hamiltonian of $Z_2$ QSLs in which the projective symmetry group contains symmetry elements that are non-diagonal in the particle-hole sector.

\subsection{Square lattice: $G_b=p4gm\times Z_2^{\mathcal{T}}$}\label{sec:p4gmz2T}

\begin{figure}[t]
\centering
\label{p4gmmodelhopping}
\includegraphics[width=0.4\textwidth]{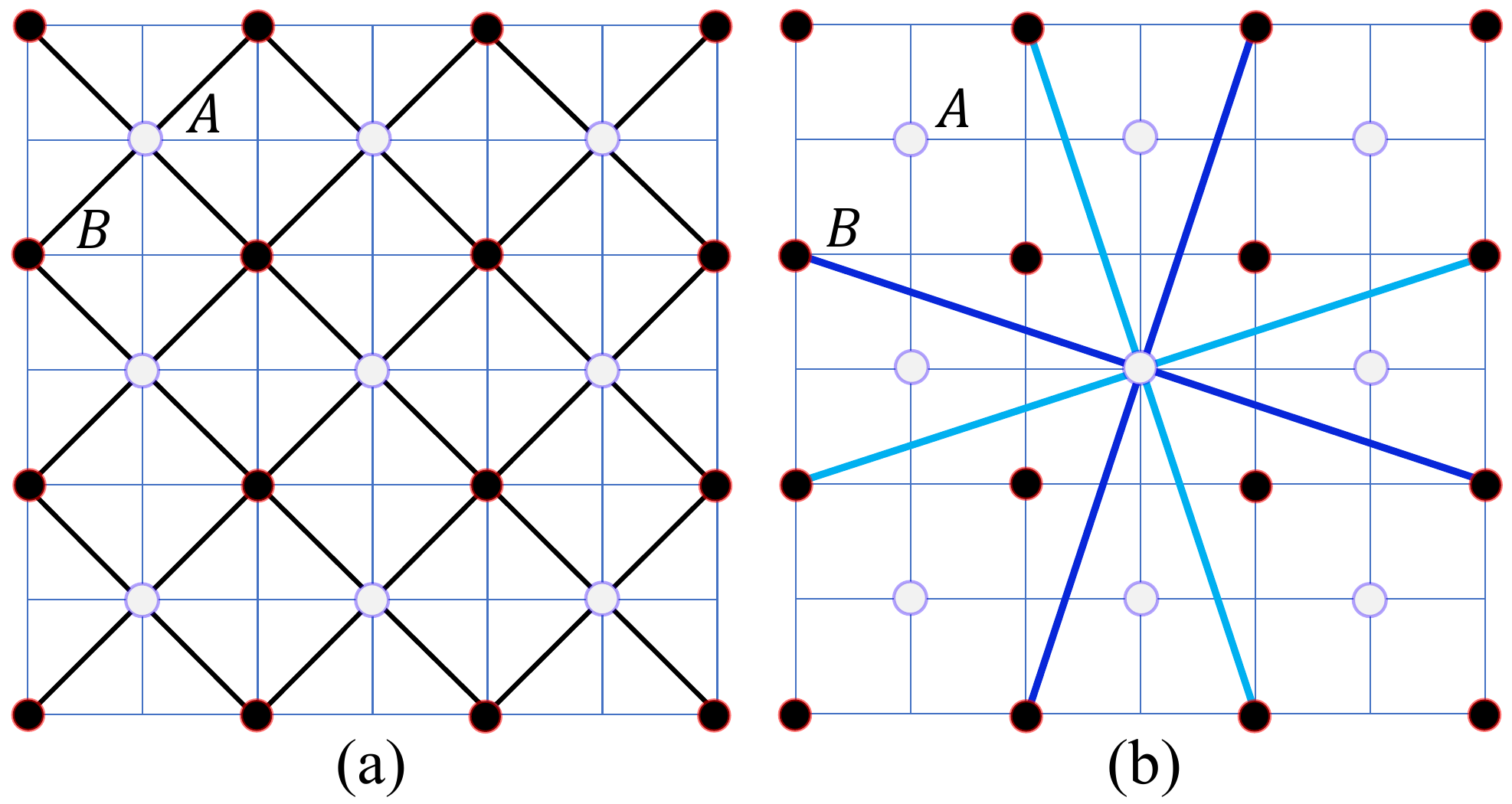}
 \caption{The model (\ref{p4gmHwosoc}) with symmetry  group $G_b=p4gm\times {Z}_2^\mathcal{T}$. (a) the hopping terms; (b) the pairing terms. The black and white circles stand for A and B sub-lattices respectively. The black bonds in (a) represent the hopping terms, and in (b), the cyan bonds stand for the phase $e^{\mi\theta_{i,j}}=+1$ while the blue bonds stand for $e^{\mi\theta_{i,j}}=-1$.}
\label{p4gmmodel}
\end{figure}

The first model is a spin-1/2 SC/SF of the CI class. We consider the square lattice with the magnetic  wallpaper group $G_b=p4gm\times Z_2^{\mathcal{T}}$. On the 2a Wyckoff positions A=(0,0) and B=(${1\over2}, {1\over2}$) with site group $\mathscr C_4\times Z_2^{\mathcal{T}}$, we place spin-1/2 fermionic orbitals $[c_\up^\dag, c_\dn^\dag]$. Assuming SOC is negligible, then the spin-1/2 fermion forms a Kramers doublet under time reversal $\T$ and is invariant under the site group rotation $\mathscr C_4$. Due to the SU(2) spin rotation symmetry, the BdG Hamiltonian can be written in the reduced Nambu bases 
$
\psi_i^\dag=[c_{{A}\up}^\dag,  c_{{B}\up}^\dag, c_{{A}\dn}, c_{{B}\dn}]_i,
$ 
where the symmetry group $G_f$ is generated by $C_4=C_4, \T'=(IK ||\T), \mathcal M_x'=\{e^{-\mi{\tau_z\over2}\pi}||\mathcal M_x|{1\over2},{1\over2}\}$ (here $e^{-\mi{\tau_z\over2}\pi}$ is a charge operation defined in (\ref{SU2cr})) and translation, with
\Beq
&&\!\!\!\!\!\!\hat C_4 \psi^\dag_i \hat C_4^{-1} \!=\! \psi^\dag_{C_4 i}, \ \ \ 
\hat \T' \psi^\dag_i \hat \T^{'-1} \!=\! \psi^\dag_i K, \\
&&\!\!\!\!\!\!\hat {\mathcal M}_x' \psi^\dag_i \hat {\mathcal M}_x^{'-1} = \\
&&\mi[c_{{B}\up, \mathcal M_x i}^\dag,  c_{{A}\up, \mathcal M_x i + (0,1)}^\dag, -c_{{B}\dn, \mathcal M_x i}, -c_{{A}\dn, \mathcal M_x i + (0,1)}].
\Eeq
Besides the particle-hole transformation
\[
\hat{\mathcal P} \psi^\dag_i \hat{\mathcal P}^{-1} = [c_{{A}\dn},  c_{{B}\dn}, -c_{{A}\up}^\dag, -c_{{B}\up}^\dag]_i K,
\]
we  further consider an effective charge conjugation $\C= P_f^B e^{-\mi{\tau_y\over2}\pi} $ with $P_f^B[c_{A}, c_B](P_f^B)^{-1} = [c_A,-c_B]$,
$$
\widehat { \C} \psi^\dag_i \widehat { \C}^{-1} = [c_{{A}\dn}, -c_{{B}\dn}, -c_{{A}\up}^\dag, c_{{B}\up}^\dag]_i,
$$ 
and ${\widehat\C}^2=P_f$. Then we obtain the complete symmetry group $F^c_f$, whose band representation can be obtained according to App.\ref{sec:BandReps}. Irreducible zero modes are found at $X=({1\over2},0), Y=(0,{1\over2}), M=({1\over2},{1\over2})$ points, where the $X$ and $Y$ points are related by $C_4$ symmetry.

\begin{figure}[t]
\centering
\subfigure[ band structure with $\C$]{\includegraphics[scale=0.17]{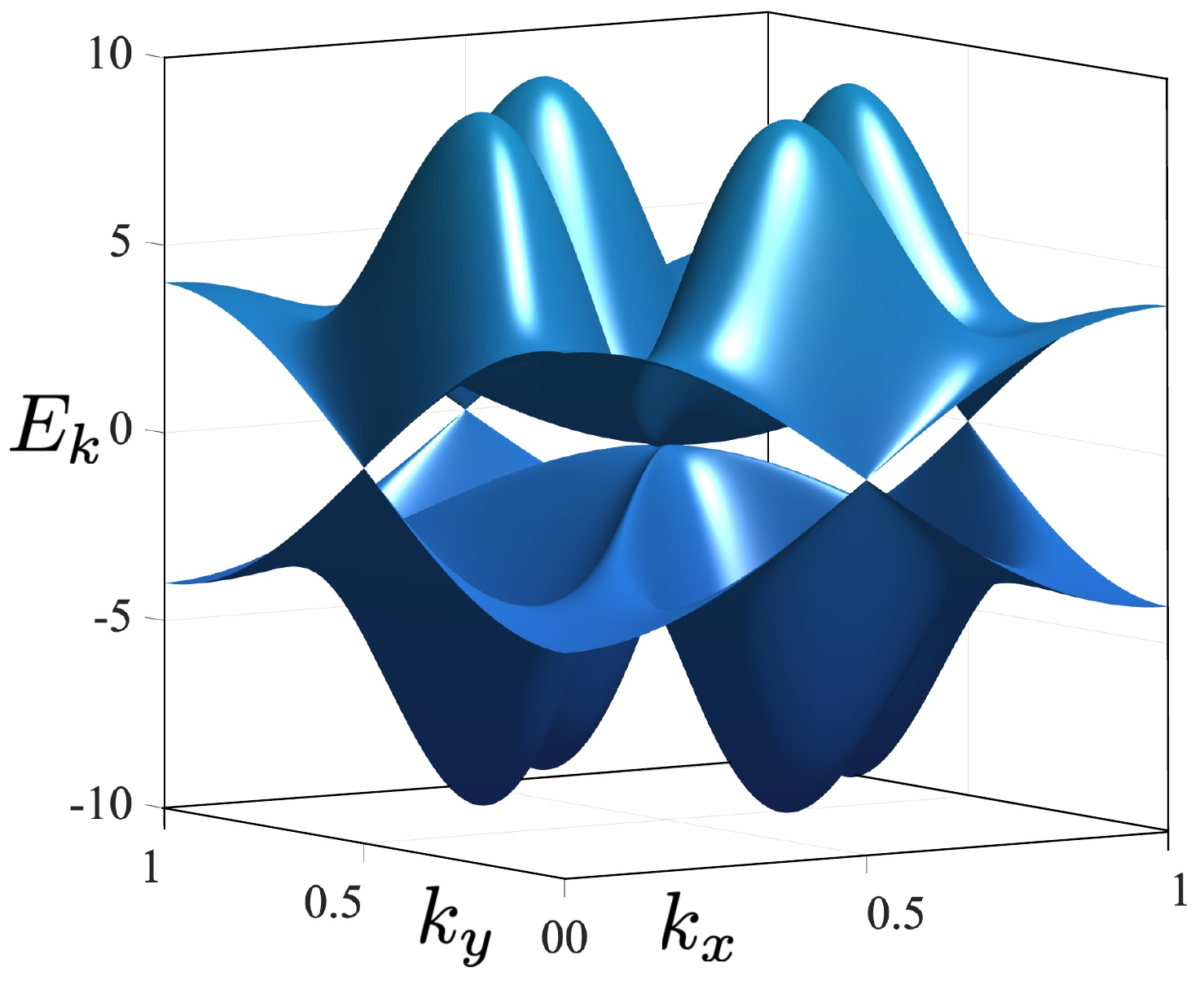}\label{p4gmBand_a}}
\subfigure[ zero modes in (a)]{\includegraphics[scale=0.19]{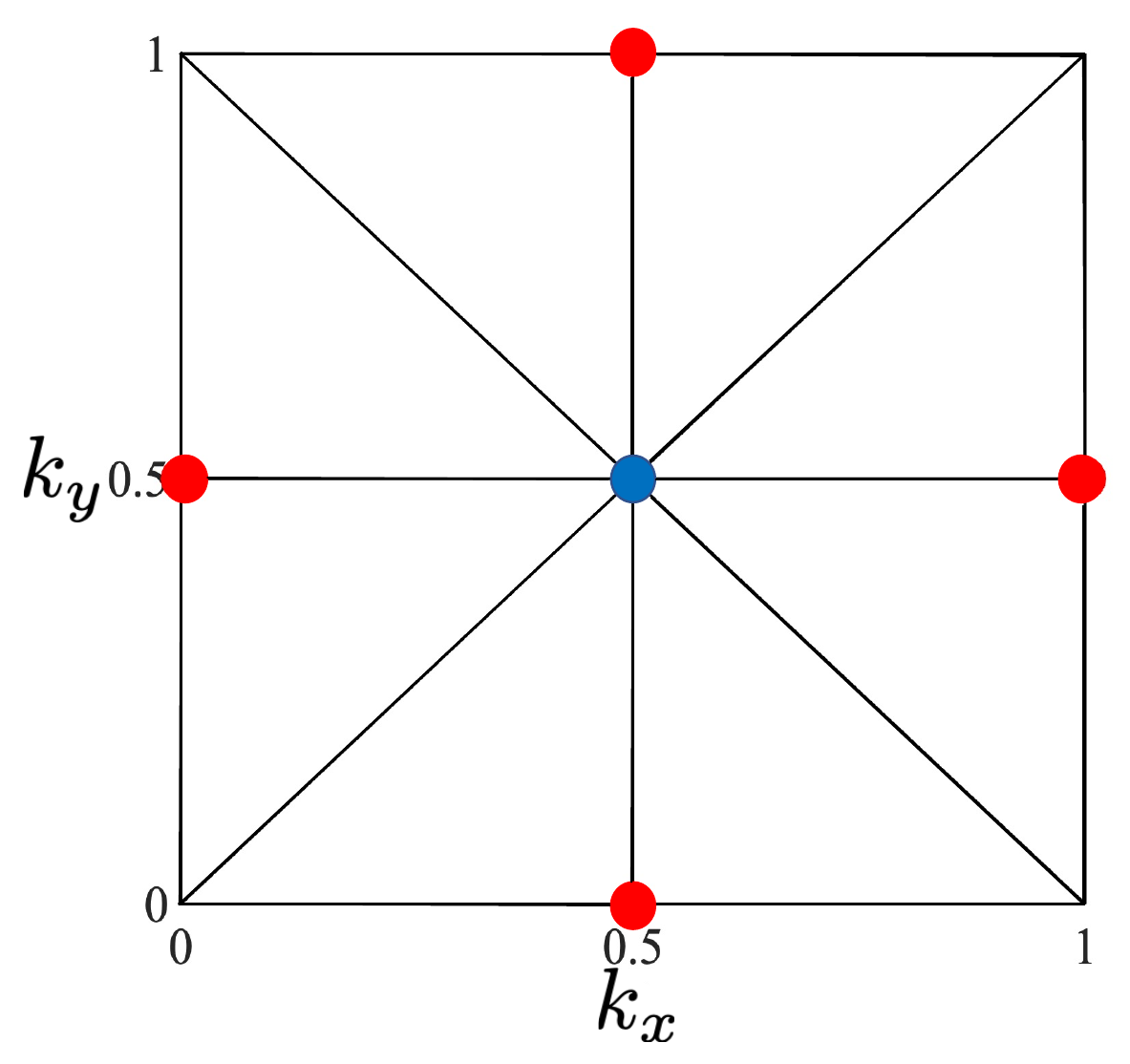}\label{p4gmBand_a1}}
 \subfigure[ band structure without $ \C$]{\includegraphics[scale=0.17]{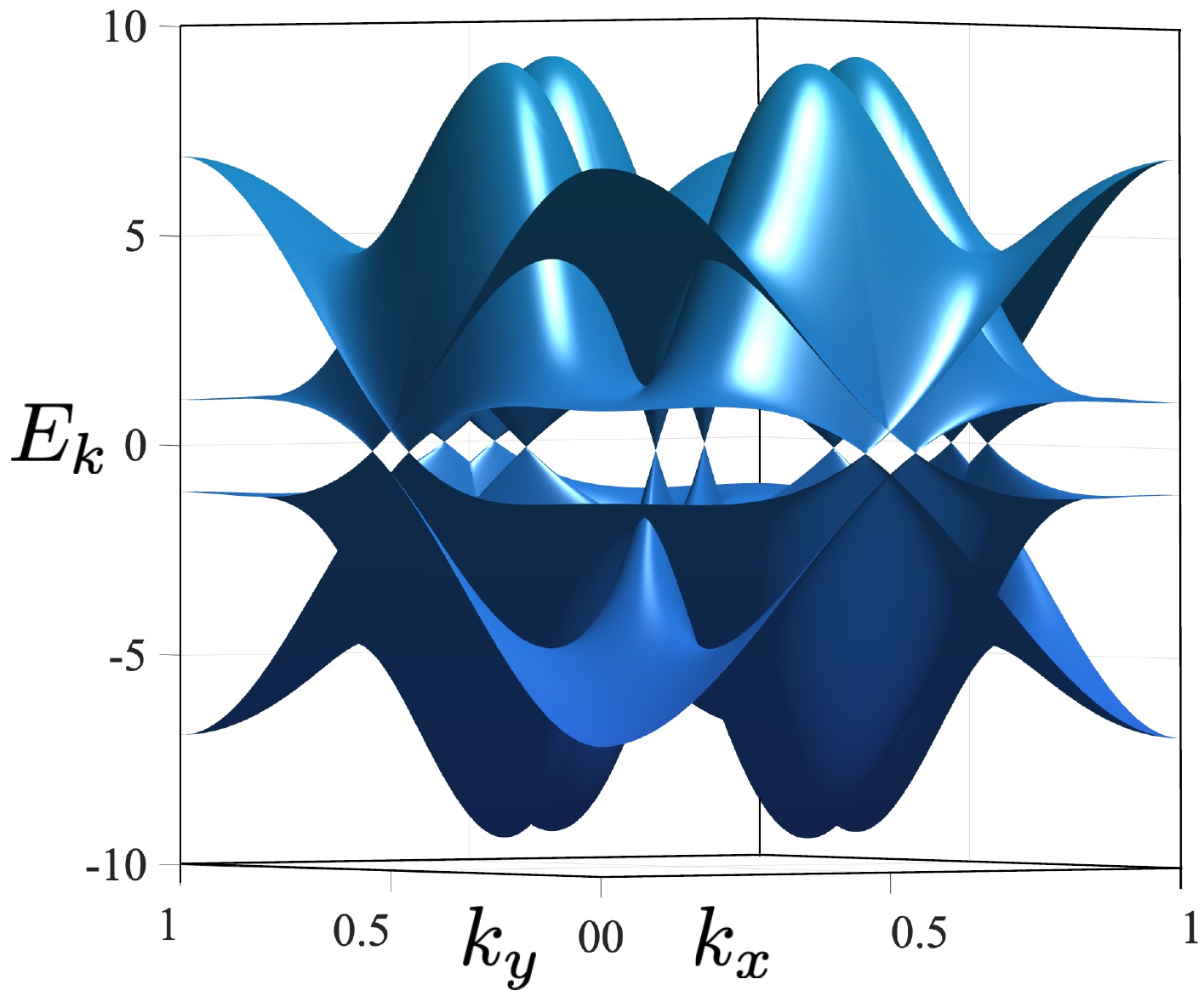}\label{p4gmBand_b}} 
 \subfigure[ zero modes in (c)]{\includegraphics[scale=0.19]{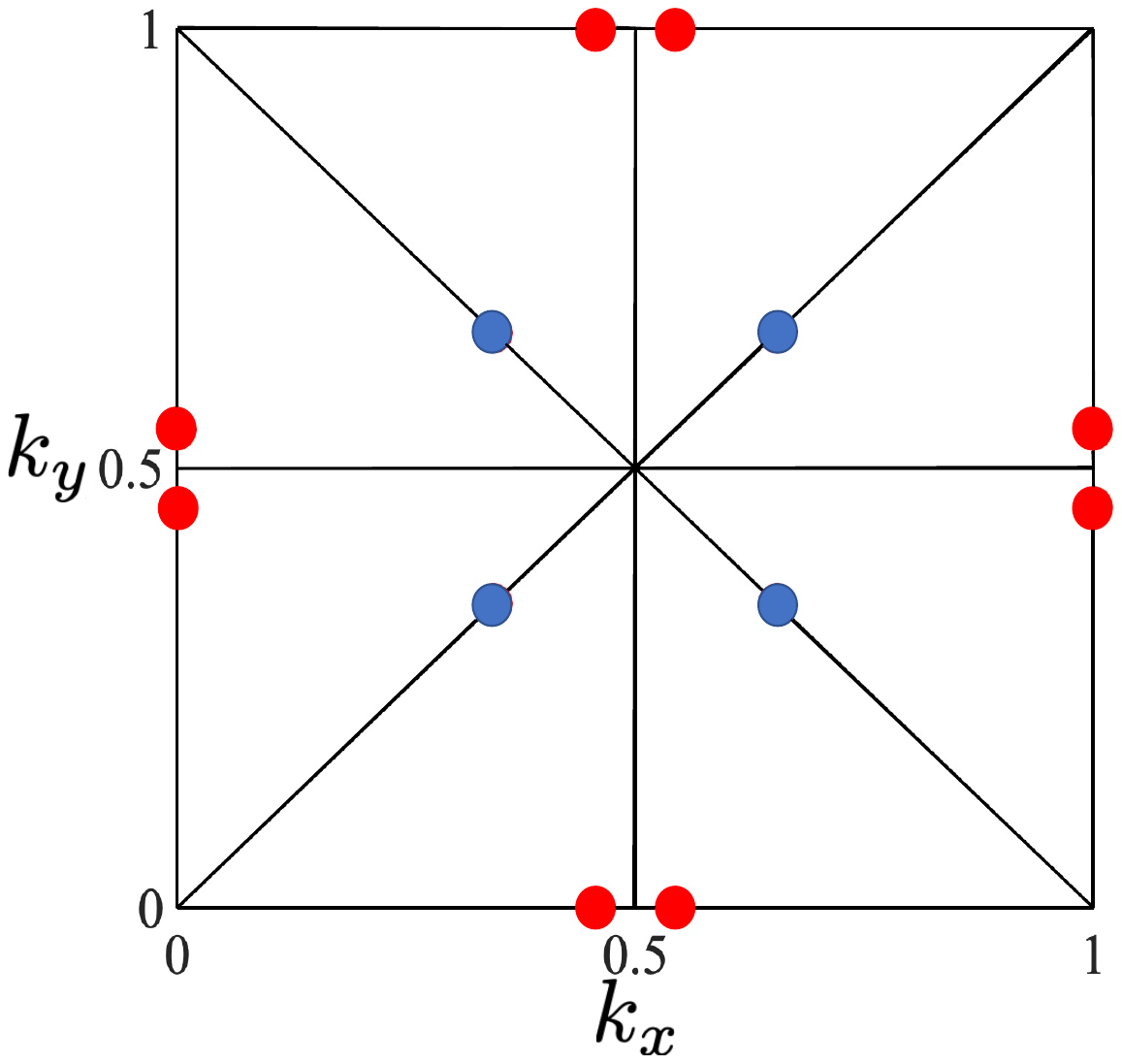}\label{p4gmBand_b1}}
\caption{(a) The band of the Hamiltonian $H_0$ in (\ref{p4gmHwosoc}) with $t=1,\Delta_2=1.5$; (b) the blue and red dots illustrate the positions of the zero modes in (a); (c)The Band structure of the Hamiltonian $H_0+H_1$ with  $t=1,\Delta_2=1.5,\lambda=0.5,t_0=0.6$; (d) the blue and red dots illustrate the positions of the zero modes in (c).} \label{p4gmBand}
\end{figure}

To illustrate the above results, we consider a lattice model with inter-sublattice hopping and pairing terms,
\beq\label{p4gmHwosoc}
H_0 &=&{t\over2}\sum_{\langle i_\zeta,j_{\bar\zeta}\rangle} c_{i_\zeta\sigma }^\dag c_{j_{\bar\zeta}\sigma}  
\notag\\
&&+\Delta_2 \!\!\sum_{\langle\langle i_\zeta ,j_{\bar\zeta}\rangle\rangle}\!\! e^{\mi\theta_{i, j}} (c_{i_\zeta\up }^\dag c_{j_{\bar\zeta}\dn}^\dag - c_{i_\zeta\dn }^\dag c_{j_{\bar\zeta}\up}^\dag )  + {\rm h.c.}
\eeq
where $\zeta=A,B$ and $\bar\zeta\neq\zeta$.
The phase $e^{\mi\theta_{i, j}} =\pm1$ of the pairing term on the bond $(ij)$ alternatively equals to $1,-1$, as shown in Fig.\ref{p4gmmodel}. After Fourier transformation, the Hamiltonian can be written in the bases $\psi_k^\dag=[c_{kA\up}^\dag, c_{kB\up}^\dag, c_{-kA\dn}, c_{-kB\dn}]$.  The band structure with parameters $t=1, \Delta_2=1.5$ is shown in Figs.\ref{p4gmBand_a}\!\! \&\!\! \ref{p4gmBand_a1}, where the 2-fold degeneracy of the bands are due to the anti-unitary symmetry $ {\C}\mathcal I \T'$ with $\mathcal I={C}_4^2$, $({ {\C}\mathcal I\T'})^2 =-1$. In the following, we analysis the protection of zero modes.

At the $X({1\over2},0)$ and $Y(0,{1\over2})$ points, the little co-group $F_f^c(k)/Z_2^f =  \mathscr C_{2v}\times Z_2^{\C} \times Z_2^\mathcal{T}\times Z_2^\mathcal{P}$ is represented as
\Beq
&&\!\!\!D(C_2)=I_2\otimes\nu_z, \ D(\mathcal M'_m)= \tau_z\otimes \nu_y,\ D(\T')K=I_4K, \\
&&\!\!\! D(\mathcal P)K=\tau_y\otimes I_2K,\ D(\C)=-\mi\tau_y\otimes\nu_z, 
\Eeq
with $m=x,y$ and $\nu_{x,y,z}$ the Pauli matrices acting on the $A,B$ sub-lattices. This rep corresponds to a set of 4-fold irreducible zero modes.  Around the $X,Y$ nodal points, the effective $k\cdot p$ Hamiltonian is given by 
$$H_{\rm eff1}= a\delta k_1 \tau_x\otimes \nu_y + b\delta k_2 \tau_z\otimes\nu_y,$$ 
where $a,b$ are constants, $ k_1= k_y, k_2=k_x$ for the X point and $k_1=k_x, k_2=k_y$ for the Y point. 
Furthermore, the response to the external fields can also be derived based on the theory discussed in Sec.\ref{sec:kdotpmodel}. For instance, the response to the temperature gradient $\pmb \nabla T$ at the $Y$ point is linear, which is given by 
\Beq
H_{\pmb \nabla T}=a'\partial_x T  \tau_z\otimes\nu_x + b'\partial_y T  \tau_x\otimes\nu_x. \notag
\Eeq


At the $M({1\over2},{1\over2})$ point, the little co-group $F_f^c(M)/Z_2^f =  \mathscr C_{4v}\times Z_2^{\C} \times Z_2^\mathcal{T}\times Z_2^\mathcal{P}$ is represented as 
\Beq
&&D(C_4)=I_2\otimes \nu_z,\ D(\mathcal M'_x)= \tau_z\otimes\nu_y,\ D(\T')K= I_4K, \\
&&D(\mathcal P)K= \mi\tau_y\otimes I_2 K,\ D( \C)=-\mi\tau_y\otimes\nu_z.
\Eeq 
The dispersion is quadratic at the nodal point $M$ with the following $k\cdot p$ Hamiltonian 
$$
H_{\rm eff2}=c(\delta k_x^2-\delta k_y^2)\tau_x\otimes \nu_x + d\delta k_x \delta k_y\tau_z\otimes \nu_x.
$$ 
The lowest order response to the temperature gradient $\pmb \nabla T$ at the $M$ point is of second order, which is given by
$$
H_{\pmb\nabla T}=c'({\partial_x T}^2 - {\partial_y T}^2) \tau_x\otimes \nu_x  + d'{\partial_x T} {\partial_y} T  \tau_z\otimes \nu_x.
$$ 

No zero modes appear at the $\Gamma$ point with $F_f^c(\Gamma)/Z_2^f =  \mathscr C_{4v}\times Z_2^{\C} \times Z_2^\mathcal{T}\times Z_2^\mathcal{P}$. Compared with the M point,  $C_4$ and $\mathcal M'_x$ are represented differently as $D(C_4)=I_2\otimes I_2$ and $D(\mathcal M'_x)=\mi\tau_z\otimes\nu_x$, consequently the irrep at $\Gamma$ point correspond to irreducible nonzero modes.

Now we investigate the stability of the zero modes by adding symmetry breaking terms to the Hamiltonian. Interestingly, if one removes the mirror symmetry and preserves all the other symmetries by adding, for instance, the $\Delta_0$ perturbation term in Eq.(\ref{H''}), then the dispersion is qualitatively the same as Fig.\ref{p4gmBand_a} and the 4-fold degenerate zero modes at the X, Y and M points still remain. However, unlike the case $\Delta_0=0$ where the set of 4-fold $\C$-zero modes is irreducible, when $\Delta_0\neq0$ the 4-fold zero modes are reducible minimal $\C$-zero modes. 

On the other hand, if one breaks the $ \C$ symmetry while preserves all the other symmetries by adding the intra-sublattice terms 
\beq\label{H'}
H_1 = \lambda\sum_{i,\zeta,\sigma}c^\dag_{i\zeta\sigma}c_{i\zeta\sigma} + t_0\sum_{\langle i,j\rangle,\zeta,\sigma} \big(c_{i\zeta\sigma}^\dag c_{j\zeta\sigma}^{} +{\rm h.c.}\big),
\eeq
the situations are quite different. From the band structure shown in Figs.\ref{p4gmBand_b}\!\! \&\!\! \ref{p4gmBand_b1} for the parameters $\lambda=0.5, t_0=0.6$, one can see that the 4-fold zero modes at the X (or Y) points split into two 2-fold zero modes which are separated along the $\hat x$ (or $\hat y$) direction. Since the zero modes are away from the high symmetry point, they are now resulting from level crossings protected by different quantum numbers of the $\mathcal M'_{y}$ (or $\mathcal M'_{x}$) symmetry. Similarly, the 4-fold zero modes at the M point split into two pairs of 2-fold zero modes, which are separated and spreading along the $(\hat x+\hat y)$- and $(\hat x-\hat y)$- directions respectively.  These zero modes are associated with level crossings protected by different quantum numbers of the $\mathcal M'_{x-y}$ or $\mathcal M'_{x+y}$ symmetry.

\begin{figure}[t]
  \centering
\subfigure[$\mathcal{I}\mathcal{T'}$ symmetric: band structure]{\includegraphics[scale=0.18]{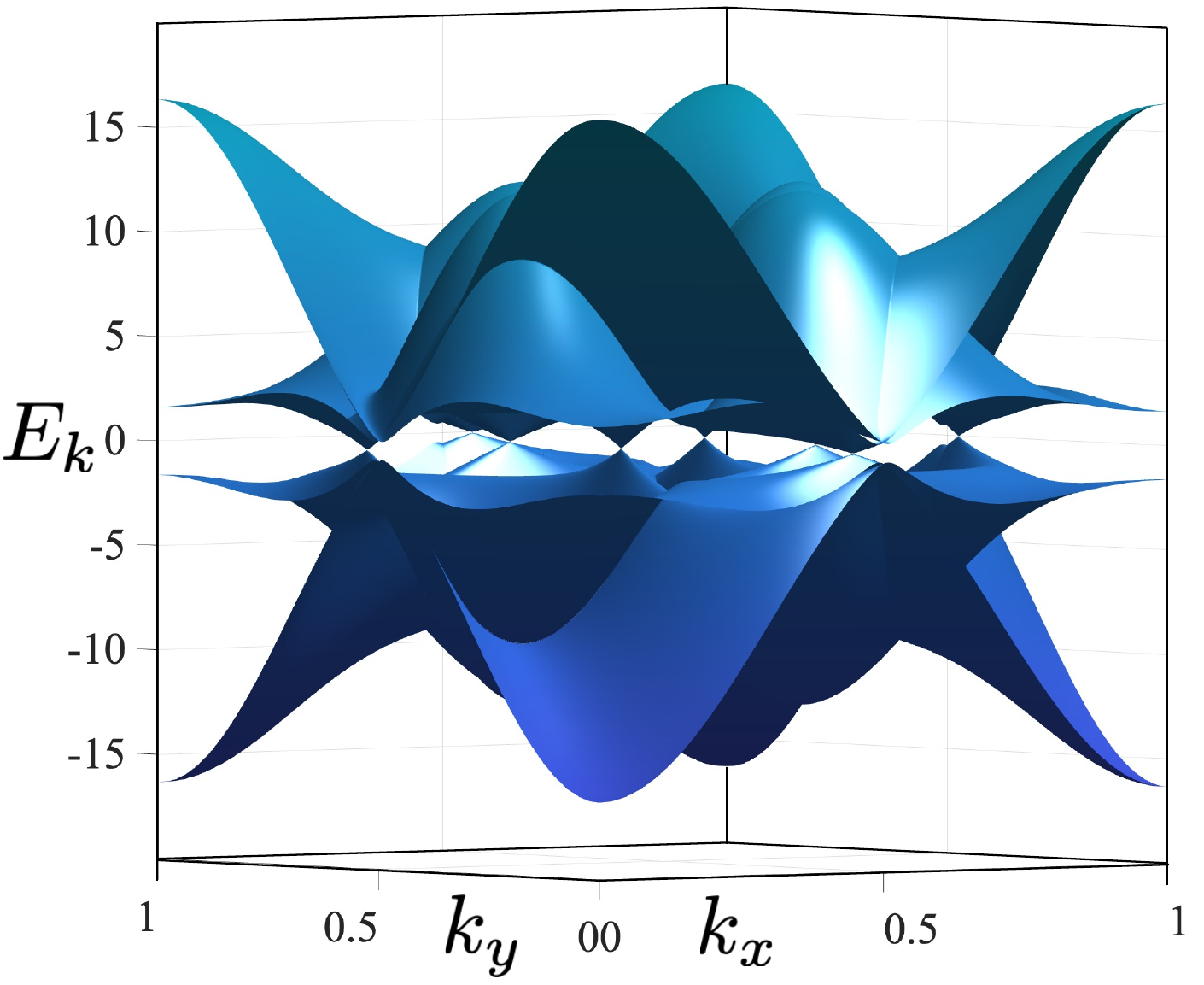}}
\subfigure[zero modes in (a)]{\includegraphics[scale=0.19]{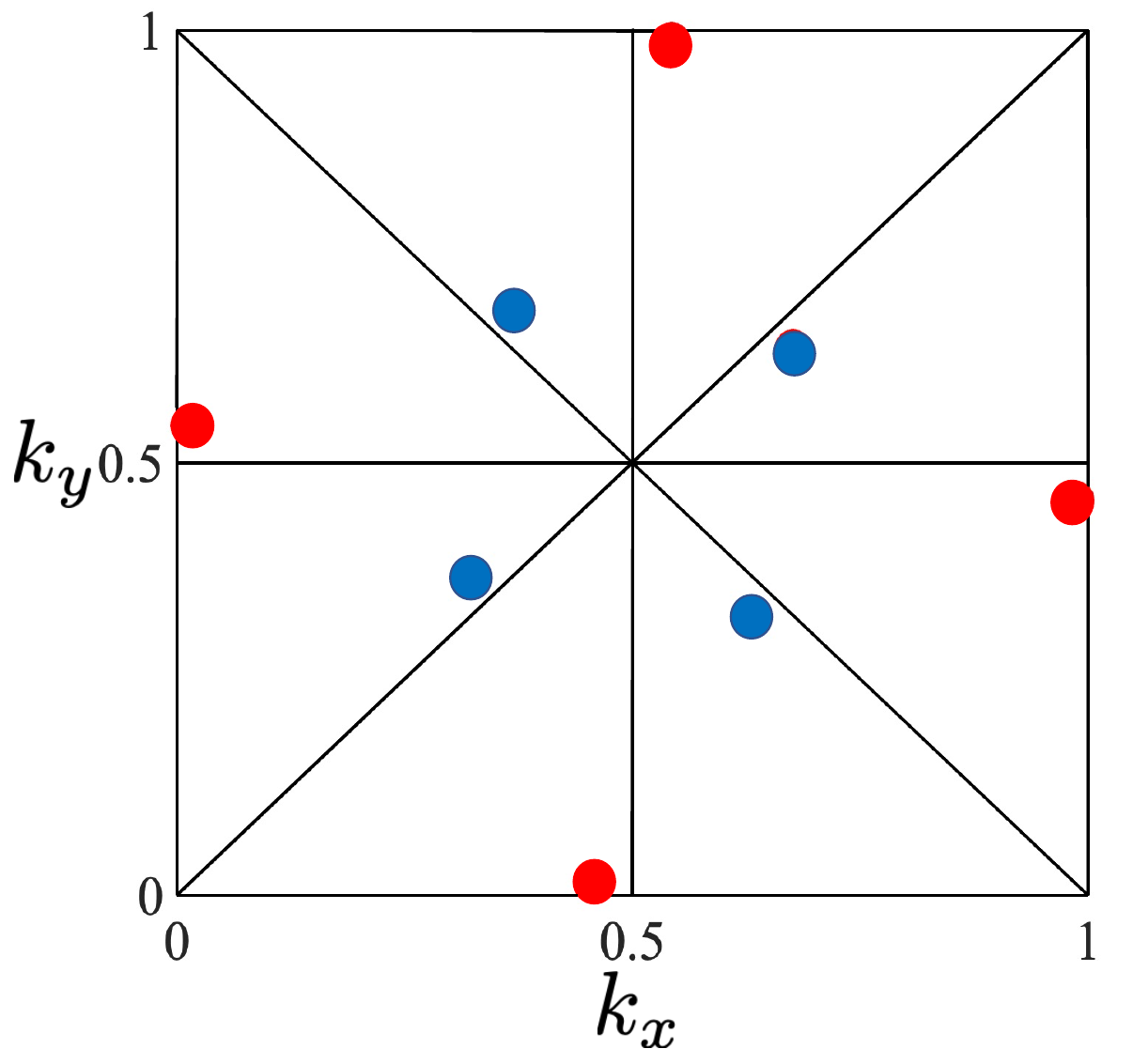}}
  \caption{(a) The Band structure of the Hamiltonian $H_0+H_1+H_2$ with $t=1,\Delta_2=1.5,\lambda=0.5,t_0=0.6,\Delta_0=2,\Delta_1=1.7$. (b) the blue and red dots illustrate the positions of zero modes in (a).}\label{p4gmDeltaBand}
\end{figure}

Zero modes still exist if one breaks both the mirror and $\C$ symmetries by further turning on the nearest neighbor inter-sublattice pairing and intra-sublattice pairing terms 
\beq\label{H''}
H_2 &=&  \Delta_0\!\!\!\sum_{\langle i_\zeta ,j_{\bar\zeta}\rangle}\!\! (c_{i_\zeta\up }^\dag c_{j_{\bar\zeta}\dn}^\dag - c_{i_\zeta\dn }^\dag c_{j_{\bar\zeta}\up}^\dag )  \notag\\
&&+\Delta_1 \!\! \sum_{\langle i,j\rangle,\zeta,\sigma} (c_{i\zeta\up}^\dag c^\dag_{j\zeta\dn}-c_{i\zeta\dn}^\dag c^\dag_{j\zeta\up}) +\ {\rm h. c.}
\eeq 
Here the $\Delta_0$ term breaks the the mirror symmetry while the $\Delta_1$ term breaks both the mirror and $\C$ symmetries, and both terms preserve the $C_4$ and $\T'$ symmetries. The Band structure of the Hamiltonian $H_0+H_1+H_2$ with $t=1,\Delta_2=1.5,\lambda=0.5,t_0=0.6,\Delta_0=2,\Delta_1=1.7$ are shown in Fig.\ref{p4gmDeltaBand} with the cones protected by the combined $\mathcal I\T'$ symmetry with $(\widehat {\mathcal I\T'})^2 = 1$. The $\mathcal I\T'$ symmetry gives rise to $\pi$-quantized Berry phase for loops in the BZ surrounding a single zero mode, which protects the cones from being solely gapped out. This mechanism makes the zero modes robust when the $\mathcal I\T'$ symmetry is preserved.

\begin{figure}[b]
  \centering
\includegraphics[scale=0.25]{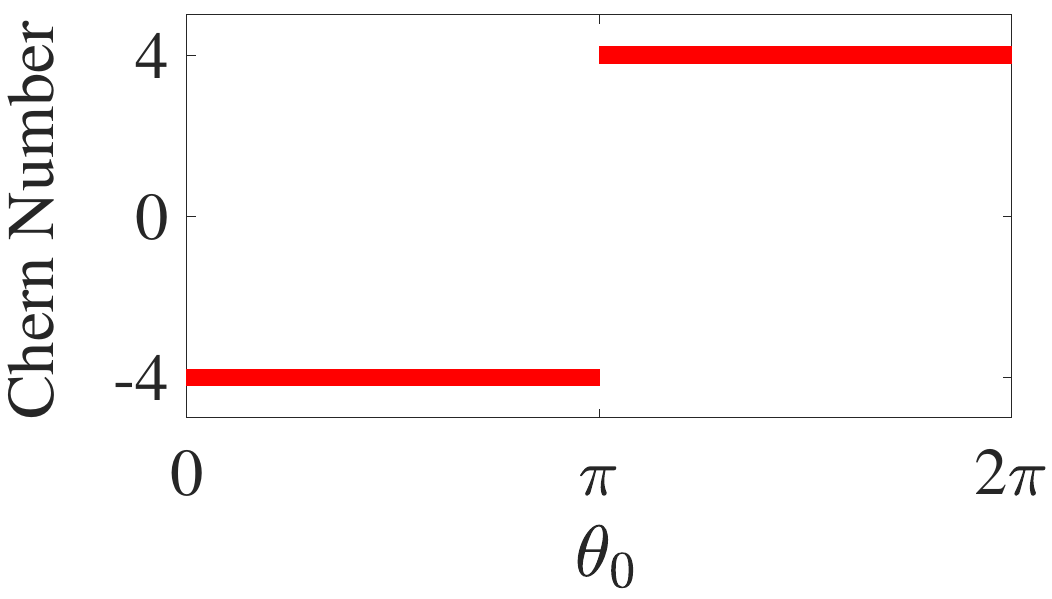}
  \caption{The Chern Number as a function of $\theta_0$, with $\theta_0$ the phase of $\tilde\Delta_0=\Delta_0 e^{\mi\theta_0}$. }\label{Chern_Number}
\end{figure}

Finally, when the pairing terms are complex, then the time reversal symmetry $\T'$ and $\mathcal I\T'$ are broken and the zero modes disappear. The fully gapped bands becomes a topological SC/SF with nonzero Chern number. Due to the unbroken $C_4$ symmetry, the Chern number is either $4$ or $-4$. As shown in Fig.\ref{Chern_Number}, if we keep $\Delta_{1,2}$ to be real but vary the phase of $\Delta_{0}$, then the Chern number can be changed from 4 to $-4$ or vice versa. The phase transition occurs at $\theta_0={\pi}$ ($\Delta_0$ is real again), where the gap closes at 8 nodal  points in the BZ. 

From the above model, we can see that the zero modes are quit robust even in the presence of symmetry breaking perturbations. When the effective charge conjugation symmetry $\C$ is preserved, the zero modes appear at high symmetry points. When the $\C$ is violated, like in most realist SCs/SFs, the zero modes do not disappear at once as long as certain protecting symmetries are still preserved. The $\C$ symmetry is of theoretical significance for BdG systems in the sense that it provides the physical origin of the zero modes when the $\C$ breaking terms are treated as perturbations. This provides a systematic way to investigate the physical properties of nodal SCs/SFs with gapless quasiparticle excitations.

\subsection{Triangular lattice (spin-${1\over2}$): $G_b=p3\times  {Z}_2^{\mathcal{T}}$} \label{LatticeModel:B}

The second example is a model for $p$-wave SC of the DIII class. We consider the magnetic  wallpaper group $G_b=p3\times Z_2^{\mathcal{T}}$ on triangular lattice, and place two orbitals $[c^\dag_{\up}, c^\dag_{\dn}]$ at the Wyckoff position $1a (0,0)$ with cite group $\mathscr C_{3}\times Z_2^{\mathcal{T}}$, as shown in Fig.~\ref{p3_wickoff_direction}. The bases $[c^\dag_{\up}, c^\dag_{\dn}]$ respectively carry angular momentum $j=\pm{3\over2}$ as a consequence of spin-orbit coupling, they reverse their sign under $C_3$ and form a Kramers doublet under $\T$.
\begin{figure}[t]
  \centering
\includegraphics[scale=0.25]{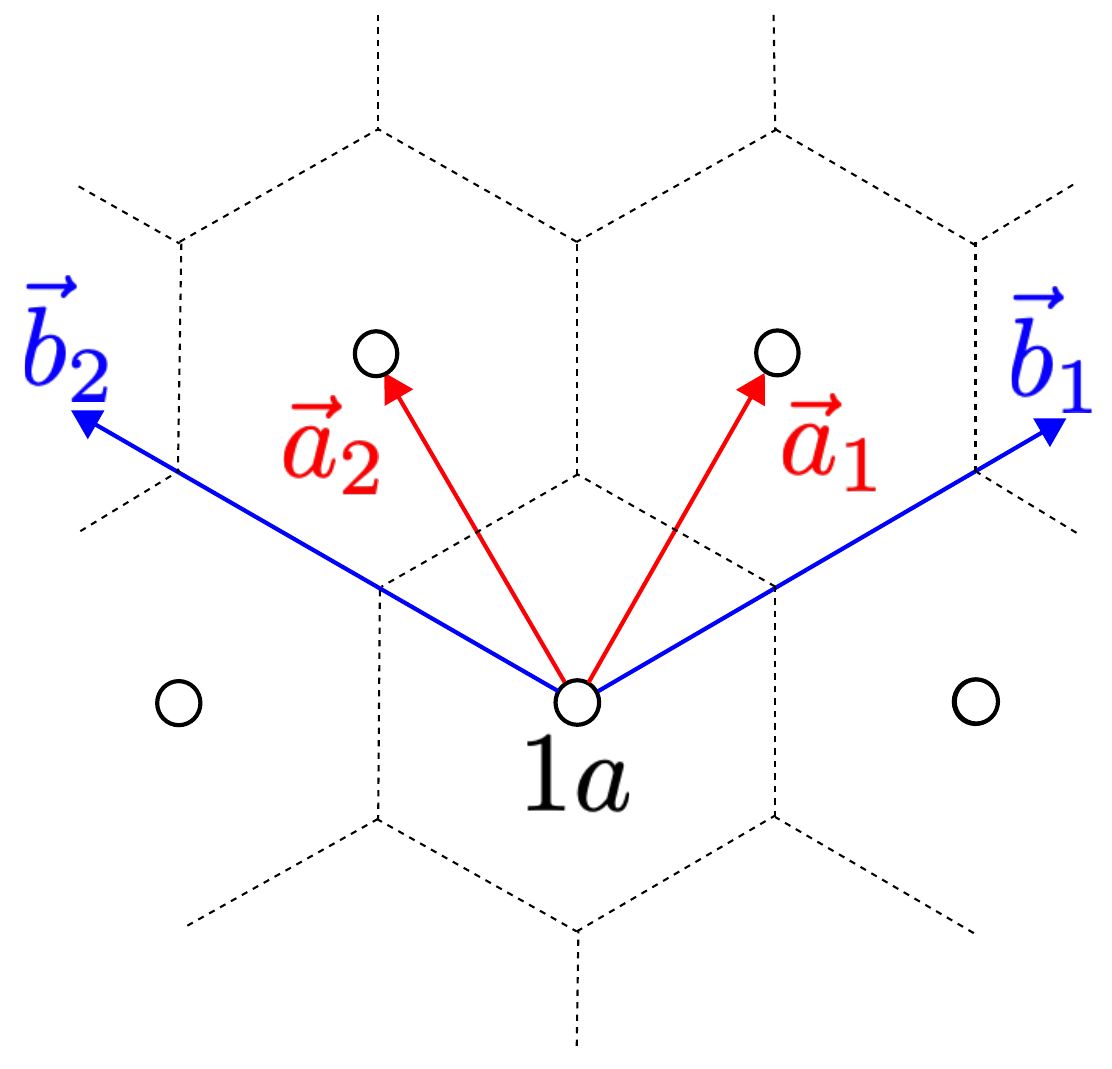}
\caption{Lattice basis vectors (red) and their dual vectors (blue) for the magnetic wallpaper group $G_b=p3\times Z_2^{\mathcal{T}}$. The hollow circles label the (maximal) $1a$ Wyckoff positions and the hexagons enclosed by the doted lines stand for unit cells.}\label{p3_wickoff_direction}
\end{figure}
In the complete Nambu bases $\Psi^\dag_i =[c_\up^\dag, c_\dn^\dag, c_\up, c_\dn]_i$, the symmorphic symmetry operations $C_3$ and $\T$ are diagonal in the particle-hole sector, 
\Beq
\hat C_{3} \Psi^\dag_i \hat C_{3}^{-1}= - \Psi^\dag_{C_3 i} I_4,\ \ \hat{\mathcal{T}}\Psi^\dag_i \hat{\mathcal{T}}^{-1}=\Psi^\dag_i I_2\otimes (-\mi\sigma_y)K.
\Eeq
With the particle-hole symmetry $\mathcal P$ and the effective charge-conjugation symmetry $\C =-e^{-\mi{\tilde\sigma_y\over2}\pi}e^{-\mi{\tilde\tau_y\over2}\pi}$ (see (\ref{SU2s0}) and (\ref{SU2c0})) which act as 
\Beq
\ \hat{\mathcal P} \Psi^\dag_i \hat{\mathcal P}^{-1}=[ c_\up, c_\dn, c^\dag_{\up},c^\dag_{\dn}]_iK,\ \ \widehat{\C} \Psi^\dag_i \widehat{\C}^{-1} = [c_\up, c_\dn, c_\up^\dag, c_\dn^\dag]_i,
\Eeq
one obtains the full symmetry group $F^c_f$. From the band representation, minimal $\mathcal{C}$-zero modes are found at the $\Gamma$ point and the three M points.  

To illustrate, we construct the following lattice model preserving the $G_b\times Z_2^{\C}$ symmetry ($Z_2^{\C}=\{E,{\C}\}$), 
\beq\label{TheHamiltonianP3T}
H_0&= &\mi t\sum_{\langle ij\rangle} \left(c_{i\up}^\dag c_{j\dn}+c_{i\dn}^\dag c_{j\up}\right) +\mi \Delta_1 \sum_{\langle ij\rangle} \left( c_{i\up}^\dag c_{j\up}^\dag - c_{i\dn}^\dag c_{j\dn}^\dag \right)  \notag\\
&&+\mi \Delta_2 \sum_{\langle\langle i'j'\rangle\rangle} \left( c_{i'\up}^\dag c_{j'\dn}^\dag +c_{i'\dn}^\dag c_{j'\up}^\dag \right) + {\rm h.c.} 
\eeq
where the $\langle ij\rangle$ represents the nearest neighbors and $\langle\langle i'j'\rangle\rangle $ represents the next nearest neighbors. The Hamiltonian can be diagonalized in the full Nambu bases $\Psi^\dag_k=[c_{k\up}^\dag, c_{k\dn}^\dag, c_{-k\up}, c_{-k\dn}]$. The band structure with parameters $t=0.7,\Delta_1=0.5,\Delta_2=0.4$ is shown in Fig.\ref{p3Z2TBand}, which verifies the existence of zero modes.

At the $\Gamma$ point $(0,0)$, the little co-group $F_f^c(\Gamma)/Z_2^f=\mathscr C_{3}\times Z_2^{\mathcal{T}}\times Z_2^{\mathcal{P}}\times Z_2^{\C}$ is represented as 
\Beq
&&D(C_{3}^z)=-I_4,\ \ \ D(\mathcal{T})K=-\mi I_2\otimes \sigma_yK,\\
&&D(\mathcal{P})K=\tau_x\otimes I_2K,\ D(\C) =\tau_x\otimes I_2,
\Eeq
in the bases $\Psi^\dag_k$. Without $\C$, the little co-group becomes $F_f(k)/Z_2^f=\mathscr C_{3}\times Z_2^{\mathcal{T}}\times Z_2^{\mathcal{P}}$. Defining $\p =\T\mathcal P$, then the above rep of $F_f(k)$ corresponds to a set of RMNZM of the case (6) in the Tab.\ref{tab:RMNZM}. When $\C$ is turned on, then the nonzero modes are turned into reducible minimal $\C$-zero modes. 
Furthermore, the $k\cdot p$ theory indicates that the lowest order dispersion is cubic (see Sec.\ref{sec:k.p}), and around the $\Gamma$ point, the third-order $k\cdot p$ Hamiltonian can be written as
\begin{align}
H_{\rm eff}=h_1(\delta k_x^3-3\delta k_x\delta k_y^2)+h_2(\delta k_y^3-3\delta k_y\delta k_x^2)
\end{align} 
where $k_x, k_y$ are coordinates in the orthogonal frame, and the matrices $h_1= (aI_2\otimes\nu_z+bI_2\otimes\nu_x+c\tau_x\otimes\nu_z+d\tau_x\otimes\nu_x), h_2 =(a'I_2\otimes\nu_z+b'I_2\otimes\nu_x+c'\tau_x\otimes\nu_z+d'\tau_x\otimes\nu_x)$ contain 8 real parameters $a,b,c,d, a',b',c',d'$.

\begin{figure}[t]
  \centering
\subfigure[band structure with ${\mathcal{C}}$]{\includegraphics[scale=0.14]{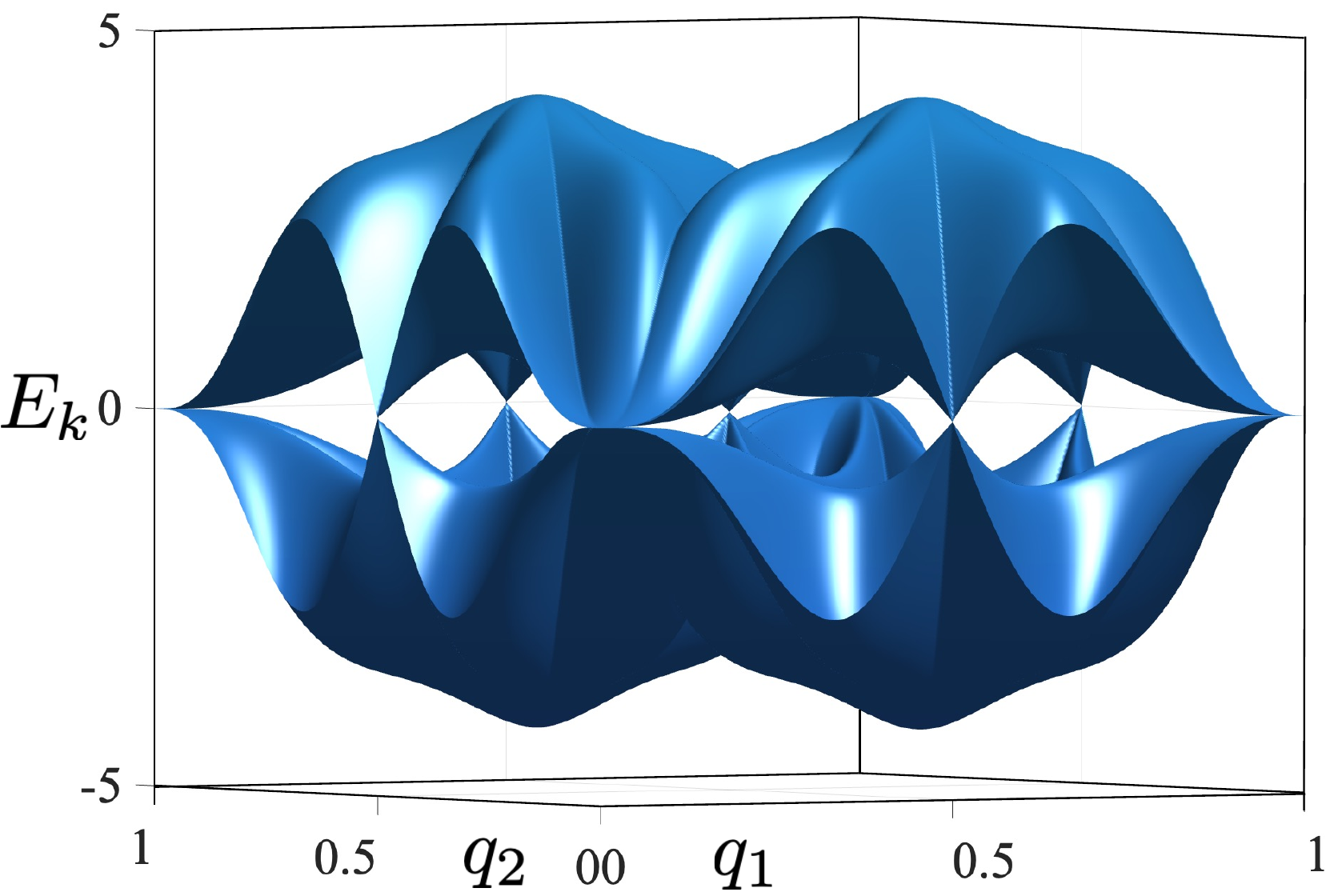}}
\subfigure[the zero modes in (a) ]{\includegraphics[scale=0.16]{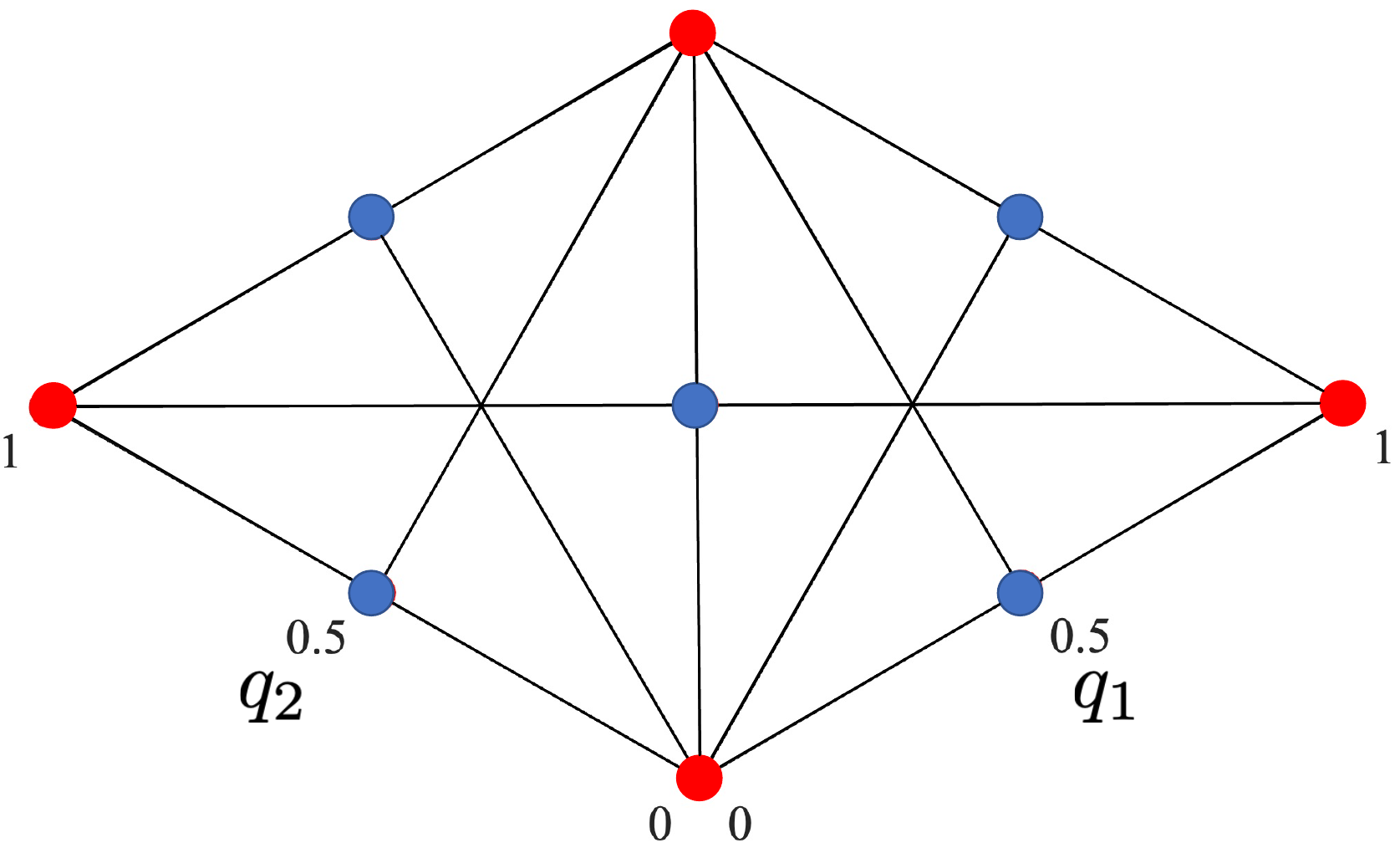}}
\caption{(a)The Band structure of the Hamiltonian (\ref{TheHamiltonianP3T}); (b) the blue and red dots illustrate the positions of the zero modes.}\label{p3Z2TBand}
\end{figure}

At the $C_3$ related three $M$ points $(0, {1\over 2}), ({1\over 2},0), ({1\over 2}, {1\over 2})$, the little co-group $F_f^c(M)/Z_2^f=Z_2^{\mathcal{T}}\times Z_2^{\mathcal{P}}\times Z_2^{\C}$ is represented as 
\Beq
&&D(\T)K=-\mi I_2\otimes \sigma_yK, \ \ D(\mathcal P)K=\tau_x\otimes I_2K,\\
&&D(\C)=\tau_x\otimes I_2.
\Eeq
For the same reason as the $\Gamma$ point, the $\C$ symmetry turns the RMNZM into reducible minimal $\C$-zero modes. 
The $k\cdot p$ theory indicates that the energy dispersion is linear at the $M$ points (see Sec.\ref{sec:k.p}), for instance, the $k\cdot p$ Hamiltonian for the $(0,{1\over2})$ point is given by
\begin{align}
H_{\rm eff}=h_1\delta k_x+h_2\delta k_y
\end{align} 
where $h_1= (e I_2\otimes\nu_z+ f I_2\otimes\nu_x + g \tau_x\otimes\nu_z + h \tau_x\otimes\nu_x), h_2 =(e'I_2\otimes\nu_z + f' I_2\otimes\nu_x + g'\tau_x\otimes\nu_z + h'\tau_x\otimes\nu_x)$ and $e,f,g,h, e',f',g',h'$ are real parameters.

If one breaks the ${\C}$ symmetry (by preserving the time reversal symmetry $\T$) by including the real hopping terms $H_1=t_0\sum_{\langle i,j\rangle, \sigma} (c_{i\sigma}^\dag c_{j\sigma} + {\rm h.c.})$, or the real pairing terms $H_2=\Delta_0\sum_{\langle i,j\rangle, \sigma} [(c_{i\up}^\dag c_{j\dn}^\dag-c_{i\dn}^\dag c_{j\up}^\dag) + {\rm h.c.}]$, or the chemical potential term $\lambda \sum_{i\sigma}c_{i\sigma}^\dag c_{i\sigma}$, then the zero modes disappear immediately giving rise to a fully gapped SC/SF state.

\subsection{Triangular lattice (spinless): $G_b=p6m\times Z_2^{\mathcal{T}}$} \label{sec:spinless}

The third model is about a SC for spinless fermions of the BDI class. We consider a triangular lattice with the magnetic wallpaper group $G_b=p6m\times Z_2^{\mathcal{T}}$.
\begin{figure}[t]
  \centering
\includegraphics[scale=0.25]{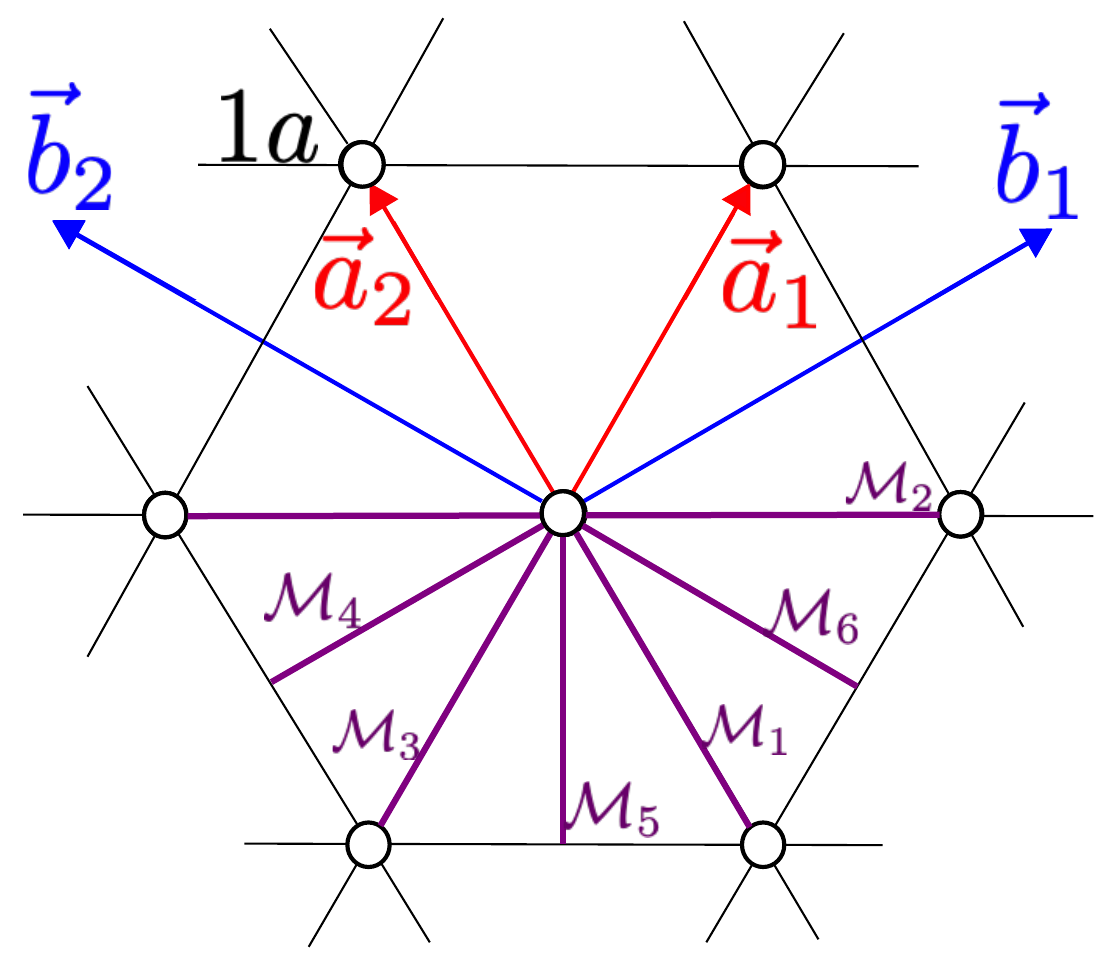}
\caption{Lattice basis vectors (red) and their dual vectors (blue) for the magnetic wallpaper group $G_b=p6m\times Z_2^{\mathcal{T}}$. The (maximal) $1a$ Wyckoff positions are indicated by hollow circles, and all mirror plans are indicated by purple lines.}\label{p6mlattice}
\end{figure}
We put spinless fermonic bases $[c^\dag, c]$ on the Wyckoff position $1a (0,0)$ which carry projective rep of the site group $C_{3v}\times \mathscr C_{\mathcal{I}}\times Z_2^{\mathcal{T}}$, see Fig.\ref{p6mlattice}. 
In the fermionic symmetry group $G_f$, the inversion operation $\mathcal I$ is associated with a charge operation, namely $\mathcal I'=(e^{\mi{{\tau}_z\over2} \pi} || \mathcal I)$. The Nambu bases $\Psi^\dag_i = [c^\dag, c]_i$ vary under the action of the symmetry group as
\Beq
&&\hat C_3 \Psi^\dag_i \hat C_3^{-1} = \Psi^\dag_{C_3 i},\ \ \hat{ \mathcal{M}}_{m} \Psi^\dag_i \hat{\mathcal{M}}_{m}^{-1} = \Psi^\dag_{\mathcal{M}_{m}i},\\ 
&&\hat{\mathcal I}' \Psi^\dag_i \hat{\mathcal I}^{'-1} = \Psi^\dag_{\mathcal I i} (\mi{\tau}_z) = [\mi c^\dag, -\mi c]_{\mathcal I i},\ \ \mathcal{T}\Psi^\dag_i\mathcal{T}^{-1}=\Psi^\dag_i,
\Eeq
with $m=4,5,6$. Besides the particle-hole operation $\mathcal P \Psi^\dag_i \mathcal P^{-1} = \Psi^\dag_{i} {\tau}_x K = [c, c^\dag]_{i}K$, we further consider the following charge conjugation $\C$ with 
$$
\hat \C \Psi^\dag_i \hat \C^{-1} =\Psi^\dag_i {\tau}_y= [\mi c, -\mi c^\dag]_i.
$$ 
The band representation of the above group $F^c_f$ indicates that the high symmetry lines on the mirror planes $\mathcal{M}_{1,2,3}$ (i.e. the lines linking $\Gamma$ and $K, K'$ points) are nodal lines of irreducible $\C$-zero modes. 

\begin{figure}[b]
  \centering
  \subfigure[band structure with $\C$ symmetry]
  {\includegraphics[scale=0.14]{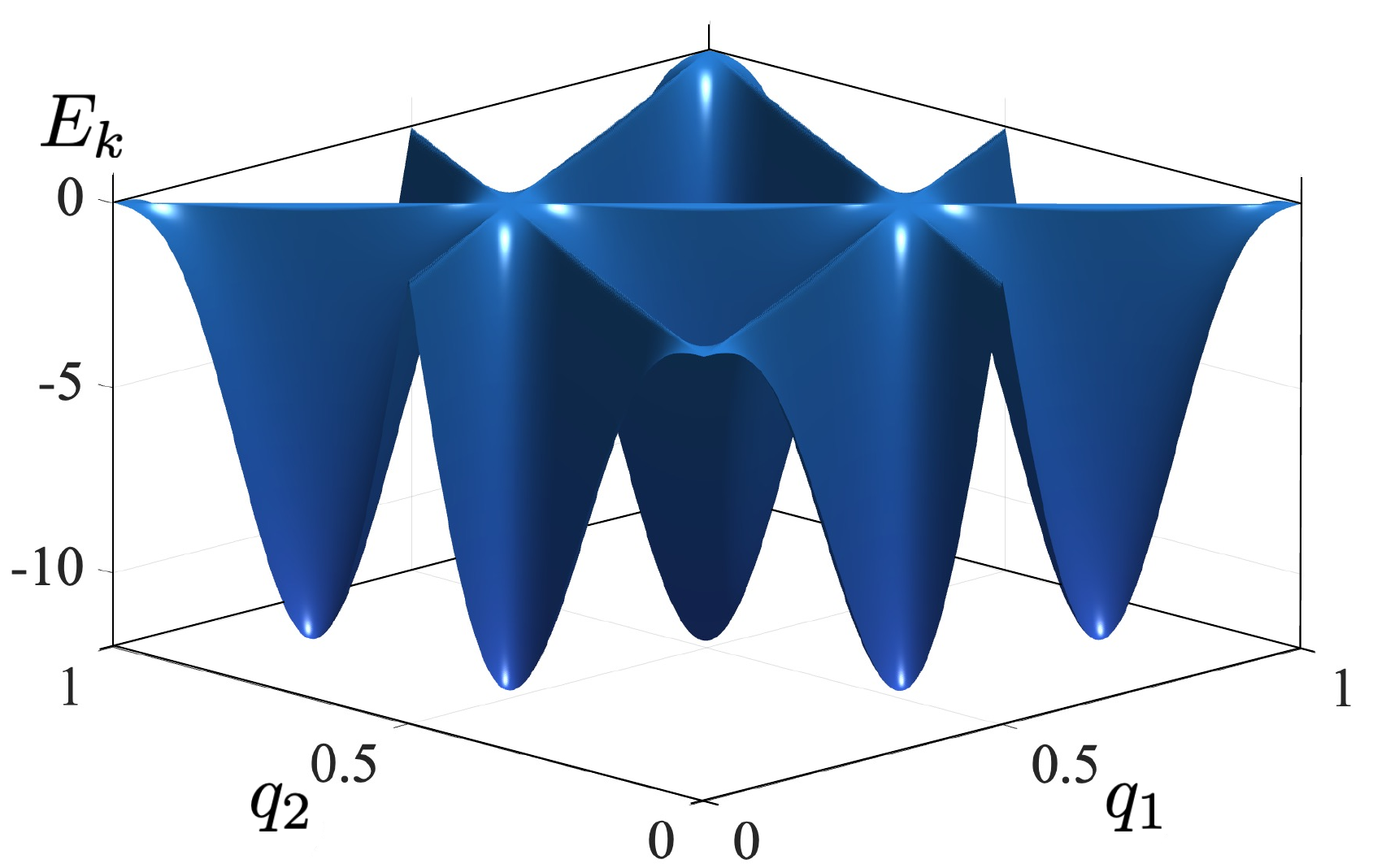}}
  \subfigure[the nodal lines in (a)]{\includegraphics[scale=0.155]{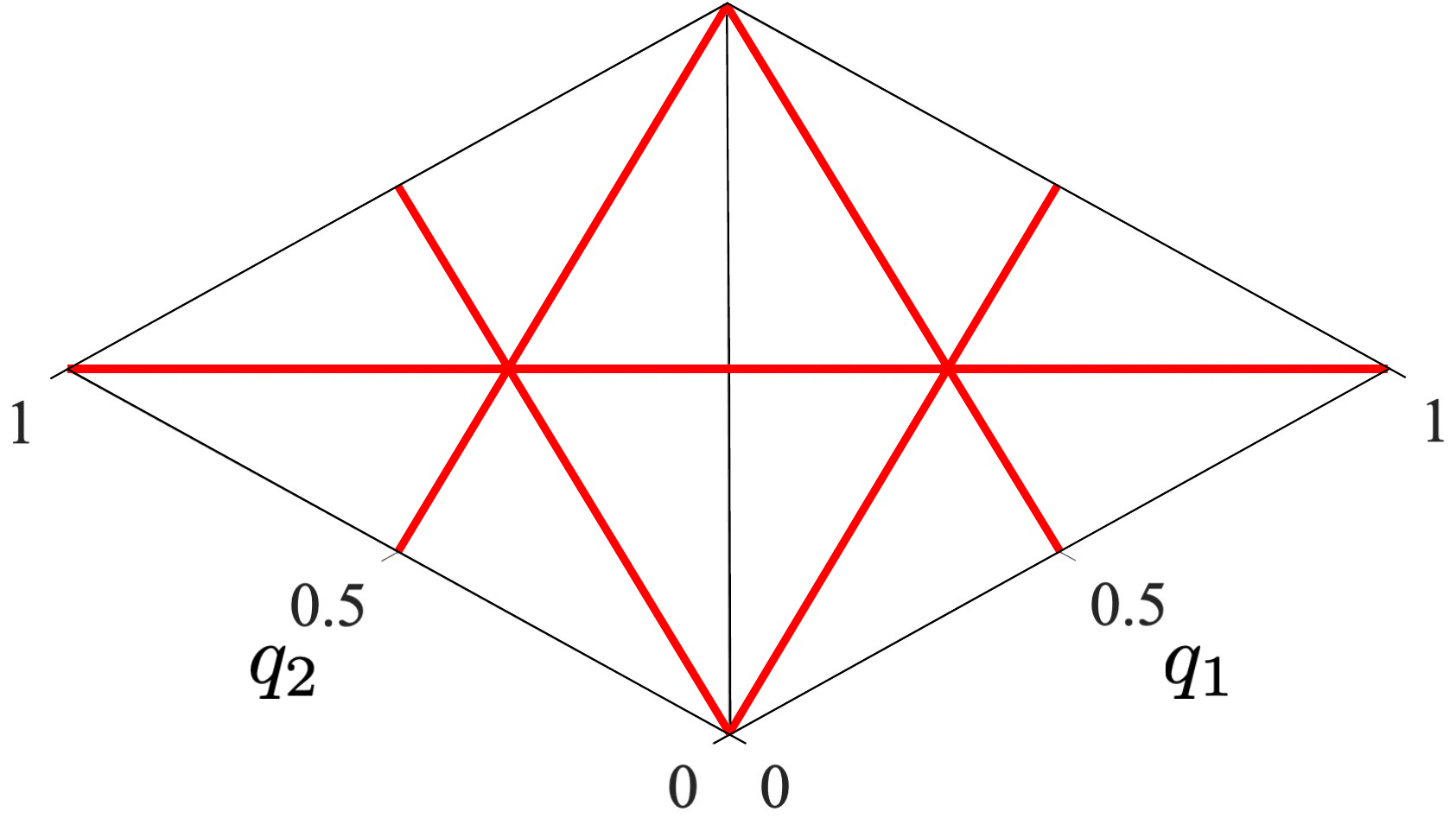}}
  \caption{The nodal-line Band structure of the Hamiltonian (\ref{Ham_p3m1wosoc}) having $\C$ symmetry.}\label{p3m1wosoc_band}
\end{figure}

The simplest lattice model preserving the above symmetry is a pure pairing model,
\begin{small}
\begin{align}\label{Ham_p3m1wosoc}
H_0=\Delta\sum_{\langle\langle ij \rangle\rangle} c_i^\dag c_j^\dag + {\rm h.c.}
\end{align}\end{small}where $\langle\langle ij \rangle\rangle$ represents next nearest neighbors (the nearest neighbor pairing terms break the $\mathcal{M}_{4,5,6}$ symmetries).

The band structure with parameter $\Delta=2$ is shown in Fig.\ref{p3m1wosoc_band}, 
where the three nodal liens (locating on the mirror planes $\mathcal{M}_{1,2,3}$) have the symmetry group $F_f^c(k)/Z_2^f=Z_2^{\C}\times Z_2^{\mathcal I'\T}\times Z_2^{\mathcal{M}\T} \times Z_2^{\T\mathcal P}$. The symmetry group $F_f^c(k)$ is represented as 
\Beq
&&D(\mathcal I'\T)K = \mi {\tau}_zK,\ D(\T\mathcal P) ={\tau}_x,\\ 
&&D(\C)={\tau}_y,\ D(\mathcal{M}_{m}\T)K=I_2K,  {\  m=4,5,6}.
\Eeq 
It can be seen that the subgroup $Z_2^{\C}\times Z_2^{\mathcal{M}\T}\times Z_2^{\mathcal P}$ can protect the nodal lines. The $k\cdot p$ Hamiltonian at the nodal lines are given as 
\Beq
H_{\rm eff 3} = a \delta k_v {\tau}_y,\ \ H_{\rm eff 2} = -a \delta k_x{\tau}_y, \ \ 
H_{\rm eff 1} = a\delta k_w{\tau}_y.
\Eeq
where $k_x,k_y$ are coordinates in the orthogonal frame and $\delta k_v={1\over2}\delta k_x+{\sqrt3\over2}\delta k_y, \delta k_w={1\over2}\delta k_x-{\sqrt3\over2}\delta k_y$.

\begin{figure}[t]
  \centering
\subfigure[band structure with $\mathcal M'_{1,2,3}$]{ \includegraphics[scale=0.14]{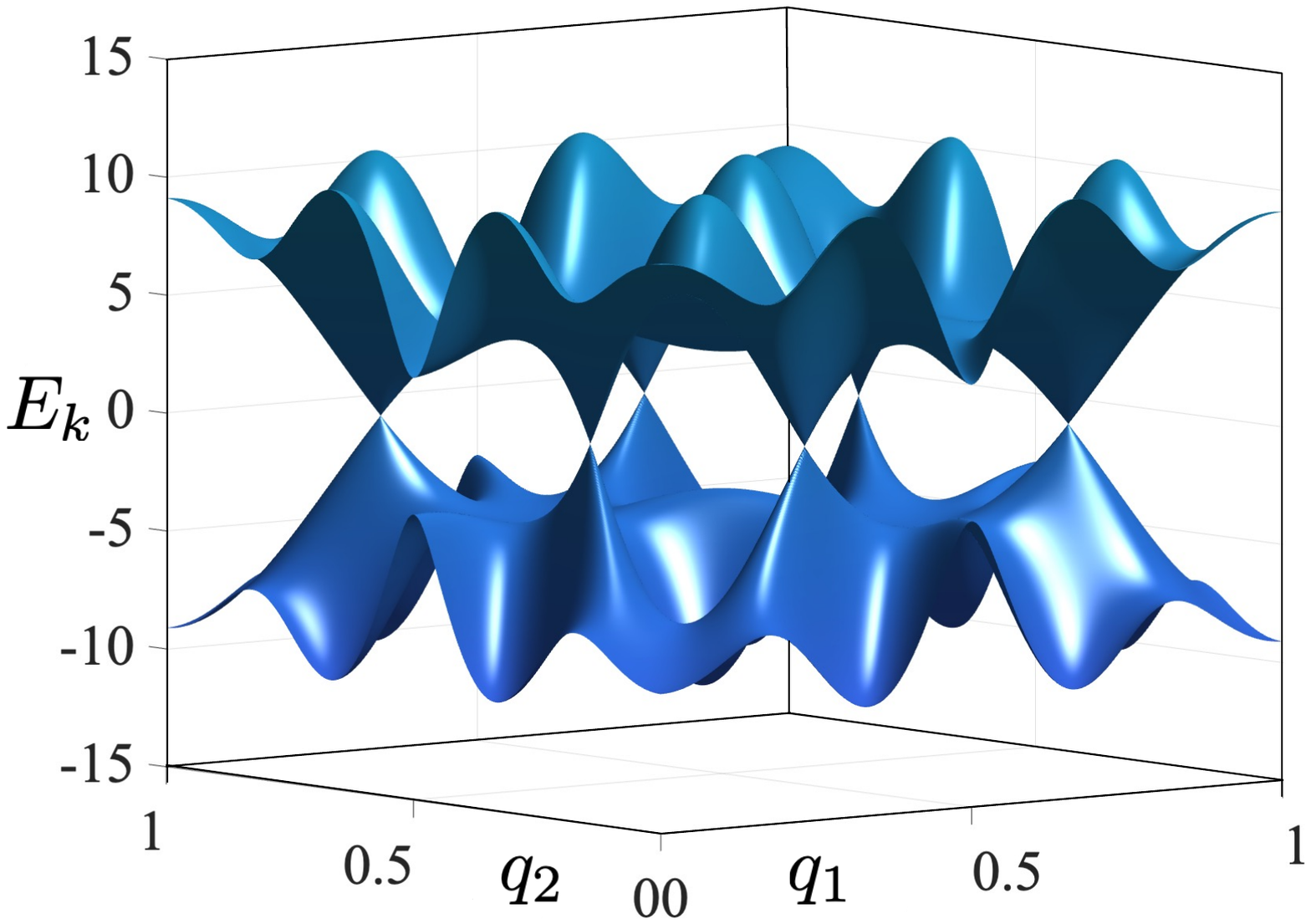}\label{Gaphop:p3m1Z2TBand_a}}
\subfigure[zero modes in (a)]{ \includegraphics[scale=0.15]{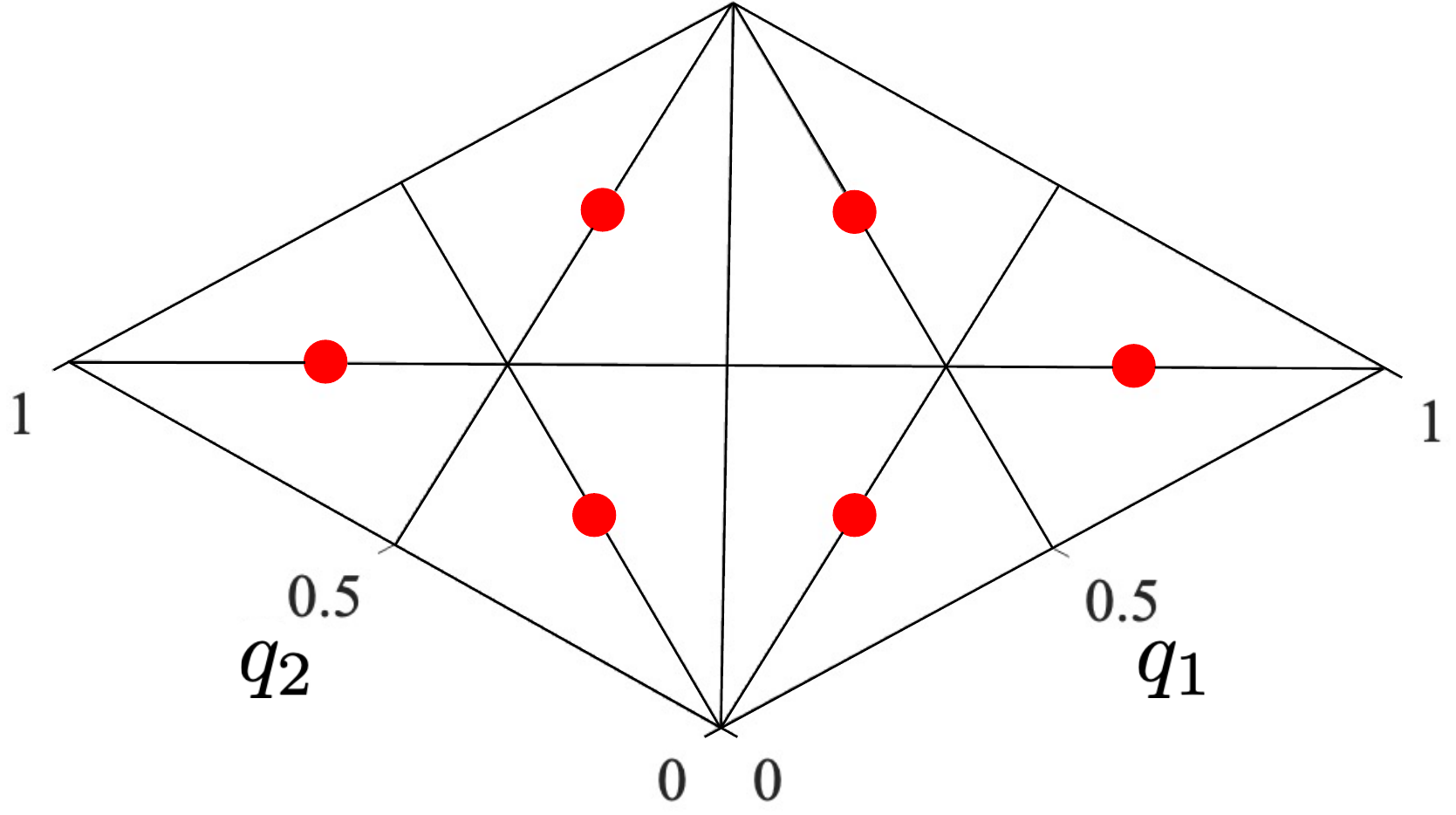}\label{Gaphop:p3m1Z2TBand_ap}}
\subfigure[band structure without $\mathcal M'_{1,2,3}$]{\includegraphics[scale=0.14]{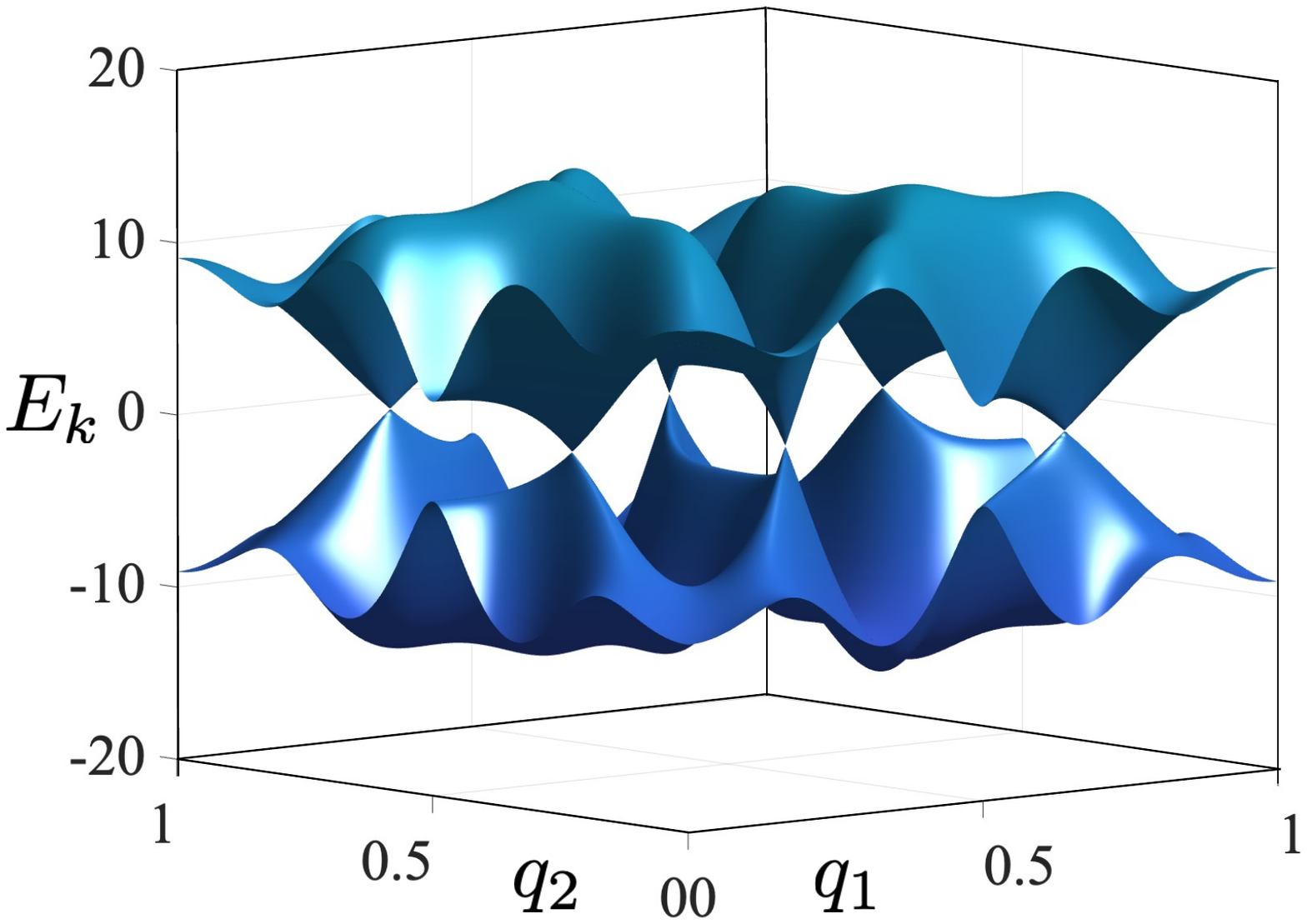} \label{Gaphop:p3m1Z2TBand_b}}
\subfigure[zero modes in (c)]{ \includegraphics[scale=0.15]{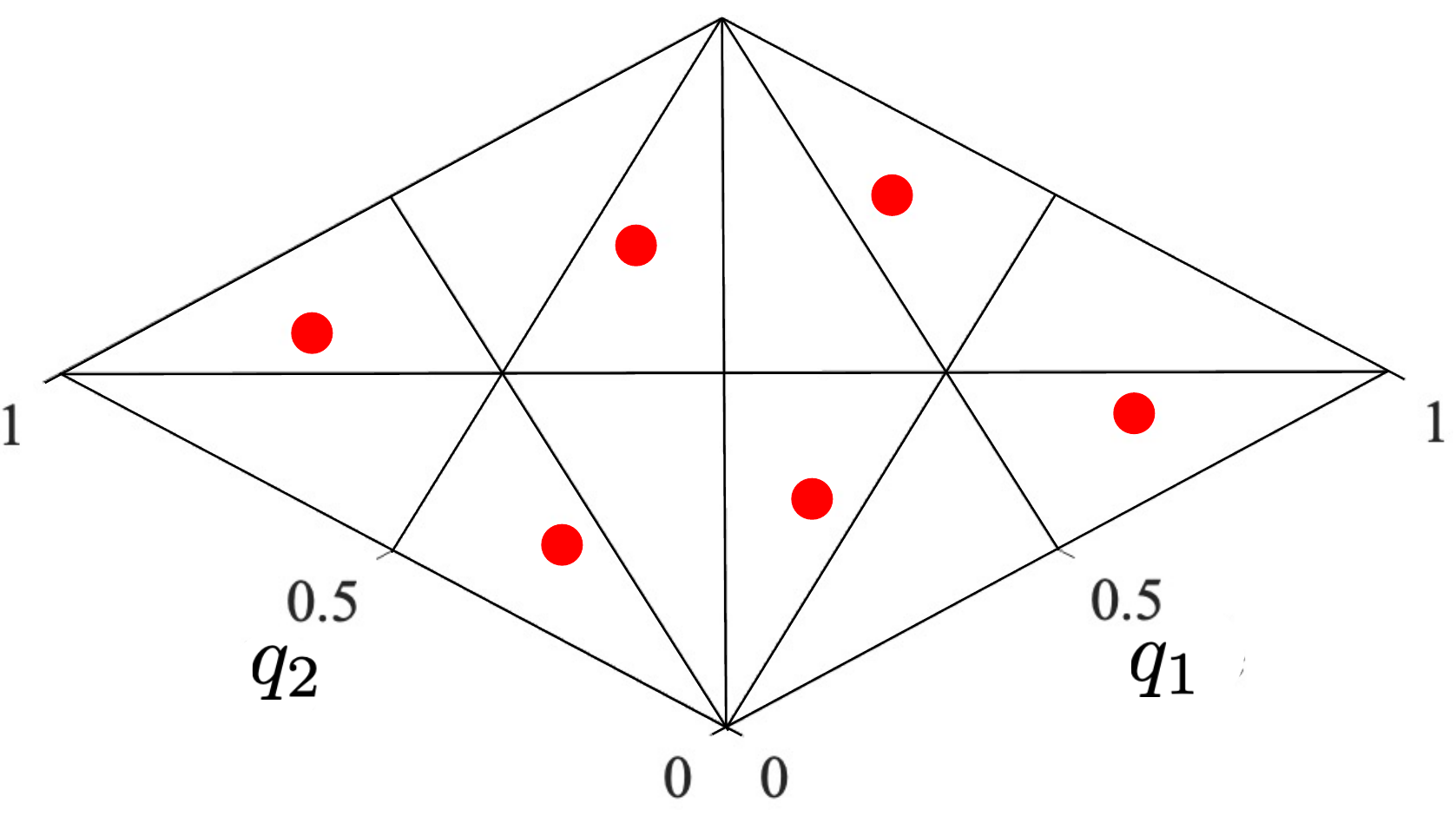}\label{Gaphop:p3m1Z2TBandBM}}
  \caption{(a)Band structure of the Hamiltonian (\ref{Ham_p3m1wosoc}) with parameters $t=1.5,\lambda=0.1,\Delta=2$. The mirror symmetries $\mathcal M'_{1,2,3}$ are preserved; (b) the red dots illustrate the positions of the zero modes in (a);  (c) Band structure for $t=1.5,\lambda=0.1, \Delta=2, \Delta_2=-2.5$. The mirror symmetries are violated but the $\mathcal{I'T}$ symmetry is preserved; (d) the red dots illustrate the positions of the zero modes in (c).}\label{Gaphop:p3m1Z2TBand}
\end{figure}

Now we break the $\C$ symmetry by adding the terms 
\Beq
H_1=\lambda\sum_i c_i^\dag c_i + \big(t\sum_{\langle i,j \rangle}c_{i}^\dag c_{j} +{\rm h.c.}\big),
\Eeq 
then the zero modes on the nodal line are lifted except for six nodal points, as shown in Figs.\ref{Gaphop:p3m1Z2TBand_a} \!\!\& \!\!\ref{Gaphop:p3m1Z2TBand_ap}. These nodal points are resulting from level crossing protected by the quantum numbers of the mirror symmetries $\mathcal{M}'_{1,2,3}=\mathcal{I}'\mathcal{M}_{4,5,6}$, and the six nodal points are related with each other by $C_3$ and $\mathcal I'$ symmetries. The $k\cdot p$ effective Hamiltonian for the three nodal points are 
\Beq
&H'_{\rm eff 3} = a\delta k_v{\tau}_y+b \delta k_{\tilde v} {\tau}_z, \ \ 
H'_{\rm eff 2} = -a \delta k_x{\tau}_y-b\delta k_y{\tau}_z, \notag\\
&H'_{\rm eff 1} = a\delta k_w {\tau}_y+b \delta k_{\tilde w}{\tau}_z, \notag
\Eeq
with $\delta k_v={1\over2}\delta k_x+{\sqrt3\over2}\delta k_y,   \delta k_{\tilde v} = -{\sqrt{3}\over 2}\delta k_x+{1\over2}\delta k_y, 
\delta k_w={1\over2}\delta k_x-{\sqrt3\over2}\delta k_y, \delta k_{\tilde w}={\sqrt{3}\over2}\delta k_x+{1\over2}\delta k_y$.

If we further break the mirror symmetry $\mathcal{M}'_{1,2,3}$ by adding the nearest pairing term 
\beq
H_2 =\Delta_2\sum_{\langle i,j\rangle}c_{i}^\dag c_{j}^\dag   + {\rm h.c.},
\eeq 
the nodal points still exist but their positions drift away the high symmetry lines, as shown in Figs.\ref{Gaphop:p3m1Z2TBand_b}\!\! \&\!\! \ref{Gaphop:p3m1Z2TBandBM}. 
These nodal points are protected by the $\mathcal I'\T$ symmetry with $(\widehat{\mathcal I'\T})^2=1$. If we break the time-reversal symmetry by setting $\Delta_2$ as a complex number, then all the zero modes are removed and a topological SC/SF is obtained with Chern number $\pm3$.

\subsection{Multi-node $Z_2$ QSL on Honeycomb lattice: $G_b=p\bar 31m\times Z_2^{\mathcal{T}}$}\label{sec:Non-trivialPSG}

The fourth example is a $Z_2$ spin liquid model whose symmetry group is a PSG containing `off-diagonal' elements in the particle-hole sector. We adopt the PSG of the exactly solvable Kitaev spin liquid on the honeycomb lattice \cite{Kitaev06, You13} whose magnetic layer group is $G_b=p\bar 31m\times Z_2^{\mathcal{T}}$ with point group $D_{3d}\times Z_2^{\mathcal{T}}$, see Fig.\ref{p31mlattice}.  Placing spin-${1\over 2}$ Nambu bases 
$
\Psi^\dag_{i,\zeta}\equiv [c^\dag_{\uparrow}, c^\dag_{\downarrow}, c_{\uparrow}, c_{\downarrow}]_{i,\zeta},\ \zeta=A,B
$ 
respectively at the $2c$ Wyckoff positions $A=({1\over 3}, {2\over 3})$ and $B=({2\over 3}, {1\over 3})$ of the $i$th unit cell, then the PSG is obtained by replacing the generators $C_{3c}, C_{2{\alpha}}, \T, \mathcal I$ of $D_{3d}\times Z_2^T$ by the spin and charge dressed operations (see (\ref{SU2s0}) and (\ref{SU2c0})) 
\begin{align}
& C_{3c}' \!=\! (e^{-\mi{ \tilde\tau_c\over2}{2\pi\over3}}e^{-\mi{\tilde\sigma_c\over2}{2\pi\over3}}||C_{3c}),\ C_{2{\alpha}}'\!=\!(e^{-\mi{ \tilde\tau_{{\alpha}}\over2}{\pi}} e^{-\mi{ \tilde\sigma_{{\alpha}}\over2}{\pi}}||C_{2{\alpha}}),\notag\\
& \T'=( P_f^B e^{-\mi{ \tilde\tau_y\over2}{\pi}}e^{-\mi{ \tilde\sigma_y\over2}{\pi}}K||\T),\ \mathcal I'=( P_f^B ||\mathcal I),\notag
\end{align}
where $\tilde\tau_c={1\over\sqrt3}( \tilde\tau_x+ \tilde\tau_y+ \tilde\tau_z),  \tilde\tau_{{\alpha}}={1\over\sqrt2}(\tilde\tau_y- \tilde\tau_z)$ and $\tilde\sigma_c, \tilde\sigma_{{\alpha}}$ are similarly defined. The Nambu bases $[\Psi^\dag_{i,A}, \Psi^\dag_{i,B}]$ carry a 8-D projective rep of the site group $D_3\times Z_2^{\mathcal{T}}$ with  
\Beq
&&\hat{\mathcal{T}}'[\Psi^\dag_{i,A},\Psi^\dag_{i,B}] \hat{\mathcal{T}'}^{-1} = [\Psi^\dag_{i,A},\Psi^\dag_{i,B}] \nu_z\otimes {\tau}_x\otimes I_2 K,\\
&&\hat{\mathcal{I}}' [\Psi^\dag_{i,A},\Psi^\dag_{i,B}] \hat{\mathcal{I}}^{'-1}=[\Psi^\dag_{\mathcal I(i),B}, -\Psi^\dag_{\mathcal I(i),A}],
\Eeq
and
\Beq
&&\hat C_{3c}'\Psi^\dag_{i,\zeta} \hat{C_{3c}'}^{-1} = {1\over2}\Psi^\dag_{C_{3c}(\mi),\zeta}
\Bmat 1 & {-\mi} & 1 & {-\mi}\\ 1 & {\mi} & {-1} & {-\mi}\\  1 & {\mi} & 1 & {\mi}\\ {-1} & {\mi} & 1 & {-\mi}\Emat, \\
&&\hat C_{2{\alpha}}'\Psi^\dag_{i,\zeta} \hat{C_{2{\alpha}}'}^{\!\!-1} = {1\over2}\Psi^\dag_{C_{2{\alpha}}(i),\zeta} \Bmat{-1} & {-\mi} & {-1} & {-\mi}\\ {\mi} & {1} & {-\mi} & {-1}\\  {-1} & {\mi} & {-1} & {\mi}\\ {\mi} & {-1} & {-\mi} & {1}\Emat.
\Eeq
The particle-hole operation is simply represented as $\Gamma_xK$ as shown in (\ref{PHS}).

\begin{figure}[t]
  \centering
\includegraphics[scale=0.25]{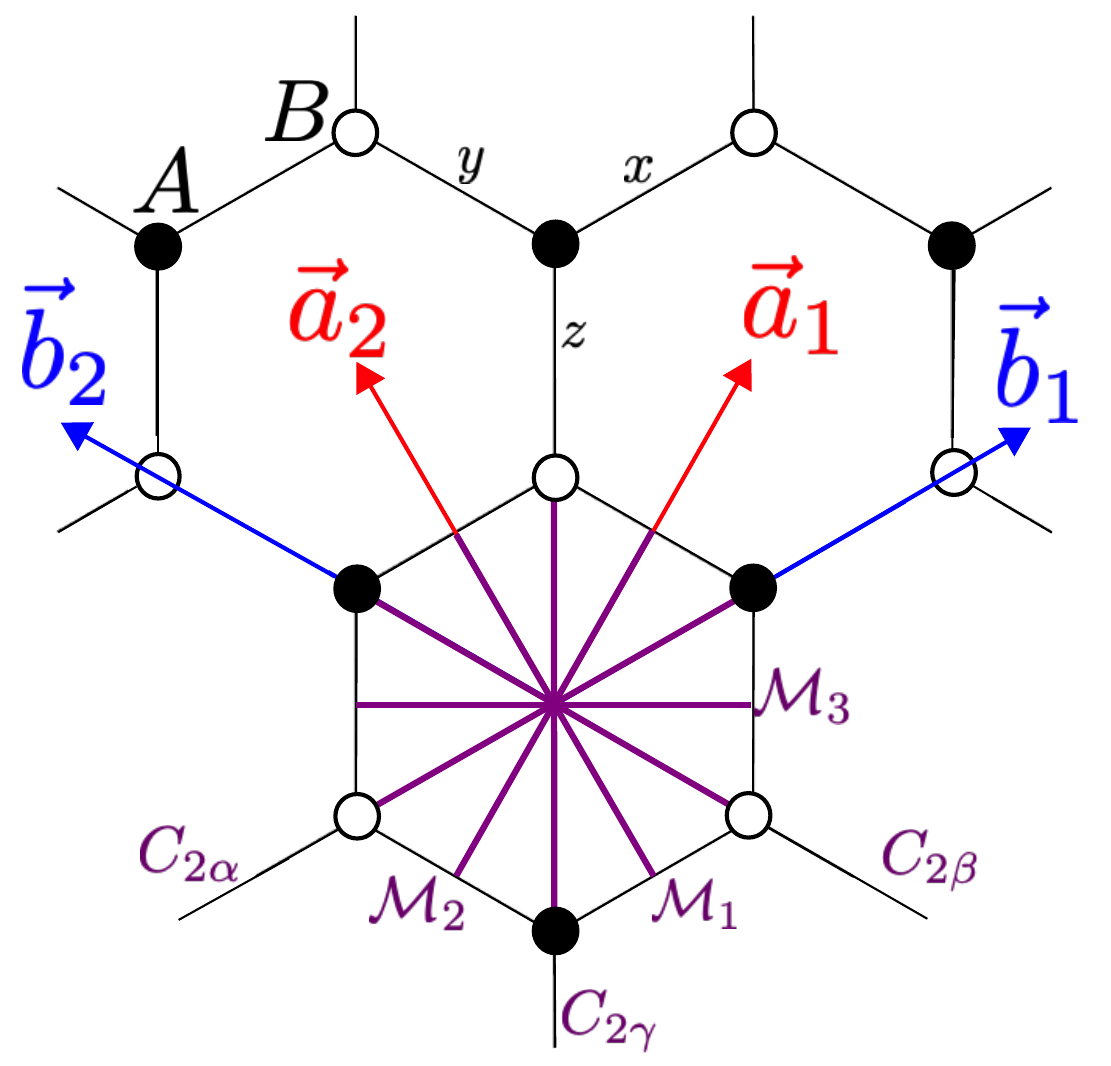}
\caption{Lattice basis vectors (red) and their dual vectors (blue) for the magnetic layer group $G_b=p\bar 31m\times Z_2^{\mathcal{T}}$. The (maximal) $2c$ Wyckoff positions $A,B$ are indicated by black dots and hollow circles respectively, and all mirror plans and 2-fold axes  are indicated by purple lines.}\label{p31mlattice}
\end{figure}

Considering nearest and next next nearest neighbor terms between different sub-lattices, the most general BdG Hamiltonian that preserves the above PSG is given by (the intra-sub-lattice terms are not allowed by the PSG)
\beq\label{Ham:GKSL}
\!\!\! H_0 &\!\!\!=\!\!\!\!& \sum_{\langle i,j\rangle\in z} \!\! \big( \Psi_{i,A}^\dag U \mathcal H_1U^\dag \Psi_{j,B} \! + {\rm h.c.} \big) + (\mbox{x-,y-bonds})  +  \notag\\
&&\!\!\!\!\! \sum_{\langle\langle\langle i,j'\rangle\rangle\rangle \in Z}\!\!\!\!\!\! \big( \Psi_{i,A}^\dag U \mathcal H_2 U^\dag \Psi_{j',B} \!+\! {\rm h.c.} \big) \!\!+\! (\mbox{X-,Y-bonds}) , 
\eeq
where $\mathcal H_1, \mathcal H_2$ are the Hamiltonian matrices in the Majorana bases
\Beq
\mathcal H_1=
\mi  \Bmat 
\!\!\eta_0&\eta_1&-\eta_1&0\\ 
\!\!\eta_1&\eta_2&\eta_4&\eta_3\\ 
\!\!-\eta_1 & \eta_4& \eta_2& \eta_3\\
\!\!0 & \eta_3 &\eta_3 & \eta_5 
\Emat\!,  
\mathcal H_2=\mi  \Bmat
\!\!\eta'_0&\eta'_1&-\eta'_1&0\\ 
\!\!\eta'_1&\eta'_2&\eta'_4&\eta'_3\\ 
\!\!-\eta'_1 & \eta'_4& \eta'_2& \eta'_3 \\
\!\!0 & \eta'_3 & \eta'_3 &  \eta'_5
\Emat\!,
\Eeq
and $U\!\!=\!\!{1\over\sqrt{2}}\Bmat1 & 0 & 0 & {\mi}\\ 0 &{\mi} &{-1} & 0\\  {1} & 0 & 0 & {-\mi}\\ 0 &{-\mi} &{-1} & 0\Emat$ is the transformation matrix mapping the complex fermion bases $\Psi^\dag_{i,\zeta}$ into the Majorana bases\cite{Kitaev06}, namely $[c, b_x, b_y, b_z]_{i,\zeta}=\Psi^\dag_{i,\zeta} U$. The parameters $\eta_{0,1,2,3,4,5}$, $\eta'_{0,1,2,3,4,5}$ are real numbers among which $\eta_0,\eta_5$ come from the solution of the pure Kitaev model. The terms on the x-,y-bonds (X-,Y-bonds) can be obtained from ones on the z-bonds (Z-bonds) by successive $C'_{3c}$ rotations.

The spectrum of the Hamiltonian (\ref{Ham:GKSL}) can support a series of zero modes in the BZ and the corresponding state is a multi-node QSL. For instance, Fig.\ref{GKSL} shows the band structure with parameters $\eta_0=2, \eta_1=1.2, \eta_2=0.8, \eta_3=0.5, \eta_4=0.2, \eta_5=1.2$ and $\eta'_0=1.6, \eta'_1=0.2, \eta'_2=0.8, \eta'_3=2.5, \eta'_4=0, \eta'_5=0$.  There are totally $20$ cones,  where each of the $2$ cones in green color contains irreducible zero modes protected by the little co-group $F_f(K)=\mathscr C'_{3v}\times Z_2^{\mathcal I'\T'}$; the $6$ cones in blue color contain zero modes protected by crossing of energy bands carrying different quantum numbers of the mirror symmetries $\mathcal M'_{1,2,3}$; the $12$ cones in red color contain zero modes protected by $\pi$-quantized Berry phase due to the $\mathcal{I'T'}$ symmetry.

\begin{figure}[t]
  \centering
\subfigure[Kitaev QSL: band structure]{\includegraphics[scale=0.15]{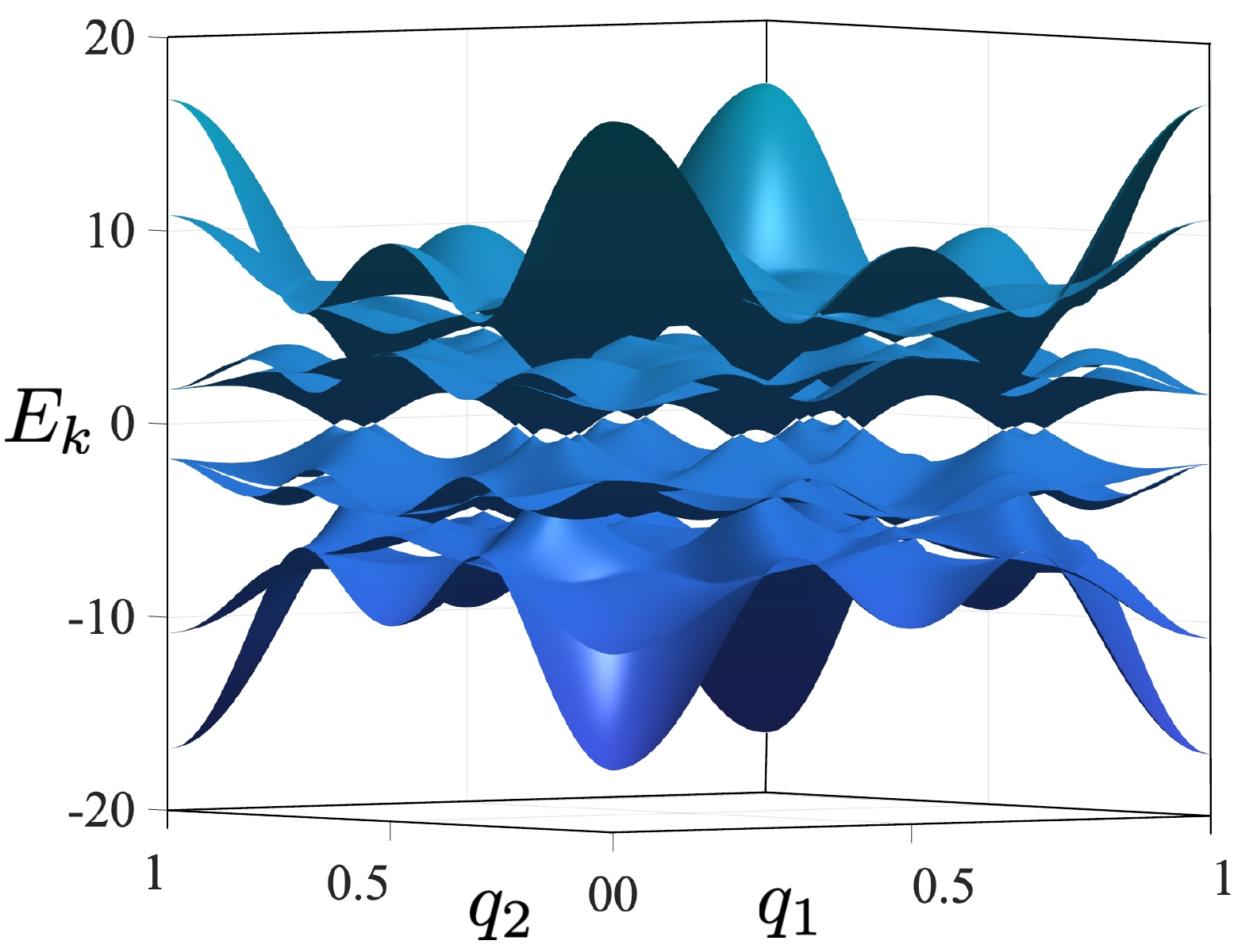}}
\subfigure[Kitaev QSL: positions of the zero modes]{\includegraphics[scale=0.16]{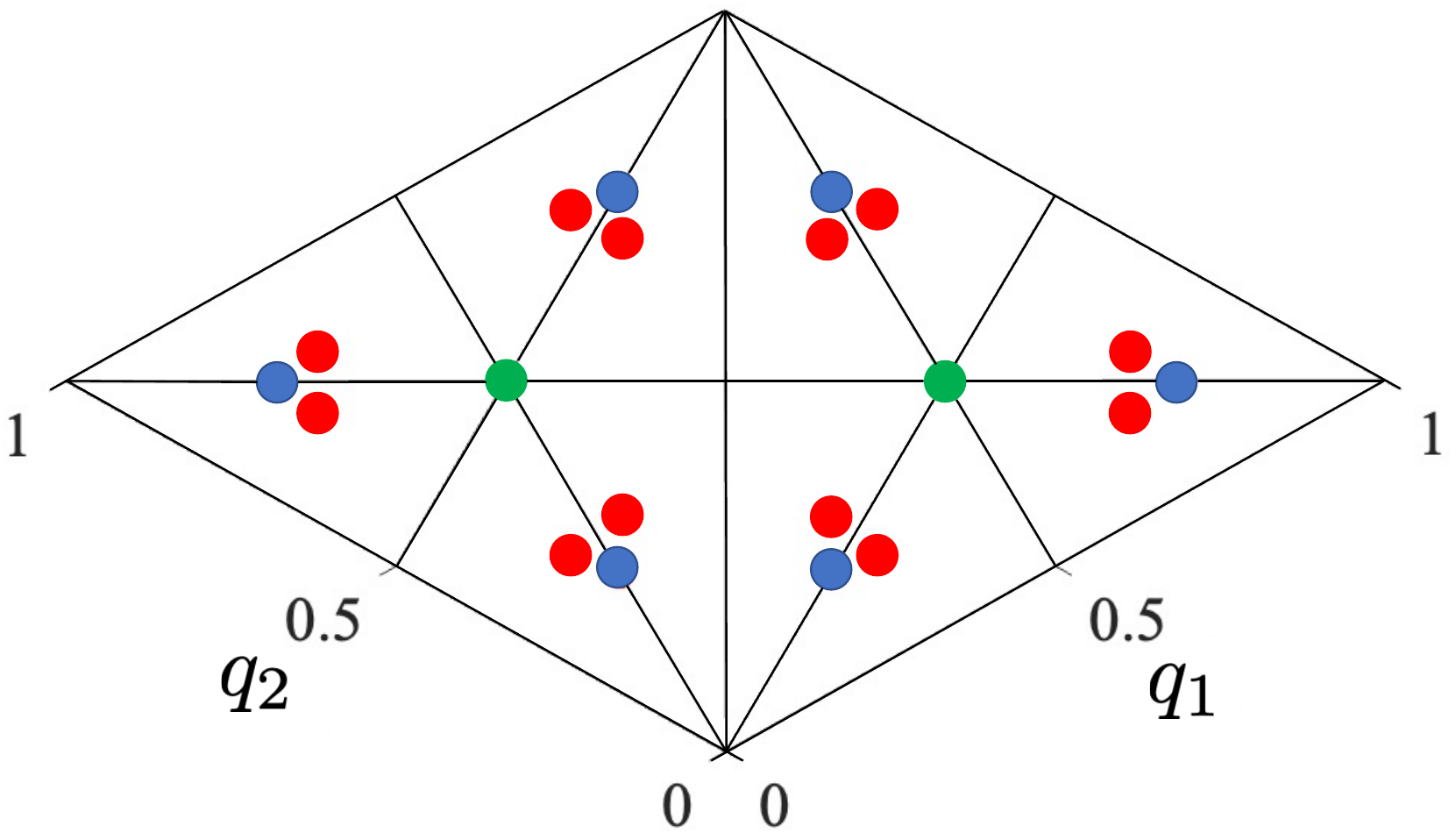}}
  \caption{(a)The Band structure of the Hamiltonian (\ref{Ham:GKSL}) with parameters given in the maintext. The number of the cones is `$12+6+2$', where the `$12$' cones in red are due to zero modes protected by the $\mathcal{I'T'}$ symmetry, the `$6$' cones in blue are protected by the mirror symmetries $\mathcal{M}'_{1,2,3}$, and each of the `$2$' green cones results from irreducible zero modes; (b) the dots illustrate the positions of the zero modes in (a). }\label{GKSL}
\end{figure}

The two high symmetry points $K=({1\over3},{2\over3})$ and $K'=({2\over3},{1\over3})$ have the little co-group $F_f(k)/Z_2^f =\mathscr C'_{3v}\times Z_{2}^{\mathcal{I'T'}}\times Z_2^{\mathcal{PT'}}$. All of the $C'_{3c}$, $\mathcal {M'}_{1,2,3}$ symmetries can be treated as effective charge conjugation, and the 2-fold irreducible zero modes are represented as 
\Beq
&&D(C'_{3c})={\sqrt{3}\over2}I + {\mi\over2}\mu_z, \ \ 
D(\mathcal{M}'_1) = {\sqrt{3}\mi\over2}\mu_y-{\mi\over2}\mu_x, \\
&&D(\mathcal{I'T'})=\mu_x K,\ \  D(\mathcal{T'P})=-\mu_z. 
\Eeq
According to Eq.(\ref{DispersionCriterion}) one has $a_{({\rm v\times ph})}=1$. So the dispersion around the $K$ and $K'$ points is linear and the $k\cdot p$ Hamiltonian has one free parameter $a$, yielding
\Beq	
H_{\rm eff1}=a( \delta k_x\mu_x +\delta k_y \mu_y).
\Eeq
Similarly, the $k\cdot p$ Hamiltonian for the cones at the high symmetry lines with mirror symmetries $\mathcal{M}'_{1,2,3}$ have two free parameters $b,c$. For example, around the left blue cone, the effective Hamiltonian reads
\Beq
H_{\rm eff2}=&b\delta k_x {\sqrt2\over2}(\cos\theta_2\mu_x-\sin\theta_2\mu_y) \\
&+c\delta k_y{\sqrt2\over2}(\sin\theta_2\mu_x + \cos\theta_2\mu_y)
\Eeq
with $\theta_2\approx 0.47\pi$.  
Furthermore, the $k\cdot p$ Hamiltonian for the cones at a generic point is given as
\begin{align}
H_{\rm eff3}=(d\mu_x + e\mu_y)\delta k_x+(f\mu_x + g\mu_y)\delta k_y,
\end{align}
which has four parameters $d,e,f,g$.

Notice that the point group symmetry elements are associated with charge and spin operations in the PSG. This makes the symmetries easy to break. For instance, the Zeeman coupling term $H_1=B\sum_{i,\zeta} \Psi_{i,\zeta}^\dag {\tilde\sigma_c\over2} \Psi_{i,\zeta}$
from a weak magnetic field breaks the symmetries $e^{-\mi{\tilde\sigma_{\alpha}\over 2}\pi}, e^{-\mi{\tilde\sigma_y\over2}\pi}K$ in the spin sector. Consequently the mirror symmetries $\mathcal M'_{1,2,3}$ and the time reversal symmetry $\T'$ in the PSG are broken, hence all of the cones are gapped out. The resulting state has nonzero Chern number $\pm 8$. Furthermore, a chemical potential term $H_2
= {\lambda}\sum_{i,\zeta} (\Psi^\dag_{i,\zeta} {\tilde\tau_z\over2} \Psi_{i,\zeta}) +{\rm const}$ can be introduced by doping the system with charge carriers. This term breaks the $e^{-\mi{\tilde\tau_{\alpha}\over2}\pi}, e^{-\mi{\tilde\tau_y\over2}\pi}K$ symmetries in the charge sector, consequently the mirror symmetries $\mathcal M'_{1,2,3}$ and the time reversal symmetry $\T'$ in the PSG are also broken. All of the cones are then gapped out, resulting in a topological SC state with nonzero Chern number $\pm2$\cite{You13, ZhangLiu2022} (SC with hole doping has opposite Chern number with the one with electron doping).

In Fig.\ref{GKSL}, the cones with the same color are symmetry related and form a $\{k^*\}$.  With the varying of the parameters $\eta_0\sim\eta_5$ and $\eta'_0\sim \eta'_5$, the positions of the cones, the number of the cones and even the pattern of the cones (the collection of several $\{k^*\}$) may change accordingly. The pattern of the cones determines the physical response of probe fields. Hence spin liquid states with the same pattern of cones can be considered as belonging to the same phase. Therefore, the pattern of cones is an important information of a gapless QSL in addition to the PSG.

\section{Discussion and Conclusion}\label{sec:conc}

Before concluding this work, we comment on some interesting issues.

\textit{To be Majorana or not.} We have shown that with the protection of effective charge conjugation symmetry $\C$, zero modes can appear at high-symmetry points. Since charge conjugation transforms a particle into its antiparticle, the quasiparticle corresponding to the $\C$-zero modes is its own antiparticle and can be identified as a Majorana quasiparticle.

We have also shown that when effective charge conjugation symmetry $\C$ is violated, bulk zero modes can still appear due to quantum number-protected level crossings or $\mathcal{I}'\T'$-protected quantized Berry phases. Since the Hamiltonian does not commute with $\C$, the $\C$-transformed quasiparticle is not an eigenstate of the Hamiltonian. Therefore, strictly speaking, quasiparticles associated with zero modes without $\C$ symmetry are not of Majorana type.

Although $\mathcal{P}$ always anti-commutes with the Hamiltonian, $\mathcal{P}$-transformed quasiparticles are not independent of the original quasiparticles since $\mathcal{P}$ is a redundancy of the SCs/SFs/QSLs.

\textit{Classification of nodal SCs/SFs/QSLs.} If two BdG systems have the same symmetry and their low-energy quasiparticles correspond one-to-one, we can identify the two systems as the same phase. In this sense, we propose a method to classify gapless nodal SCs/SFs/QSLs:

(i) The symmetry group, which extends the space group of the lattice by internal symmetries such as spin rotation and charge operations.

(ii) The pattern of the bulk zero modes, whose positions form several sets of $k$ stars, denoted $\{k^\star\}$, where each $k^\star$ is a set of symmetry-related equivalent momentum points. For each set $\{k^\star\}$, we determine the total number of zero-energy nodes in the system.

(iii) The degeneracy of the zero modes at each point in $\{k^\star\}$. The degeneracy of zero modes at equivalent $k$ points is the same.

(iv) The dispersion and physical properties of the quasiparticles, described by the $k \cdot p$ effective theory.

\textit{Effect of interactions.} It is expected that under weak interactions, the quasiparticles corresponding to the bulk zero modes remain gapless. However, with increasing interaction strength, the fate of the quasiparticles depends on the degeneracy and the form of dispersion. Quantitative results should be obtained using renormalization group calculations, which are beyond the scope of this work.

In summary, using projective representation theory, we systematically studied the symmetry conditions for the existence of bulk zero modes in the mean-field theory of superconductors, superfluids, and $Z_2$ QSLs, where fermions pair to form Cooper pairs. We then provided an efficient method to obtain the $k \cdot p$ effective theory of these gapless quasiparticles, from which one can determine the dispersions and their physical responses to external fields. The positions of the zero modes can be controlled by adjusting the symmetry of the Hamiltonian. Furthermore, the dispersion of the gapless quasiparticles can be linear or of higher order. These results are illustrated with concrete lattice models. Our symmetry representation-based theory complements the topological origin (theory of symmetry indicators) of gapless quasiparticles in BdG systems and aids in the experimental realization of Majorana-like gapless quasiparticles.
\\

We thank Yu-Xin Zhao, Zhi-Da Song, Chen Fang, Yuan-Ming Lu, Jiucai Wang, Houmin Du, Zhen Feng and Bruce Normand for helpful discussions. A.Z.Y thank Xiao-Yue Wang for her helpful comments. This work is supported by NSFC (Grants No.12374166, 12134020) and National Key Research and Development Program of China (Grants No.2023YFA1406500, 2022YFA1405300).

\appendix

\section{Meaning of Pauli matrices $\pmb \tau, \pmb \sigma, \pmb \mu, \pmb \nu, \pmb \omega$}
In the main text, we have used different notations of Pauli matrices. In the following we list their physical meaning.

$\tau_{x,y,z}$ generate the charge $SU(2)_c$ group and act on the particle-hole degrees of freedom. For instance the charge conjugation can be chosen as $\C=e^{-i{\tau_x\over2}\pi}$. The product of $\tau_{x,y,z}$ with identity matrix are noted as $\Gamma_{x,y,z}=\tau_{x,y,z}\otimes I$.

$\sigma_{x,y,z}$ generate the spin $SU(2)_s$ group and act on the spin degrees of freedom.

$\mu_{x,y,z}$ act on the positive/negative eigen spaces of the Hamiltonian. The eigenvalue $1$ (or $-1$) of $\mu_z$ labels the positive (or negative) eigen sapce of the total Hamiltonian, and $\mu_{x,y}$ exchange the two subspaces.

$\nu_{x,y,z}$ act on the A,B sublattices. The eigenvalue $1$ (or $-1$) of $\nu_z$ labels the positive (or negative) labels the A (or B) sublattice, and $\nu_{x,y}$ exchange the two sublattices.

For an irrep $D(G)$ of anti-unitary group $G$ of the $\mathbb C$ or $\mathbb H$ type, one can reduce the restricted rep $D(H)$ into two irreps with $H$ the maximal unitary subgroup of $G$. $\omega_{x,y,z}$ acts on the irrep spaces of $D(H)$. The eigenvalue $\pm1$ of $\omega_z$ respectively labels the two irrep spaces, and $\omega_{x,y}$ exchange the two subspaces. Furthermore, the notation $\Omega_{x,y,z}=\omega_{x,y,z}\otimes I$ has been used in the main text.

\section{ Proof of $[D(\mathcal{P})K,D(g)K_{s(g)}]=0$}\label{App.:PcommutingwithGf}

In the Nambu bases $\Psi^\dag=[c_1^\dag, ..., c_N^\dag, c_1, ..., c_N]$, suppose a general symmetry operation $g$ (mixing or non-mixing) is represented as
\Beq
\hat g\Psi^\dag \hat g^{-1}=\Psi^\dag \begin{pmatrix} A(g)&B^*(g)\\B(g) &A^*(g) \end{pmatrix}K_{s(g)} =\Psi^\dag D(g) K_{s(g)}.
\Eeq
Specially, the particle-hole symmetry is represented as $D(\mathcal{P})K=[\tau_x\otimes I_N] K=\Gamma_x K$. Due to the fact that $D(g)=[I_2\otimes \rm{Re} (A) +\mi\tau_z\otimes \rm{Im}(A)+\tau_x\otimes \rm{Re}(B)+\tau_y\otimes \rm{Im}(B)]$ and that $\tau_x K$ commutes with the four matrices $\{I_2,\mi\tau_z,\tau_x,\tau_y\}$, it is easily verified that
$
D(g)\Gamma_x = \Gamma_x [D(g)]^*,
$
Namely,  
\Beq
D(g)K_{s(g)}D(\mathcal{P}) K = D(\mathcal{P})KD(g)K_{s(g)}.
\Eeq

This conclusion can be seen more straightforwardly in the majorana representation with 
\Beq
&&\gamma_i^{l} = c_i+c_i^\dag,\ \ \gamma_i^{r} = -i(c_i-c_i^\dag),\\
&&\{\gamma_i^\alpha,\gamma_j^\beta\}=2\delta_{\alpha\beta}\delta_{ij},\ \ \alpha,\beta=l,r.
\Eeq 
With these bases, any symmetry operation $g$ is represented as a REAL matrix $\mathscr D(g)$, especially $\mathcal P$ is represented as $\mathscr D(\mathcal P)=I_{2N}K$. Hence the commutation of $\mathscr D(\mathcal P)K$ and the real rep $\mathscr D(g)K_{s(g)}$ manifests itself.

In the reduced Nambu subspaces $\psi^\dag =[C_\up^\dag, C_\dn^T]$  with $C_\up^\dag=[c_{1\up}^\dag, c_{2\up}^\dag, ... , c_{L\up}^\dag]$, $C_{\dn}^T=[c_{1\dn}, c_{2\dn}, ..., c_{L\dn}]$ for spin-${1\over 2}$ fermions having spin rotation symmetry, one has $\mathcal{P}^2=-1$ with $D(\cP)K = \mi\tau_y\otimes I_L K$. Generally, a symmetry operation $g$ acts on the lattice and charge (but no spin) degrees of freedom as 
\beq\label{sm:S1/2}
\hat g\psi^\dag \hat g^{-1}\!=\!\psi^\dag \!\begin{pmatrix} A(g)&-B^*(g)\\B(g) &A^*(g) \end{pmatrix}\!\! K_{s(g)} \!=\! \psi^\dag D(g)K_{s(g)},
\eeq
which can be decomposed as $D(g)=[I_2\otimes \rm{Re}(A) +\mi\tau_z\otimes \rm{Im}(A)+\mi\tau_x\otimes \rm{Im}(B)-\mi\tau_y\otimes \rm{Re}(B)].$ Since $\mi\tau_y K$ commutes with  the four matrices $\{I_2,\mi\tau_z,\mi\tau_x,\mi\tau_y\}$. So, $[D(\cP)K, D(g)K_{s(g)}]=0$.

Now we verify equation (\ref{sm:S1/2}). Suppose $g$ acts on $C_\up^\dag$ as
\[
\hat gC_\up^\dag \hat g^{-1} = \left(C_\up^\dag A(g) +  C_{\dn}^T B(g)\right)K_{s(g)}.
\]
Noticing that $C_{\up}^\dag$ and $ C_{\dn}^T$ are related by complex conjugation,
we have $\hat gC_\up^T \hat g^{-1} = \left(C_\up^T A^*(g) +  C_{\dn}^\dag B^*(g)\right)K_{s(g)}.$
On the other hand, noticing that $\hat C_{2y} C_{\up}^T \hat C_{2y}^{-1} = C_{\dn}^T, \hat C_{2y} C_{\dn}^\dag \hat C_{2y}^{-1} = -C_{\up}^\dag$ and that $\hat C_{2y}$ commutes with $\hat g$ (because $\hat C_{2y}$ only acts on spin while $\hat g$ does not act on spin), therefore
\Beq
\hat gC_\dn^T \hat g^{-1} &=&\hat g \hat C_{2y} C_\up^T \hat C_{2y}^{-1}\hat g^{-1} =  \hat C_{2y} \hat g C_\up^T \hat g^{-1}\hat C_{2y}^{-1}\\
 &=& \left(\hat C_{2y} C_\up^T\hat C_{2y}^{-1} A^*(g) +  \hat C_{2y} C_{\dn}^\dag \hat C_{2y}^{-1} B^*(g) \right)K_{s(g)}\\
  &=&\left(- C_{\up}^\dag B^*(g) + C_\dn^T A^*(g)K_{s(g)} \right)K_{s(g)},
\Eeq
which yields eq.(\ref{sm:S1/2}).

\section{Generalized band representation for SCs/SFs/QSLs}\label{sec:BandReps} 

When the pairing terms are switched off, the SCs/SFs become metals or insulators where a number of localized orbitals in the unit cell in real space determine the energy band structure in the Brillouin zone. The energy bands as entire entities form a special representation of the system's symmetry group $G$ called the band representation \cite{zak1982band,Bradlyn2017TopoQuanChemistry,cano2021band}. Such a band rep is called {\it elementary} if it cannot be decomposed as a direct sum of smaller band representations.  Once the elementary band reps are known, all the band reps can be obtained by the direct sum of elementary ones.

The band representations can be generalized to BdG systems by introducing `particle-hole' degrees of freedom. We start from the full `symmetry' group $F_f=G_f+\mathcal{P}G_f$. For simplicity, here we focus on the case where the translation symmetries of $ G_f $ are not fractionalized ({\it i.e.} the unit cell is not enlarged). The construction of an elementary band representation needs two ingredients. The first is a `maximal' Wyckoff position $q_1$ with a site-symmetry group $G_{q_1}$ (composed of operations in $G_f$ which stabilize the site $q_1$). The second is a set of creation and annihilation operators of localized orbitals $\phi_1(q-q_1-R_m),\phi_2(q-q_1-R_m), ... ,\phi_r(q-q_1-R_m)$ centered at $q_1$ in the $m$-th unit cell, namely $[c_{\sigma, q_1+R_m}^\dag, c_{\sigma, q_1+R_m}]$ with $\sigma =1,2,...,r$ and 
$$
c_{\sigma, q_1+R_m}^\dag|{\rm vac}\rangle \equiv \phi_\sigma(q-q_1-R_m),
$$ 
where $|{\rm vac}\rangle$ is the vacuum state. 

The operators $[c_{\sigma, q_1}^\dag, c_{\sigma, q_1}]$ carry a $2r$-dimensional rep $d^{(l)}(G_{q_1})$ of $G_{q_1}$, that is, for any $\gamma=(s_\gamma,c_\gamma|| p_\gamma|t_{\gamma}) \in G_{q_1}$, 
\Beq
\widehat{\gamma}[c_{q_1}^\dag, c_{q_1}] \widehat{\gamma}^\dag = [c_{q_1}^\dag, c_{q_1}] d^{(l)}(\gamma),
\Eeq
where the index $\sigma$ has been hidden, namely $[c_{q_1}^\dag, c_{q_1}]\equiv[c_{1,q_1}^\dag, ... ,c_{r,q_1}^\dag, c_{1,q_1}, ... ,c_{r,q_1}]$. More generally, one has
\Beq
\widehat{\gamma}[c_{q_1+R_m}^\dag, c_{q_1+R_m}] \widehat{\gamma}^\dag \!\!=\! [c_{q_1+p_\gamma R_m}^\dag, c_{q_1+p_\gamma R_m}] d^{(l)}(\gamma).
\Eeq
Performing a Fourier transformation, one obtains the bases carrying a rep of both $G_{q_1}$ and the translation group,
\Beq
c^\dag_{k,\sigma,q_1} &=& \sum_{R_m} e^{\mi k\cdot R_m} c^\dag_{\sigma ,q_1+R_m},\\
c_{-k,\sigma,q_1} &= &\sum_{R_m} e^{\mi k\cdot R_m} c_{\sigma ,q_1+R_m}.
\Eeq
Then for a group element $\gamma=(s_\gamma,c_\gamma|| p_\gamma|t_{\gamma}) \in G_{q_1}$, we have
\begin{align} \label{gammac}
 \widehat{\gamma} [c^\dag_{k, q_1}, c_{-k, q_1}] \widehat{\gamma}^{-1} =[c^\dag_{p_\gamma k, q_1} , c_{-p_\gamma k, q_1}] d^{(l)}(\gamma).
\end{align}

On the other hand, Wyckoff positions equivalent to $q_1$ form a star $\{q_1^*\}=\{q_1, q_2, ... , q_n\}$ with $q_i\equiv \{p_{\alpha_i} | t_{\alpha_i}\}q_1$ and $\{p_{\alpha_i} | t_{\alpha_i}\}$ the space operations of the representative $\alpha_i=(s_{\alpha_i},c_{\alpha_i}||p_{\alpha_i} | t_{\alpha_i})$ in one of the $n$ cosets in the left coset decomposition of $G_f$ with respect to its subgroup $G_{q_1}$. Notice that the choice of coset representative is not unique, but in later discussion, we fix the representative $\alpha_i$ in the corresponding coset and use it to label the coset. Especially,  one has
\Beq
&&\!\!\!\!\!\!\widehat{\alpha_i} \phi_{\sigma}(q-q_1-R_m)\\ 
&&\ \ \ \ =\widehat{(s_{\alpha_i},c_{\alpha_i})} \phi_{\sigma}(p_{\alpha_i}^{-1}(q-t_{\alpha_i})- q_1- R_m)\\
&&\ \ \ \ =\widehat{(s_{\alpha_i},c_{\alpha_i})} \phi_{\sigma}(p_{\alpha_i}^{-1}(q-q_i -p_{\alpha_i}R_m)).
\Eeq
Supposing $\alpha_i R_m=R_n$  and denoting 
$$
c_{\sigma,q_i+R_n}^\dag |{\rm vac}\rangle \equiv \widehat{(s_{\alpha_i},c_{\alpha_i})}\phi_{\sigma}\big(p_{\alpha_i}^{-1}(q-q_i-R_n)\big),
$$
then we have
\beq\label{alphac}
\widehat{\alpha_i}  [c^\dag_{k,\sigma,q_1}, c_{-k,\sigma,q_1}]   \widehat{\alpha_i}^{-1}  \equiv\! \big[c^\dag_{ p_{\alpha_i} k,\sigma, q_i}, c_{-p_{\alpha_i} k,\sigma, q_i}\big].
\eeq

The set of bases $$\big\{ [c^\dag_{ k,\sigma,q_i}, c_{-k,\sigma,q_i}]; 1\leq\sigma\leq r, 1\leq i\leq n, k\in {\rm BZ} \big\}$$ span a (infinite dimensional) band rep of $G_f$. To illustrate, we investigate the rep of $\alpha=(s_{\alpha},c_{\alpha}||p_\alpha | t_{\alpha}) \in G_f$. For any coset representative $\alpha_i$, one can always find another coset representative $\alpha_j$ such that
\begin{align}\label{AppC:2}
\alpha\cdot \alpha_i=(E,E||E| R_{\alpha,\alpha_i})\cdot\alpha_j\cdot\gamma
\end{align}
with $\gamma=(s_\gamma,c_\gamma|| p_\gamma|t_{\gamma}) \in G_{q_1}$ and $R_{\alpha,\alpha_i}$ is a Bravais lattice vector. The $R_{\alpha,\alpha_i}$ can be calculated by acting the (\ref{AppC:2}) on $q_1$, and we have 
$$R_{\alpha,\alpha_i}=p_{\alpha}\cdot q_i+t_{\alpha}-q_j.$$ 
Using (\ref{gammac}) and (\ref{alphac}), the induced projective band representation for $\alpha$ is given by
\beq\label{BandReps}
&&\!\!\!\!\! \widehat{\alpha}[c^\dag_{k, q_i}, c_{-k,  q_i}]\widehat{\alpha}^{-1}\notag \\
&=&\!\!\! [c^\dag_{p_{\alpha} k, q_j}, c_{-p_{\alpha} k, q_j}] d^{(l)}(\gamma)e^{-\mi (p_{\alpha} k)\cdot (p_{\alpha}\cdot q_i+t_{\alpha}-q_j)}, 
\eeq
where the index $\sigma$ has been hidden.  Equivalently, the entries of the band rep of $\alpha$ in the equation (\ref{BandReps}) are given as:
\begin{align}\label{BandRepsEntry}
D^{(l)}_{(k,\Sigma,q_i), (p_{\alpha} k,\Sigma', q_j)}\!(\alpha) \!=\! d^{(l)}_{\Sigma,\Sigma'}(\gamma) e^{-\mi (p_{\alpha} k)\cdot (p_{\alpha}\cdot q_i+t_{\alpha}-q_j)}
\end{align}
with $\Sigma=1, ..., 2r$.

Especially, if $p_{\alpha}\cdot k = k + b_m$ with $b_m$ a reciprocal lattice vector, namely, if $\{(p_{\alpha} | t_\alpha)\}$ belongs to the little group of the momentum $k$, for the $R_{\alpha,\alpha_i}$ is a Bravais lattice vector, then (\ref{BandRepsEntry}) reduces to
\begin{align}\label{LittleBandRepsEntry}
D^{(l)k}_{(\Sigma,q_i), (\Sigma', q_j)}(\alpha)  = d^{(l)}_{\Sigma,\Sigma'}(\gamma) e^{-\mi k\cdot (p_{\alpha}\cdot q_i+t_{\alpha}-q_j)}.
\end{align}
Moreover, if one defines the representation matrix for the little co-group element $g_\alpha=(s_\alpha, c_\alpha ||p_\alpha) \in G_f(k)$ with 
\Beq
D^{(l)k}_{(\Sigma,q_i), (\Sigma', q_j)}(g_\alpha) &=& D^{(l)k}_{(\Sigma,q_i), (\Sigma', q_j)}(\alpha)e^{\mi k\cdot t_\alpha}\\
&= &d^{(l)}_{\Sigma,\Sigma'}(\gamma) e^{-\mi k\cdot (p_{\alpha}\cdot q_i -q_j)},
\Eeq
then $D^{(l)k}_{(\Sigma,q_i), (\Sigma', q_j)}(G_f(k))$ form a projective rep of the little co-group $G_f(k)$ with the factor system 
$$
\omega_2(g_\alpha, g_\beta) = e^{-i K_{\alpha} \cdot t_{\beta}} 
$$ 
where $g_\alpha=(s_\alpha, c_\alpha ||p_\alpha), g_\beta=(s_\beta, c_\beta ||p_\beta)$, $K_{\alpha} = p_\alpha^{-1}\cdot k-k$ with $\alpha,\beta\in G_f$ and $g_\alpha, g_\beta \in G_f(k)$.

Above we illustrated the band rep of unitary elements. The procedure can be generalized to anti-unitary group elements, which will not be repeated here.

Now we extend the band representation to the full symmetry group $F_f(k)$. Supposing $\tp =  g_0\cP \in F_f(k)$ with $g_0\notin F_f(k)$ and $g_0$ mapping $\big[c_{k,\{\sigma,q_i\}}^\dag,c_{-k\{\sigma,q_i\}}\big]$ to $\big[c_{-k,\{\sigma',q_j\}}^\dag, c_{k,\{\sigma',q_j\}}\big]$, then
the rep of $g_0$ can be obtained from (\ref{BandReps}) or (\ref{BandRepsEntry}) and the rep of $\mathcal{P}$ has two possibilities. For the class D or DIII, the Nambu bases read $\Psi_k^\dag=[c_{k,\{\sigma,q_i\}}^\dag, c_{-k,\{\sigma,q_i\}}]$ and the $\mathcal{P}$ is represented as $\tau_x\otimes I K$; while for the class C or CI, the reduced Nambu bases read $\psi_k^\dag= [C^\dag_{k\up}, C_{-k\dn}]=[ c_{k\up \{q_i\}}^\dag, c_{-k\dn \{q_i\}}]$ in which the rep $d^{(l)}(\gamma)$ is $r$-dimensional and the $\mathcal{P}$ is represented as $\mi\tau_y\otimes I K$. Then the rep of $\p$ can be obtain from the reps of $g_0$ and $\mathcal{P}$.

\section{The restriction from irreps of $F_f(k)$ to its normal subgroup $G_f(k)$}\label{App.:AA}

We start from a theorem:\\ 

\begin{theorem}\label{thrm:1}
If a group $G$ has normal subgroup $H$ with $G = H + r H, r^2\in H$ and $G/H\cong Z_2$, supposing a Hilbert space $\mathcal L$ carries a projective irrep $D(H)$ of $H$, then the restricted rep $D(H)$ contains at most two irreps of $H$. There are two possibilities:
\bit
\item[(I)] $D(H)$ is an irrep of $H$; 
\item[(II)] $D(H)$ can be transformed into a direct sum of two irreps, namely $\mathcal{L}=\mathcal{L}_{1}\oplus \mathcal{L}_{2}$ where $\mathcal{L}_{1}, \mathcal{L}_{2}$ each carries an irrep of $H$, and $r$ permutes $\mathcal{L}_{1},\mathcal{L}_{2}$.  
\eit
\end{theorem}
\begin{proof}
\rm 
Since $\mathcal{L}$  carries a representation of $H$, it must contain an invariant subspace $\mathcal{L}_{1}\subseteq \mathcal L$ of $H$. Namely, for any $u\in H$, one has
\[
u \mathcal{L}_{1} = \mathcal{L}_{1}.
\]
Meanwhile, $\mathcal{L}$ carries an irrep of $G$, so $r$ transforms $\mathcal{L}_{1}$ into another linear subspace $\mathcal{L}_{2}\equiv r(\mathcal{L}_{1})\subseteq \mathcal L$,
then it is easily seen that $\mathcal{L}_{2}$ is also invariant under the action of $u\in H$:
\Beq
u\mathcal{L}_{2} = u r \mathcal{L}_{1} &=&r \left(r^{-1} u r\right) \mathcal{L}_{1} \\&=& 
r u' \mathcal{L}_{1} = r \mathcal{L}_{1} = \mathcal{L}_{2},
\Eeq
where we have used the fact that $u'= (r^{-1} u r) \in H$. Therefore, $\mathcal{L}_{2}$ also carries an irrep of $H$. 

Since both $\mathcal{L}_{1}$ and $\mathcal{L}_{2}$ are irreducible representation spaces of  $H$, there are two possibilities: \\
(I) $\mathcal{L}_{1} = \mathcal{L}_{2}=\mathcal{L}$. In this case, the restricted rep $D(H)$ from $D(G)$ is irreducible, we call such $D(G)$ a {\it simple irrep};\\
(II) $\mathcal{L}_{1}\cap\mathcal{L}_{2}=\emptyset$. In this case, $\mathcal L' = \mathcal{L}_{1} \cup \mathcal{L}_{2}$ is an irrep space of $G$. Noticing that $\mathcal L$ is an irrep space of $G$, and that irreducible space of $G$ containing $\mathcal L_1$ as a subspace is unique, we have $\mathcal L'=\mathcal L$, namely $\mathcal L= \mathcal{L}_{1} \cup \mathcal{L}_{2}$. In this case, the restricted rep $D(H)$ from $D(G)$ is reducible, we call such $D(G)$ a {\it composite irrep}.

The above proof is valid no matter if $H$ is an unitary group or an anti-unitary, and no matter if the element $r$ is unitary or anti-unitary. \\

The structure of the simple irreps in case (I) is relatively simple since $D(H)$ is already irreducible. Here we discuss the structure of the composite irreps in case (II).

Fixing a set of bases $[\Phi] = [\phi_1, ..., \phi_n]$ in $\mathcal{L}_{1}$ (assuming the dimensionality of $\mathcal{L}_{1}$ is $n$), and denote $d(H)$ as the $d$-dimensional representation of $H$ on $\mathcal{L}_{1}$. In the other subspace $\mathcal{L}_{2}$, we simply choose the bases as  $r[\Phi] = [r\phi_1, ..., r\phi_n]$. Then in the space $\mathcal L$ the particle-hole operator $r$ is represented as
\Beq
D(r)K_{s(r)}=\Bmat 0& \omega_2(r,r) d(r^2) \\ I_n & 0 \Emat K_{s(r)},
\Eeq
where $r^2\in H$. For any $u\in H$, we have
\Beq
D(u)K_{s(u)}=\begin{pmatrix} d(u)& 0\\ 0 & d'(u)\end{pmatrix}K_{s(u)}.
\Eeq
with
\Beq
d'(u)={\omega_2(u,r)\over\omega_2(r,r^{-1}gr)} K_{s(r)} d(r^{-1}u r)K_{s(r)}.\\
\Eeq
\end{proof}

Now we apply the above theorem to the little co-group $F_f(k)$ at momentum $\bm k$ that is invariant under $\tp$, with  $F_f(k)=G_f(k)+\tp G_f(k)$ and  
\[
F_f(k)/G_f(k) \cong Z_2.
\]
Denoting $D(F_f(k))$ as a projective irrep of $F_f(k)$, according to theorem \ref{thrm:1} the restricted rep $D(G_f(k))$ is either irreducible or a direct sum of two irreps. In the following we discuss three different situations of $F_f(k)$.

(1) $F_f(k)$ is an unitary group. The simple irrep of $F_f(k)$ either contribute a set of irreducible zero modes (if there are effective charge conjugation symmetry) or couple with another simple irrep to form a set of RMNZM; the composite irrep of $F_f(k)$ gives rise to a set of irreducible non-zero modes, where the energy eigenspaces $\mathcal L_1$ and $\mathcal L_2$ have opposite energies $\pm\veps_k$. 

(2) $F_f(k)$ is an anti-unitary group, $G_f(k)$ is unitary and $\tp$ is anti-unitary.  The irrep of $F_f(k)$ falls into three classes, namely the $\mathbb R, \mathbb C,\mathbb H$ classes.  The irrep in the $\mathbb R$ class is a simple irrep; the $\mathbb C$ and $\mathbb H$ classes are composite irreps (irreducible nonzero modes), the two classes are distinguished by whether the representations carried by $\mathcal{L}_{1}$ is equivalent to that of $\mathcal{L}_{2}$ or not)\cite{CMP}. 

(3) $F_f(k)$ is anti-unitary, $G_f(k)$ is anti-unitary and $\tp$ is unitary. Similar to the case where $F_f(k)$ is unitary, the irrep of $F_f(k)$ can be simple or composite. 

The case where both $G_f(k)$ and $\tp$ are anti-unitary can be categorized into the case (2) by redefining $\tp$ to be a unitary coset representative.

\section{Reducible minimal nonzero modes}\label{app:RMNZMs}

A set of reducible minimal nonzero modes (RMNZM) contains two simple irreps of $F_f(k)$. To ensure that the two irreps can  couple with each other to form nonzero modes, the restricted reps of the normal subgroup $G_f(k)$ should be equivalent. Now we start with two irreps $D(F_f(k))$ and $D'(F_f(k))$ satisfying the relation
\Beq
D(G_f(k))=D'(G_f(k)).
\Eeq
Defining a matrix 
$$
X\equiv D(\tp)D'^{-1}(\tp) ,
$$ then we have the following lemma:\\
\begin{lemma}\label{lemma2.1}
If two simple irreps $D(F_f(k))$ and $D'(F_f(k))$ satisfies the relation $D(G_f(k))=D'(G_f(k))$, then the unitary matrix $X= D(\tp)D'^{-1}(\tp)$ commutes with the restricted rep $D(G_f(k))$. \\
\end{lemma}

\begin{proof}
\rm
For any $h\in G_f(k)$, one has
\Beq
D(\tp)K_{s(\tp)}D(h) K_{s(h)}&=&\omega_2(\tp,h)D(\tilde Ph)K_{s(\tp h)} \\ 
&=&D(\tilde P)K_{s(\tp)}D'(h)K_{s(h)} \\
&=& X D'(\tp)D'(h)K_{s(\tp h)}  \\
&=&\omega_2(\tp,h) X D'(\tp h)K_{s(\tp h)},
\Eeq
we have
\beq\label{relation1}
D(\tp h)K_{s(\tp h)}=X D'(\tp h)K_{s(\tp h)}.
\eeq

On the other hand, 
\Beq
&&\omega_2(h,\tp^{-1})D(h\tp^{-1})K_{s(\tp h)}\\
&=&D(h)K_{s(h)}D(\tp^{-1})K_{s(\tp)}\\
&=&D(h)K_{s(h)}\omega_2(\tp^{-1},\tp)K_{s(\tp)}D^{-1}(\tp)\\
&=&D'(h)K_{s(h)}\omega_2(\tp^{-1},\tp)K_{s(\tp)}D^{-1}(\tp)\\
&=&D'(h)K_{s(h)}\omega_2(\tp^{-1},\tp)K_{s(\tp)}\left[D'^{-1}(\tp)D'(\tp)\right]D^{-1}(\tp)\\
&=&D'(h)K_{s(h)}D'(\tp^{-1})K_{s(\tp)}D'(\tp)D^{-1}(\tp)\\
&=&\omega_2(h,\tp^{-1})D'(h\tp^{-1})K_{s(\tp h)}X^{-1},
\Eeq
so we have
\beq\label{relation2}
D(h\tp^{-1})K_{s(\tp h)}=D'(h\tp^{-1})K_{s(\tp h)}X^{-1}.
\eeq

Using the above two relations (\ref{relation1}) and (\ref{relation2}), for any $h_1,h_2\in G_f(k)$, we have 
\Beq
&&D(\tp h_1)K_{s(\tp h_1)}D(h_2\tp^{-1})K_{s(\tp h_2)}\\
&=&\omega_2(\tp h_1,h_2\tp^{-1})D(\tp h_1h_2\tp^{-1})K_{h_1h_2}\\
&=&\omega_2(\tp h_1,h_2\tp^{-1})D'(\tp h_1h_2\tp^{-1})K_{h_1h_2}.
\Eeq
On the other hand, 
\Beq
&&D(\tp h_1)K_{s(\tp h_1)}D(h_2\tp^{-1})K_{s(\tp h_2)}\\
&=&X D'(\tp h_1)K_{s(\tp h_1)}D'(h_2\tp^{-1})K_{s(\tp h_2)}X^{-1}\\
&=&\omega_2(\tp h_1,h_2\tp^{-1})X D'(\tp h_1h_2\tp^{-1})K_{h_1h_2}X^{-1}.
\Eeq
Denoting $\tp h_1h_2\tp^{-1} = g \in G_f(k)$,  then we have $D'(g)K_{s(g)}=X D'(g)K_{s(g)}X^{-1}$, namely
\Beq
D(g)K_{s(g)}=X D(g)K_{s(g)}X^{-1}.
\Eeq
Since $h_1,h_2$ are arbitrary elements in $G_f(k)$, we conclude that for any $g\in G_f(k)$ the unitary matrix $X$ commutes with the restricted rep $D(G_f(k))$.\\ 
\end{proof}

From the above lemma\ref{lemma2.1}, if two 
simple irreps of $F_f(k)$ only differ by the reps of $\tp$, then the difference $X=D(\tp)D'^{-1}(\tp)$ must belong to the centralizer of $D(G_f(k))$. This is a necessary condition under which two 
simple irreps can couple with each other to yield a set of reducible nonzero modes. In the following, we will list all possible cases that may lead to a set of RMNZM.

\subsection{When the $\tp$ is unitary (chiral-like, such as $\tp=\mathcal{TP}$)}

From $D(\tp)=XD'(\tp)$ one has $$[D(\p)]^2=X D'(\p)X D'(\p).$$
On the other hand since $\tp^2\in G_f(k)$, we have $\omega_2(\p,\p)D(\p^2)=\omega_2(\p,\p)\m(\p^2),$ which yields
\Beq
[D(\p)]^2=[D'(\p)]^2. 
\Eeq
From the above two relations, we have $D'(\p)=X D'(\p)X$, or equivalently 
\beq\label{XDX}
D(\p)=X D(\p)X.
\eeq

\subsubsection{\bf when the $G_f(k)$ is unitary}\label{Appsubsubsection:Gfuni}

If $G_f(k)$ is unitary, the unitary elements in the centralizer of $D(G_f(k))$ form a $U(1)$ group $\{ e^{\mi\theta}I;\theta\in[0,2\pi)\}$ according to Schur's lemma.  Since $X$ belongs to the unitary centralizer, we have $X=\eta_{\tp} I$ with $|\eta_{\p}|=1$, namely $D(\p)=\eta_{\tp} D'(\p)$.

From (\ref{XDX}) we have $\eta_{\tp}^2=1$, namely $\eta_{\tp}=\pm 1$. So there are only two possible partners: 
$$
D'(\tp)=\pm D(\p).
$$

(1) If $\eta_{\tp}=1$, the $D'(F_f(k))$ is a copy of $D(F_f(k))$. Hence the direct sum $\tilde{D}(F_f(k))$ of two $D(F_f(k))$ is written as 
\Beq
\tilde{D}(g)= I_2\otimes D(g),\ \ 
\tilde{D}(\tp)= I_2\otimes D(\p),
\Eeq
with $g\in G_f(k)$. Now we solve the allowed Hamiltonian which falls in the centralizer of $\tilde{D}(G_f(k))$, namely 
$\begin{footnotesize}\left\{\begin{pmatrix}
x_1 & x_2\\
x_3 & x_4
\end{pmatrix}\otimes I,x_i\in \mathbb{C}\right\}\end{footnotesize}$. Moreover, the anti-commutation relation between the Hamiltonian and $\tilde{D}(\p)$ further yields
\begin{align}
\begin{pmatrix}
x_1 & x_2\\
x_3 & x_4
\end{pmatrix}=-\begin{pmatrix}
x_1 & x_2\\
x_3 & x_4
\end{pmatrix}\notag
\end{align}
So, $x_1=x_2=x_3=x_4=0$. The solutions of the Hamiltonian can only be the zero matrix $h_k=0$. Namely, this case cannot yield reducible nonzero modes.

(2) If $\eta_{\tp}=-1$, the reps $D(\tp)$ and $-D(\tp)$ together can obtain a nonzero energy: the direct sum $\tilde{D}(F_f(k))$ of two reps is 
\Beq
\tilde{D}(g)= I_2\otimes D(g),\ \ 
\tilde{D}(\tp)= \mu_z \otimes D(\p),
\Eeq
so there exists a nonzero hermitian Hamiltonian $h_k=\veps_k \mu_x\otimes I$ commuting with $\tilde{D}(G_f(k))$ and anti-commuting with $\tilde{D}(\tp)$. Namely, this case corresponds to RMNZM.

\subsubsection{when the $G_f(k)$ is anti-unitary}\label{Appsubsubsec:Gfanti}

If $G_f(k)$ is anti-unitary, the centralizer of $D(G_f(k))$ can be isomorphic to the algebra of the real numbers $\mathbb{R}$, complex numbers $\mathbb{C}$ or quaternions $\mathbb{H}$. Let $H_f(k)$ denote its maximum unitary subgroup. This will lead to three different situations:

(b1) When $D(G_f(k))$ is of the real class $\mathbb{R}$, the restricted rep $D(H_f(k))$ is still irreducible. Moreover, the centralizer of $D(G_f(k))$ only contains two elements $\{\pm I\}$. Therefore, similar to the case when $G_f(k)$ is unitary, one has $\eta_{\tp} = \pm1$ and accordingly $D'(\tp)=\pm D(\tp)$.

For $\eta_{\tp}=1$, the direct sum of two $D(F_f(k))$ reps is 
$$
\tilde{D}(g)K_{s(g)}=I_2\otimes D(g)K_{s(g)},\ \ 
\tilde{D}(\tp)=I_2\otimes D(\p),
$$ 
where $g\in G_f(k)$, so the centralizer of the direct sum of two $D(Q_{k})$ is 
$\left\{\begin{pmatrix}
x_1 & x_2\\
x_3 & x_4
\end{pmatrix}\otimes I,x_i\in \mathbb{R}\right\}$ and the condition of anti-commuting with $\tilde D(\tp)$ also makes the solution be zero. This case does not yield reducible nonzero modes.

For $\eta_{\tp}=-1$,  the direct sum $\tilde{D}(F_f(k))$ of two reps is 
$$
\tilde{D}(g)K_{s(g)}=I_2\otimes D(g)K_{s(g)},\ \ \tilde{D}(\tp)=\mu_z\otimes D(\tp),
$$ 
so there also exists a nonzero Hamiltonian $h_k=\veps_k\mu_x\otimes I$ commuting with $\tilde{D}(G_f(k))$ and anticommuting with $\tilde{D}(\tp)$. This gives rise to a set of RMNZM.

(b2) When $D(G_f(k))$ is of the complex class $\mathbb{C}$, the $D(H_f(k))$ is a direct sum of two non-equivalent irreps. The centralizer of $D(G_f(k))$ takes the form $\left\{\begin{pmatrix}aI&0\\ 0& a^*I\end{pmatrix},a\in\mathbb{C}\right\}$. Hence $X$ belongs to the $U(1)$ group $X\in \left\{e^{\mi \theta\Omega_z},\theta\in[0,2\pi)\right\}$ formed by unitary elements in the centralizer. 

Letting $X=e^{i\theta_0\Omega_z}$,  since $X$ falls in the centralizer of $D(G_f(k)$, for any $g\in G_f(k)$, one has
\begin{align}\label{C-class:QpgCommuteGamma}
e^{\mi\theta_0\Omega_z} D(g)K_{s(g)} = D(g)K_{s(g)} e^{\mi\theta_0\Omega_z}.
\end{align}
On the other hand, equation (\ref{XDX}) yields
\begin{align}\label{C-calss:GammaP}
e^{-\mi\theta_0\Omega_z}D(\tp)=D(\tp)e^{\mi\theta_0\Omega_z}.
\end{align}

Using $e^{\mi\theta_0\Omega_z}=\cos\theta_0+\mi\sin\theta_0\Omega_z$, the equation (\ref{C-calss:GammaP}) can be written as
\beq\label{anticommuP}
-\sin\theta_0(\mi\Omega_z)D(\tp)=\sin\theta_0D(\tp)(\mi\Omega_z).
\eeq
and the equation (\ref{C-class:QpgCommuteGamma}) is equivalent to
\beq\label{center}
(\mi\Omega_z) D(g)K_{s(g)} = D(g)K_{s(g)} (\mi\Omega_z).
\eeq

1) If  the following relation holds, 
$$
D(\p) \Omega_z = -\Omega_z D(\tp),
$$ 
then any value $\theta_0\in [0,2\pi)$ satisfies (\ref{anticommuP}). Actually, $D'(\tp)=e^{-\mi\Omega_z{\theta_0}}D(\tp)$ is equivalent to $D(\tp)$ for any $\theta_0\in[0,2\pi)$ because 
\beq\label{thetatrans}
e^{-\mi\theta_0\Omega_z}D(\tp) = e^{-\mi\Omega_z{\theta_0\over 2}}D(\tp)e^{\mi\Omega_z{\theta_0\over 2}}.
\eeq   
Namely, under the transformation $e^{-\mi\Omega_z{\theta_0\over 2}}$, $D(\tp)$ is transformed into $D'(\tp)=e^{-\mi\theta_0\Omega_z}D(\tp)$ without affecting the reps of elements in $G_f(k)$. In the case $\theta_0=0$, the nonzero hermitian hamiltonian can be chosen as $h_k=\veps_k\begin{pmatrix}0 & -\mi\Omega_z\\ \mi\Omega_z & 0\end{pmatrix} = \veps_k \mu_y\otimes \Omega_z$. The Hamiltonian with nonzero $\theta_0$ can be obtained with the transformation (\ref{thetatrans}) of the second rep. This gives rise to a set of RMNZM.

2) If $$D(\p) \Omega_z \not= -\Omega_z D(\tp),$$ 
then (\ref{anticommuP}) requires that $\theta_0$ is either $0$ or $\pi$. Namely, there are only two possibilities: $D'(\p)=\pm D(\p)$. 

Firstly we consider the case $D'(\p)=D(\p)$. In later discussion, we define 
$$
A\equiv\{\mi\Omega_z,D(\tp)\}=(\mi\Omega_z)D(\p)+D(\p)(\mi\Omega_z)\not=0.
$$  

Supposing that $h_k=\begin{pmatrix}0 &H_0\\H_0^\dag & 0\end{pmatrix}$ is a possible Hamiltonian for the rep 
$$
D(g)K_{s(g)}=I_2\otimes D(g)K_{s(g)},\ \ 
D(\tp)= I_2\otimes D(\tp).$$   
It is easy to see that $H_0$ falls in the centralizer of $D(G_f(k))$, hence 
$$
H_0=r e^{\mi\alpha\Omega_z} = r(\cos\alpha + i\alpha\Omega_z)
$$ 
with $r\in\mathbb{R}, \alpha\in[0,2\pi)$. Meanwhile, $H_0$ anti-commutes with the $D(\p)$, $i.e.$
\begin{align}\label{C-calss:H0general}
D(\p)(\cos\alpha \!+\! \mi\sin\alpha\Omega_z)\!\!=\!\!- \left( \cos\alpha \!+\! \mi\sin\alpha\Omega_z \right)D(\p).
\end{align}
We will show that if $r\not=0$,  there exist a contradiction. For $r\not=0$, the above equation (\ref{C-calss:H0general}) is equivalent to 
\Beq
2(\cos\alpha) D(\p)=-\sin\alpha\left[(\mi\Omega_z)D(\p)+D(\p)(\mi\Omega_z)\right],\Eeq
namely, 
$$
2(\cos\alpha) D(\p)=-(\sin\alpha) A.
$$
If there exists a solution, then $\sin\alpha\not=0$ (because $\tp$ is invertible) and
\begin{equation}\label{C-class:The solution}
2\cot\alpha=-A[D(\p)]^{-1}.
\end{equation}
Notice that 
\begin{align}
AD(\p)=&(\mi\Omega_z)D^2(\p)+D(\p)(\mi\Omega_z)D(\p)\notag\\
=&D^2(\p)(\mi\Omega_z)+D(\p)(\mi\Omega_z)D(\p)\notag\\
=&D(\p)[D(\p)(i\Omega_z)+(i\Omega_z)D(\p)]=D(\p)A.\notag
\end{align}
So, $[A,D(\p)]=0\Rightarrow AD^{-1}(\p)=D^{-1}(\p)A$. Denoting $B\equiv AD^{-1}(\p)=D^{-1}(\p)A$, $i.e.$ 
\beq\label{DefB}
B=(\mi\Omega_z)+D(\p)(\mi\Omega_z)D^{-1}(\p),
\eeq
then we have $B\not=0$ (since $A\not=0$) and $\operatorname{Tr\ } B=0$. 

Actually, $B$ is also a centre of $D(G_f(k))$. For any $g\in G_f(k)$, 
\begin{align}
&\left[D(g)K_{s(g)}\right] B \left[ D(g)K_{s(g)}\right]^{-1}\notag\\
=& \left[D(g)K_{s(g)}\right]\left[ (\mi\Omega_z)+D(\p)(\mi\Omega_z)D^{-1}(\p)  \right]   \left[ D(g)K_{s(g)}\right]^{-1}\notag\\
=& \left[D(g)K_{s(g)}\right] (\mi\Omega_z) \left[D(g)K_{s(g)}\right]^{-1} \notag\\
&+ \left[D(g)K_{s(g)}D(\p)\right](\mi\Omega_z) \left[D(g)K_{s(g)}D(\p)\right]^{-1}\notag\\
=&(\mi\Omega_z) + \left[D(\p)D(g')K_{s(g')}\right](\mi\Omega_z) \left[D(\p)D(g')K_{s(g')}\right]^{-1}\notag\\
=& (\mi\Omega_z) + D(\p)(\mi\Omega_z)D^{-1}(\p) \notag\\
=& B,
\end{align}
where $g'=\p^{-1}g\p\in G_f(k)$. Thus, $B=\lambda e^{\mi\Phi\Omega_z},\lambda\not=0$.  From (\ref{C-class:The solution}), we learn that 
\begin{align}
2\cot\alpha =- B=-\lambda e^{\mi\Phi\Omega_z},\ \ \lambda\not=0.
\end{align}
Therefore, $\Phi$ is either $0$ or $\pi$. However, in both cases, $\operatorname{Tr\ } B\propto \lambda\not=0 $, which is contradicting with (\ref{DefB}).  So, the parameter $r$ in (\ref{C-calss:H0general}) must be zero, which means $h_k=0$.  This case does not yield reducible nonzero modes.

For the second case, $D'(\tp)=-D(\tp)$, namely for the rep 
$$
\tilde D(g)K_{s(g)} =I_2\otimes D(g)K_{s(g)},\ \ 
\tilde D(\tp) = \mu_z\otimes D(\p), 
$$
there exists a nonzero hermitian Hamiltonian $h_k=\veps_k\mu_x\otimes I$. This gives rise to a set of RMNZM.

(b3) When $D(G_f(k))$ is of the quaternionic class $\mathbb{H}$, the rep $D(H_f(k))$ is a direct sum of two equivalent irreps.   Generally, the centralizer of $D(G_f(k))$ is $\left\{\begin{pmatrix} aI & bI \\ -b^*I & a^*I \end{pmatrix},a,b\in\mathbb{C} \right\}$, whose unitary elements form an $SU(2)$ group $\left\{\begin{pmatrix} aI & bI \\ -b^*I & a^*I \end{pmatrix},aa^*+bb^*=1,a,b\in\mathbb{C} \right\}$.  Since the unitary matrix $X=\left[ D(\p)\m^{-1}(\p)\right]$ falls in the centralizer of $D(G_f(k))$, we can write
\Beq
X=\cos\theta_0+\sin\theta_0\left[ n_x(\mi\Omega_x)+n_y(\mi\Omega_y)+n_z(\mi\Omega_z) \right]
\Eeq
 with $\theta_0\in[0,2\pi), n_x^2+n_y^2+n_z^2=1$. Denoting $\Omega_n\equiv \pmb n\cdot\pmb\Omega=n_x\Omega_x+n_y\Omega_y+n_z\Omega_z$, then
 \begin{equation}\label{H-class:DDpGamman}
X=e^{\mi\theta_0\Omega_n}=\cos\Phi I+\mi\sin\theta_0\Omega_n.
 \end{equation}

From (\ref{XDX}), we have $e^{-\mi\theta_0\Omega_n}D(\tp)=D(\tp)e^{\mi\theta_0\Omega_n}$, 
which is equivalent to 
\beq\label{H-calss:GammaP}
-\sin\theta_0(\mi\Omega_n)D(\p)=\sin\theta_0 D(\p)(\mi\Omega_n).
\eeq

1) When the following relation holds 
$$
D(\p) \Omega_{n_0} = -\Omega_{n_0} D(\tp)
$$ 
for certain $\pmb n_0\in \bm{S}^2$, then any value $\theta_0\in [0,2\pi)$ in the above $\pmb n_0$ direction satisfies (\ref{H-calss:GammaP}). The discussions in the $\mathbb{C}$-class also applies here.  In the case $\theta_0=0$, one possible Hamiltonian is $h_k=\veps_k\begin{pmatrix}0 & \mi\Omega_z\\-\mi\Omega_z & 0\end{pmatrix}=\veps_k\mu_y\otimes \Omega_z$. This gives rise to a set of RMNZM.

2) When $$D(\p) \Omega_{n} \not= -\Omega_{n} D(\tp)$$ is true for all $\pmb n\in S^2$, then (\ref{H-calss:GammaP}) requires that $\theta_0$ is either $0$ or $\pi$, namely, there are only two possibilities $D'(\tp)=\pm D(\tp)$. 

Similar to the discussion in the $\mathbb{C}$-class, the case $D'(\tp)=D(\tp)$ corresponds to reducible zero modes, and the case $D'(\tp)=-D(\tp)$ corresponds to reducible nonzero modes.

\subsection{when the $\tp$ is anti-unitary (particle-hole-like, such as $\tp=\mathcal{IP}$)}\label{Appsubsec:tpanti}

If $G_f(k)$ is anti-unitary, we can multiply $\tp$ with an anti-unitary element $T_0\in G_f(k)$  to obtain a unitary chiral-like symmetry $\tp' = \tp T_0$, and then the discuss for the case with unitary $\tp$ and with anti-unitary $G_f(k)$ can be applied. Therefore, we only need to consider the case in which $G_f(k)$ is unitary.

When $G_f(k)$ is unitary, then the unitary elements in the centralizer of $D(G_f(k))$ form a $U(1)$ group $\{e^{\mi\theta}I; \theta\in[0,2\pi)\}$.  Lemma\ref{lemma2.1} indicates that $X$ belongs to the above $U(1)$ group. Letting $X=e^{\mi\theta_0}$, then $D'(\tp)K=e^{-\mi\theta_0}D(\tp)K$ can be transformed into $D(\tp)K$ by adjusting the phase of the bases. Hence $D'(F_f(k))$ is always equivalent to $D(F_f(k))$. Then there always exists a Hamiltonian  $h_k=\veps_k(\sin{\theta_0\over2}\sigma_x+\cos{\theta_0\over2}\sigma_y)\otimes I_n$   which commutes with $\tilde D(h)=D(h)\oplus D(h)$ and anti-commutes with $\tilde D(\tp)K = D(\tp)\oplus e^{-\mi\theta_0}D(\tp)K$. This gives rise to a set of RMNZM.

\section{Proof of theorem \ref{NoCNoZM} and consequences of $\C$ symmetry} \label{App.:thrm2}

\subsection{\bf Proof of the theorem \ref{NoCNoZM}} 

The theorem states that without `non-diagonal' (effective) charge conjugation symmetry $\C$, no zero modes are protected at any given momentum $k$.

At momentum $k$, the electron creation operators $C^\dag_{k}=[c_{k1}^\dag,c_{k2}^\dag,...c_{kN}^\dag]$ carry a representation of the point group $G_f(k)$ denoted as $M(G_f(k))$. Generally, $M(G_f(k))$ can be reduced into direct sum of irreps 
$$
M(G_f(k)) = d_1(G_f(k)) \oplus d_2(G_f(k)) \oplus ...
$$ 
under the bases $[b_{k1}^\dag, b_{k2}^\dag, ..., b_{kN}^\dag] = [c_{k1}^\dag, c_{k2}^\dag, ..., c_{kN}^\dag] U$. Since $\p$ transforms creation operators into annihilation operators, we denote $[\tilde b_{k1}, \tilde b_{k2}, ..., \tilde b_{kN}] = \widehat{\p} [b_{k1}^\dag, b_{k2}^\dag, ..., b_{kN}^\dag] \widehat{\p} ^{-1}$. Since $\p G_f(k) = G_f(k) \p$, $[\tilde b_{k1}, \tilde b_{k2}, ..., \tilde b_{kN}]$ carry a rep of $G_f(k)$,
\[
\tilde M(G_f(k)) = \tilde d_1(G_f(k)) \oplus \tilde d_2(G_f(k)) \oplus ...
\]
with $\tilde d_a(g) = {\omega_2(g,\tp)\over\omega_2(\tp,\tp^{-1}g\tp) } K_{s(\p)} d_a(\tp^{-1}g\tp) K_{s(\p)}$ for $a=1,2,...$ and $g\in G_f(k)$.

In the direct sum space of $d_a(G_f(k)) \oplus \tilde d_a(G_f(k))$, the group $F_f(k)=G_f(k) + \p G_f(k)$ is presented as
\beq
&&\!\!\!\!\!\!\!\!\!\! D_a(g)K_{s(g)}=\begin{pmatrix} d_a(g) & 0\\ 0 & \tilde d_a(g)\end{pmatrix}K_{s(g)},\ \ \ g\in G_f(k), \label{Rep_g}\\
&&\!\!\!\!\!\!\!\!\!\! D_a(\tp)K_{s(\tp)}=\Bmat 0& \omega_2(\p,\p)d_a(\p^2)\\ I_{n_a} & 0\Emat K_{s(\p)},\label{Rep_P}
\eeq
with $n_a$ the dimensionality of the rep $d_a(G_f(k))$. In this space, there exists a non-degenerate Hamiltonian 
$$h_k= \varepsilon_k \mu_z\otimes I_{n_a}$$ with $\varepsilon_k\neq 0$ 
lifting all the zero modes. Notice that the above rep (\ref{Rep_g}) and (\ref{Rep_P}) of $F_f(k)$ may be reducible if $d_a(g) =\tilde d_a(g)$ for all $g\in G_f(k)$. Therefore, in electron bases the nonzero energy modes have two possible resources:\

\indent (I) a composite irrep $D_a(F_f(k))$, in which the restricted rep $D_a(G_f(k))=d_a(G_f(k))\oplus\tilde d_a(G_f(k))$ is reducible and $d_a\neq\tilde d_a$, give rise to a set of {\it irreducible nonzero modes}; \\
\indent (II) a direct sum of two simple irreps $D_a(F_f(k)) = \mathcal D_+(F_f(k))\oplus \mathcal D_-(F_f(k))$, in which $\mathcal D_+(G_f(k))=d_a(G_f(k)) =\mathcal D_-(G_f(k))=\tilde d_a(G_f(k))$, couple with each other to form a set of {\it reducible minimal nonzero modes}.

The above theorem implies that any simple irrep
of $F_f(k)$ at arbitrary $\pmb k$ can find its coupling partner to form nonzero energy modes. Hence we conclude that without additional `non-diagonal' symmetries like the (effective) charge conjugation symmetry, zero energy modes are unstable at any given $\pmb k$. 

\subsection{Extension and induction: from the rep of a normal subgroup $H$ to the full group $G\rhd H$}

Now we consider the case in which the system has a (effective) charge conjugation symmetry $\mathcal{C}$ at momentum $\bm k$. Denoting 
$$
F_f^c(k) = F_f(k) +  \C F_f(k) 
$$ 
as the complete little co-group at $\pmb k$, then we have $F_f^c(k)/F_f(k)\cong {Z}_2$.  Similarly, we can define $G_{f}^c(k)$ with $G_{f}^c(k)=G_{f}(k) + \mathcal C G_{f}(k)$ and $G^c_f(k)/G_f(k)\cong {Z}_2$. In the presence of $\mathcal{C}$, a set of minimal nonzero modes can be turned into zero modes (called $\C$-zero modes) in the following two ways:  \\
(1)  Two simple irreps contained in RMNZM which are previously coupled become uncoupleable and form reducible minimal $\C$-zero modes;  \\
(2) A set of irreducible nonzero modes or a set of RMNZM are turned into a set of  irreducible $\C$-zero modes.

As will be seen, the above ways correspond to two kinds of operations on irreps, namely, the extension from an irrep of $F_f(k)$ to that of $F_f^{c}(k)$ and the induction from an irrep of $F_f(k)$ to that of $F_f^{c}(k)$.

{\bf Extension} Supposing $H$ is a normal subgroup of $G$, namely $H\lhd G$, and $d_1(H)$ is a $n_d$-dimensional projective irrep of $H$, if there exists an irrep $D_1(G)$ of $G$ having factor system $\omega_2(g_1,g_2), g_{1,2}\in G$ and the same dimensionality $n_d$ such that the restricted rep $D_1(H)$ is irreducible and is identical to $d_1(H)$, namely,
\Beq
D^{(\nu)}(h)=d^{(\nu)}(h),\ \ \mbox{for\ all\ } h\in H,
\Eeq
then we dub $D^{(\nu)}(G)$ as an {\it extension} of irrep $d^{(\nu)}(H)$, and claim $d^{(\nu)}(H)$ to be extendible to $G$ with the factor system $\omega_2$. If such extension $D^{(\nu)}(G)$ exists, then the following two conditions should be satisfied for $d^{(\nu)}(H)$:\\
\indent (A) the irrep $\omega_2^{-1}(g^{-1},g)\omega_2(g^{-1},h)\omega_2(g^{-1}h,g)d_1(g^{-1}hg)$ is equivalent to $K_{s(g)}d_1(h)K_{s(g)}$ for any $g\in G$. Namely,  for a given unitary element $u\in G$ there exist a matrix $X$ such that for any $h\in H$
\[
{\omega_2(u^{-1},h)\omega_2(u^{-1}h,u) \over \omega_2(u^{-1},u)} d_1(u^{-1}hu) = Xd_1(h)X^{-1},
\]
and for a given anti-unitary element $a\in G$ there exist a matrix $Y$ such that for any $h\in H$
\[
{\omega_2(a^{-1},h)\omega_2(a^{-1}h,a) \over \omega_2(a^{-1},a)} d_1(a^{-1}ha) = Yd_1^*(h)Y^{-1}.
\]
\indent (B) $D(g_1)D(g_2)=\omega_2(g_1,g_2)D(g_1g_2)$ for any $g_1,g_2\in G$.

{\bf Induction} Otherwise, if such irrep $D^{(\nu)}(G)$ does not exist, but there exists an irrep $\mathscr D^{(\tilde\nu)}(G)$ having factor system $\omega_2$ and with dimensionality multiple of $n_d$ such that $d^{(\nu)}(H)$ is contained in the restricted rep $\mathscr D^{(\tilde\nu)}(H)$ after it being reduced into direct sum of irreps, then we call $\mathscr D^{(\tilde\nu)}(G)$ an {\it induction} of $d^{(\nu)}(H)$ to $G$ with the factor system $\omega_2$.

For instance, if $G$ is an anti-unitary group with $H$ its maximal unitary subgroup, then the irreps of $G$ with a given factor system $\omega_2$ have three classes, the real class $\mathbb R$, complex class $\mathbb C$ and quaternionic class $\mathbb H$. In the $\mathbb{R}$ class, the restrict rep to $H$ is irreducible. Hence, $\mathbb R$ class is an extension of irrep from $H$ to $G$. On the other hand, the $\mathbb C$ and $\mathbb H$ classes are induction of certain irreps from $H$ to $G$.

\subsection{Consequences of $\C$ symmetry} \label{App:consequenceofC}

Now, we discuss the effect of the (effective) charge-conjugation symmetries $\C$. For simplicity we consider the case in which there is only one charge conjugation operation $\C$.

Firstly, we consider a special case in which $\C$ preserves the supporting space of every minimal nonzero modes $\mathcal L_{\pm\veps_k}$, namely we first assume that $\C$ does not cause additional degeneracy for the energy eigenstates with $\varepsilon_k\not=0$. Such charge-conjugation symmetries all have the following property: {\it for any given minimal nonzero mode, there exists an extension from $F_f(k)$ to $F^c_f(k)$.} 
Recall that a set of minimal nonzero modes is either \\
\indent (I) a set of irreducible nonzero modes which carries a composite irrep of $F_f(k)$ denoted as $D_a(F_f(k))$ with the restricted rep $D_a(G_f(k))=d_a(G_f(k))\oplus\tilde d_a(G_f(k))$, or \\
\indent (II) a set of RMNZM which form a direct sum of two simple irreps of $F_f(k)$ denoted as $D_a(F_f(k))=\mathcal D_+(F_f(k))\oplus \mathcal D_-(F_f(k))$). 

Then the presence of $\C$ can affect the set of minimal nonzero modes in different ways. Here we discuss them separately. From now on, we will omit the subscript $_a$ in the labels of reps $D_a$ and $d_a, \tilde d_a$. \\

(I) For the irreducible nonzero modes $D(F_f(k))$. 

1) If the irrep $D(G_f^c(k))$ is an induction from $d(G_f(k))$ (or from $\tilde d(G_f(k))$), namely, 
there does not exist any extension from $d(G_f(k))$ or $\tilde d(G_f(k))$ to $G_f^{c}(k)$, then $D(F^c_f(k))$ corresponds a set of {\it irreducible zero modes} of $F_f^c(k)$.

2) If there exists an extension from $d(G_f(k))$ (or $\tilde d(G_f(k)))$ to $d(G^c_f(k))$ (or $\tilde d(G^c_f(k))$), then $D(F^c_f(k))$ remains to be an {\it irreducible nonzero mode} of $F_f^{c}(k)$.

(II) For the RMNZM $D(F_f(k)) = \mathcal D_+(F_f(k))\oplus \mathcal D_-(F_f(k))$. 

1) If there exist an extension from $\mathcal D_{\pm}(F_f(k))$ to $\mathcal D_{\pm}(F_f^c(k))$ (hence there must exist an extension from $d(G_f(k))$ to $d(G_f^c(k))$), then $\mathcal D_{\pm}(F^c_f(k))$ 
are two simple irreps of $F_f^c(k)$.  If $\mathcal D_{\pm}(F^c_f(k))$ satisfy the conditions listed  in Tab.\ref{tab:RMNZM}, namely if they can couple to each other, then $\mathcal D_{+}(F^c_f(k))\oplus \mathcal D_{-}(F^c_f(k))$ remains to be a set of {\it RMNZM}; on the other hand, if they cannot couple to each other, (for example, $\C$ provides different quantum numbers in $\mathcal D_{+}(F^c_f(k))$ and $\mathcal D_{-}(F^c_f(k))$ to prevent them from coupling to each other), then $\mathcal D_{+}(F^c_f(k))\oplus \mathcal D_{-}(F^c_f(k))$ will form  {\it reducible minimal zero modes} of $F^c_f(k)$.

2) If $\C$ couples $\mathcal D_{+}(F_f(k))$ with $\mathcal D_{-}(F_f(k))$ such that the direct sum space $\mathcal D_{+}(F_f(k))\oplus \mathcal D_{-}(F_f(k))$ forms a simple irrep $D(F_f^c(k))$ for $F_f^c(k)$ (namely $D(F_f^c(k))$ and
$D(G_f^c(k))$ are both irreducible), then $D(F_f^c(k))$ forms {\it irreducible zero modes} of $F_f^c(k)$. 

3) If $\C$ couples $\mathcal D_{+}(F_f(k))$ with $\mathcal D_{-}(F_f(k))$ such that the direct sum space $\mathcal D_{+}(F_f(k))\oplus \mathcal D_{-}(F_f(k))$ forms a composite irrep $D(F_f^c(k))$ for $F_f^c(k)$ (namely $D(F_f^c(k))$ is irreducible but $D(G_f^c(k))$ is reducible), then $D(F_f^c(k))$ will form {\it irreducible nonzero modes} of $F_f^c(k)$. \\ 

Secondly, we consider a different case in which $\C$ `joins' a set of minimal nonzero modes $\mathcal L_{\pm\veps'_k}$ with another set of minimal nonzero modes $\mathcal L_{\pm\veps''_k}$. Actually, since $G_f^c(k)/G_f(k) = F_f^c(k)/F_f(k) \cong {Z}_2$, as illustrated in theorem \ref{thrm:1} the charge conjugation $\C$ joins at most two irreps of $G_f(k)$ to form an irreducible rep of $G_f^c(k)$, namely $\C$ involves at most two sets of minimal nonzero modes. The direct sum space $\mathcal L = \mathcal L_{\pm\veps'_k}\oplus \mathcal L_{\pm\veps''_k}$ form a representation space of the group $F_f^c(k)$ as well as $G_f^c(k)$. If $\mathcal L$ contains two irreps of $G_f^c(k)$, namely if $\mathcal L$ carries a composite irreps of $F_f^c(k)$ or a direct sum of two simple irreps (which are couplable partner of each other) of $F_f^c(k)$, then $\mathcal L$ forms an enlarged set of minimal nonzero modes of $F_f^c(k)$ with doubled degeneracy; otherwise, if $\mathcal L$ contains four irreps of $G^c_f(k)$, then the degeneracy of the nonzero modes will not be enlarged. The $\C$ enlarged degeneracy occurs in Sec.\ref{sec:p4gmz2T}, where the $\C$ symmetry causes the 2-fold degeneracy of the bands. Such kind of charge conjugation $\C$ does not yield zero modes in general. 

In the above discussion, we listed the cases in which zero modes can appear due to the presence of charge conjugation symmetry $\C$.  We call these zero modes as {\it `minimal $\mathcal{C}$-zero modes'}, including the {\it irreducible $\C$-zero modes} and the {\it reducible minimal $\mathcal{C}$-zero modes}.  A natural question is,  given a set of minimal nonzero modes at $k$, is it always possible to construct a $\mathcal{C}$ to turn them into minimal $\C$-zero modes? The answer is positive, as stated in the theorem \ref{Thm:Czeromode}. The proof of this theorem is given in the following App. \ref{App.:thrm3}.

\section{The proof of the theorem \ref{Thm:Czeromode}}\label{App.:thrm3}

We prove theorem \ref{Thm:Czeromode} by illustrating the existence of $\C$ to turn any given minimal nonzero modes of $F_f(k)$ at momentum $\bm k$ into minimal $\C$-zero modes. We assume that $\C$ preserves the space of minimal nonzero modes of $F_f(k)$.

\subsection{When the $\tp$ is unitary}

When $\tp$ is unitary, the simplest method to obtaining zero modes is to set $\mathcal{C} = \tp$ (or $\mathcal{C} = \tp g$ where $g\in G_f$ is unitary). In this case all minimal nonzero modes at momentum $\bm k$ are turned into zero modes. In the following we will discuss other solutions of $\mathcal C$ which turn certain nonzero modes into zero modes.

\subsubsection{When $G_f(k)$ is unitary}\label{Gfunitary}

{\it Extension: minimal $\C$-zero modes from RMNZM $\mathcal D_{+}(F_f(k))\oplus \mathcal D_{-}(F_f(k))$.}  In this case, we have $D_{+}(\tp)\not\cong D_{-}(\tp)$ and $\mathcal D_{+}(G_f(k))\cong \mathcal D_{-}(G_f(k))$.  When $\C$ is unitary then the induced rep from $\mathcal D_+(G_f(k))$ or $\mathcal D_-(G_f(k))$ to $G^c_f(k)$ is reducible. Namely, it is not possible to do an induction from $\mathcal D_+(G_f(k))$ or $\mathcal D_-(G_f(k))$ such that the induced irrep of $G^c_f(k)$ is equivalent to $\mathcal D_{+}(G_f(k))\oplus \mathcal D_-(G_f(k))$ when restricted to $G_f(k)$. 
Hence, minimal $\C$-zero modes can only be obtained by extensions from $\mathcal D_{\pm}(F_f(k))$ to $\mathcal D_{\pm}(F_f^c(k))$.

There are different ways to do the $Z_2$ extensions, we only consider the simplest case by assuming that the effective charge-conjugation $\C$ commutes with $F_f(k)$, namely $[\C, F_f(k)]=0$. 
In the direct sum space $\mathcal D_{+}(F_f(k))\oplus \mathcal D_{-}(F_f(k))$, $\C$ is reducible. Consequently one obtains two non-equivalent extension of $\mathcal D_{+}(G_f(k)) = \mathcal D_{-}(G_f(k))$ by the two 1D reps of ${Z}_2$, resulting in $\mathcal D_{+}(G_f^c(k))\not\cong \mathcal D_{-}(G_f^c(k))$ (here $\C$ provides two different quantum numbers in $\mathcal D_{+}(G_f^c(k))$ and $\mathcal D_{-}(G_f^c(k))$ respectively).  Considering $\mathcal D_{+}(\tp)\not\cong \mathcal D_{-}(\tp)$, now one obtains a set of {\it reducible minimal $\C$-zero modes} for the group $F^c_f(k)$. This construction is applied in several lattice models in Sec.\ref{Sec:ModelSection}.

{\it Induction: minimal $\C$-zero modes from irreducible nonzero modes $D(F_f(k))$.} In this case, the restricted rep $D(G_f(k))$ from the irrep $D(F_f(k))$ can be reduced into a direct sum of two nonequivalent irreps, namely $D(G_f(k)) = d_1(G_f(k))\oplus d_2(G_f(k))$. It is not possible to do extension on subgroup $G_f(k)$ to obtain a set of minimal zero modes. But a ${Z}_2$ induction is straightforward if $\C$ is closed in the space $D(G_f(k))$. Then $D(\C)$ is irreducible in the space $D(G_f(k))$ since $d_1(G_f(k))$ and $d_2(G_f(k))$ are non-equivalent. Thus obtained $D(G_f^c(k))$ forms a set of {\it irreducible $\C$-zero modes. }

\subsubsection{When $G_f(k)$ is anti-unitary} 

We denote the maximum unitary subgroup of $G_f(k)$ as $H_f(k)$.

{\it Induction from irreducible nonzero modes $D(F_f(k))$.} Since the restricted rep $D(G_f(k))$ is reducible, assume that $D(G_f(k))=d_{1}(G_f(k))\oplus d_{2}(G_f(k))$. No matter $d_1(G_f(k))$ and $d_2(G_f(k))$ are equivalent or not, and no matter which class $\mathbb{R},\mathbb{H},\mathbb{C}$ they belong to, a ${Z}_2$ induction to is always possible. 

We first present the constraints of the charge-conjugation $\mathcal{C}$ if the induction is applicable.  Since the charge-conjugation $\mathcal{C}$ permutes the two subspace $d_1(G_f(k))$ and $d_2(G_f(k))$, we assume that in the space $D(G_f(k))$ the charge conjugation $\C$ is represented as
\beq
&&\!\!\! D(\mathcal{C})=\begin{pmatrix} 
0& \omega_2(\mathcal{C},\mathcal{C})d_1(\mathcal{C}^2)[L(\mathcal{C})]^{-1}\\ L(\mathcal{C}) & 0
\end{pmatrix},\label{App.D:theinducedCC}
\eeq
where $L(\mathcal{C})$ is an invertible matrix. Then from the theory of induced rep, for any $g\in G_f(k)$ we have
\beq
&&\!\!\!\!\!\!\!\!\!\!\!\! D(g)K_{s(g)}=\notag\\
&&\!\!\!\!\!\!\!\!\!\!\!\!\!\!\!\begin{pmatrix}
d_1(g)\!K_{s(g)}& 0\\
0 & {\omega_2(g,\mathcal{C})\over\omega_2(\mathcal{C},\mathcal{C}^{-1}g\mathcal{C}) }\!  L(\mathcal{C})d_1\!(\mathcal{C}^{-1}\!g\mathcal{C})K_{s(g)}\!L^{-1}(\mathcal{C}) 
\end{pmatrix}.\label{App:F2}
\eeq
Since $D(G_f(k))=d_{1}(G_f(k))\oplus d_{2}(G_f(k))$, for any $g\in G_f(k)$ we have,
\beq\label{d1d2}
{\omega_2(g,\mathcal{C})\over\omega_2(\mathcal{C},\mathcal{C}^{-1}\!g\mathcal{C})}\!L(\mathcal{C})d_1(\mathcal{C}^{-1}\!g\mathcal{C}) \!K_{s(g)}\! L\!^{-1}\!(\mathcal{C})\! = \!d_2(g)\!K_{s(g)}.
\eeq
Especially, for unitary elements $h\in H_f(k)$, one has
\Beq\label{App.D:case1Cond1}
\operatorname{Tr}\left[ {\omega_2(h,\mathcal{C})\over\omega_2(\mathcal{C},\mathcal{C}^{-1}h\mathcal{C}) }d_1(\mathcal{C}^{-1}h\mathcal{C})\right]=\operatorname{Tr}\left[ d_2(h)\right].
\Eeq

Supposing that $D(G_f^c(k))$ form an irrep of $G_f^c(k)$, then there is only one hermitian centralizer for $D(G_f^c(k))$, namely the identity matrix. If $d_2(G_f(k))\not\cong d_1(G_f(k))$, then the induced rep $D(G_f^c(k))$ is always irreducible; if $d_2(G_f(k))=d_1(G_f(k))$, to make $D(G_f^c(k))$ irreducible the following condition should be satisfied,
\beq\label{App.D:case1Cond2}
L(\mathcal{C})\not=  \omega_2(\mathcal{C},\mathcal{C})d_1(\mathcal{C}^2)L^{-1}(\mathcal{C}).
\eeq
The simplest example for an induction with $d_2(G_f(k))=d_1(G_f(k))$ is the group $G_f^c(k)=G_f(k)+{\C}G_f(k)$ with $\C$ commuting with $G_f(k)$, $\omega_2(\mathcal{C},\mathcal{C})=1$ and $L(\mathcal{C}) =\mi I$, namely $D(\C)=\Gamma_y=\tau_y\otimes I$. Since $D(\C)$ cannot be block diagonalized via real matrix, $D(\C)$ and the rep matrices of antiunitary elements in $G_f(k)$ cannot be simultaneously block diagonalized, meaning that $D(G_f^c(k))$ is irreducible. Hence $D(F_f^c(k))$ correspond to a set of {\it irreducible $\C$-zero modes}. This kind of induction is applied in the construction of lattice models in Sec.\ref{Sec:ModelSection}. 

{\it Remark: }We comment that for unitary $G_f(k)$ (without using Frobenius reciprocity) with $d_{1}(G_f(k))= d_{2}(G_f(k))$, the induction will be impossible.  To see this, notice that the condition (\ref{App.D:case1Cond2}) becomes the following one for unitary groups,
\beq\label{app:unitary}
L(\mathcal{C})\not\propto  \omega_2(\mathcal{C},\mathcal{C})d_1(\mathcal{C}^2)[L(\mathcal{C})]^{-1}.
\eeq
However, from (\ref{d1d2}) one can also show that $\left[d_1(\mathcal{C}^2) L^{-2}(\mathcal{C})\right]$ commutes with the irreducible representation $d_1(G_f(k))$. According to Schur's lemma, $\left[d_1(\mathcal{C}^2) L^{-2}(\mathcal{C})\right]$ are proportional to identity matrix, which is contradict with the above condition (\ref{app:unitary}).\\

{\it Minimal $\C$-zero modes from the RMNZM $ \mathcal D_{+}(F_f(k))\oplus \mathcal D_{-}(F_f(k))$.} In this case, $\mathcal D_{+}(G_f(k))$ and $\mathcal D_{-}(G_f(k))$ are equivalent. One can always do  inductions following the previous construction for the irreducible nonzero modes with $d_1(G_f(k))=d_2(G_f(k))$. 

{\it Remark:} We comment that, for $\mathcal D_{\pm}(G_f(k))$ belong to the $\mathbb{C}$ class and $\mathcal D_{+}(F_f(k))\cong \mathcal D_{-}(F_f(k))$ [corresponding to the case (4) in the Tab.\ref{tab:RMNZM}], one can just do a ${Z}_2$ induction on each irreducible blocks for the subgroup $H_{f}(k)$. Then the new irrep of the group $G_f^c(k)$ is of the $\mathbb{R}$ class.  This also leads to a set of {\it reducible minimal $\mathcal{C}$-zero modes} according to the case (2) in the Table.\ref{tab:RMNZM}.
Moreover,  one can also do extensions to obtain a set of {\it reducible minimal $\C$-zero modes} following the discussions in the subsection \ref{Gfunitary}.

\subsection{When the $\tp$ is anti-unitary}

In this case, we only need to consider the case in which $G_f(k)$ is unitary (otherwise it reduces to the case with unitary $\tp$). A set of RMNZM are formed by two equivalent representations, namely $\mathcal D_+^{\mathbb{R}}(F_f(k))\oplus \mathcal D_-^{\mathbb{R}}(F_f(k))$ both of which belong to the $\mathbb{R}$ class, and a set of irreducible nonzero modes are all irreps of $F_f(k)$ which belong to the $\mathbb{H}$ or the $\mathbb{C}$ class. We denote them by $D^{\mathbb{H}}(F_f(k))$ and $D^{\mathbb{C}}(F_f(k))$ respectively.

{\it For the RMNZM $ \mathcal D_+^{\mathbb{R}}(F_f(k))\oplus \mathcal D_-^{\mathbb{R}}(F_f(k))$.} Similar to the case where $\tp$ is unitary, a ${Z}_2$ induction is impossible, but a ${Z}_2$ extensions is possible to turn the set of RMNZM into a set of {\it reducible minimal $\mathcal{C}$-zero modes. 
}

{\it For the irreducible nonzero modes $D^{\mathbb{C}}(F_f(k))$.} Similar to the case where $\tp$ is unitary, a $Z_2$ induction is possible to turn the set of irreducible nonzero modes into a set of {\it irreducible $\mathcal{C}$-zero modes.} 

{\it For the irreducible nonzero modes $D^{\mathbb{H}}(F_f(k))$.} In this case, a single charge operation $\C$ fails to turn the set of irreducible nonzero modes $D^{\mathbb{H}}(F_f(k))$ into $\mathcal{C}$-zero modes. 

But we can adopt two charge conjugation operators $\mathcal{C}_1,\mathcal{C}_2$ to obtain a set of $\mathcal{C}$-minimal zero modes. In this case we have $F_f^{c_1\times c_2}(k)/F_f(k)\cong {Z}_{2}\times {Z}_{2}$, and the representations of the two charge operators $\mathcal{C}_1$ and $\mathcal{C}_2$ should anti-commute with each other $ D(\mathcal{C}_1)D(\mathcal{C}_2)=-D(\mathcal{C}_2)D(\mathcal{C}_1) $ from the Clifford theory\cite{PRofFinie}.

The process can be understood as an extension followed by an induction. We first do a ${Z}_2$ extension by the $\C_1$ to make the two irreps contained in $D^{\mathbb{H}}(H_f(k))$ to be non-equivalent in $D^{\mathbb{H}}(H^{c_1}_f(k))$. Hence $\C_1$ is represented as $D(\C_1)=\begin{pmatrix}\mi I &0\\0 & -\mi I \end{pmatrix}$. Now the irreducible nonzero mode $D^{\mathbb{H}}(F^{c_1}_f(k))$ is changed to the $\mathbb{C}$ class. Then we do a ${Z}_2$ induction by $\C_2$, where $\C_2$ is represented as $D(\C_2)=\begin{pmatrix}0& \mi I \\ \mi I &0\end{pmatrix}$. After the two steps, a set of {\it irreducible $\mathcal{C}$-zero modes} for the group $F_f^{c_1\times c_2}(k)$ are obtained.

\bibliography{kdotp_Ref}

\end{document}